\documentclass[11pt]{article}
\usepackage[utf8]{inputenc}
\usepackage{amsmath}
\usepackage{amssymb}
\usepackage{amsthm}
\usepackage[dvipsnames]{xcolor}
\usepackage{theoremref}
\usepackage{hyperref}
\usepackage{graphicx}
\usepackage{array}
\usepackage{algorithmic}
\usepackage[Algorithm,ruled]{algorithm}
\usepackage{fullpage}
\usepackage{float}
\usepackage{authblk}

\definecolor{corlinks}{RGB}{0,0,150}
\definecolor{cormenu}{RGB}{200,0,0}
\definecolor{corurl}{RGB}{200,0,0}

\hypersetup{
colorlinks=true,
urlcolor=corlinks,
linkcolor=corlinks,
menucolor=corlinks,
citecolor=corlinks,
pdfborder= 0 0 0
}

\def\colorful{1}

\ifnum\colorful=1

\ifnum\colorful=1
\fi
\ifnum\colorful=0
\fi

\makeatletter
\newtheorem*{rep@theorem}{\rep@title}
\newcommand{\newreptheorem}[2]{
\newenvironment{rep#1}[1]{
 \def\rep@title{#2 \ref{##1}}
 \begin{rep@theorem}}
 {\end{rep@theorem}}}
\makeatother

\newtheorem{theorem}{Theorem}[section]
\newtheorem{lemma}[theorem]{Lemma}
\newtheorem{proposition}[theorem]{Proposition}
\newtheorem{claim}[theorem]{Claim}

\newtheorem{definition}{Definition}[section]
\newtheorem{remark}{Remark}
\newtheorem{corollary}[theorem]{Corollary}
\newreptheorem{theorem}{Theorem}

\newcommand{\poly}{\mathsf{poly}}
\newcommand{\polylog}{\mathsf{polylog}}

\newcommand{\NC}{\mathsf{NC}}
\newcommand{\PVAL}{\mathsf{PVAL}}
\newcommand{\Hamming}{\mathsf{Hwt}}
\newcommand{\HAM}{\mathsf{HAM}}

\newcommand{\IPP}{\mathsf{IPP}}
\newcommand{\MAP}{\mathsf{MAP}}
\newcommand{\TensorSum}{\mathsf{TensorSum}}

\newcommand{\LDE}{\mathsf{LDE}}
\newcommand{\GKR}{\mathsf{GKR}}
\newcommand{\gran}{\mathsf{gran}}

\newcommand{\F}{\mathbb{F}}
\newcommand{\N}{\mathbb{N}}
\newcommand{\R}{\mathbb{R}}

\newcommand{\sU}{\mathcal{U}}
\newcommand{\sD}{\mathcal{D}}
\newcommand{\sF}{\mathcal{F}}
\newcommand{\sO}{\mathcal{O}}
\newcommand{\sA}{\mathcal{A}}
\newcommand{\sB}{\mathcal{B}}
\newcommand{\sP}{\mathcal{P}}

\newcommand{\eps}{\varepsilon}

\makeatletter
\newenvironment{protocol}[1][htb]{
    \renewcommand{\ALG@name}{Protocol}
    \begin{algorithm}[#1]
    }{\end{algorithm}
}
\makeatother

\author[1]{Hugo Aaronson}
\author[1]{Tom Gur}
\author[1]{Ninad Rajgopal}
\author[2]{Ron D. Rothblum}

\affil[1]{University of Cambridge, \texttt{\href{mailto:ha406@cam.ac.uk}{ha406@cam.ac.uk}, \href{mailto:tom.gur@cl.cam.ac.uk}{tom.gur@cl.cam.ac.uk},\href{mailto:nr549@cam.ac.uk}{nr549@cam.ac.uk}}}

\affil[2]{Technion, \texttt{\href{mailto:rothblum@cs.technion.ac.il}{rothblum@cs.technion.ac.il}}}

\date{}

\title{Distribution-Free Proofs of Proximity}

\begin{document}
\maketitle

\begin{abstract}

Motivated by the fact that input distributions are often unknown in advance, distribution-free property testing considers a setting in which the algorithmic task is to accept functions $f : [n] \to \{0,1\}$ having a certain property $\Pi$ and reject functions that are $\varepsilon$-far from $\Pi$, where the distance is measured according to an arbitrary and unknown input distribution $\sD \sim [n]$. As usual in property testing, the tester is required to do so while making only a sublinear number of input queries, but as the distribution is unknown, we also allow a sublinear number of samples from the distribution $\sD$. 

In this work we initiate the study of \emph{distribution-free interactive proofs of proximity} (df-$\IPP$s) in which the distribution-free testing algorithm is assisted by an all powerful but untrusted prover. Our main result is that for any problem $\Pi \in \NC$, any proximity parameter $\eps > 0$, and any (trade-off) parameter $\tau\leq\sqrt{n}$, we construct a df-$\IPP$ for $\Pi$ with respect to $\eps$, that has query and sample complexities $\tau+O(1/\eps)$, and communication complexity $\Tilde{O}(n/\tau + 1/\eps)$. For $\tau$ as above and sufficiently large $\eps$ (namely, when $\eps > \tau/n$), this result matches the parameters of the best-known general purpose $\IPP$s in the standard uniform setting. Moreover, for such $\tau$, its parameters are optimal up to poly-logarithmic factors under reasonable cryptographic assumptions for the same regime of $\eps$ as the uniform setting, i.e., when $\eps \geq 1/\tau$. 

For smaller values of $\eps$ (i.e., when $\eps < \tau/n$), our protocol has communication complexity $\Omega(1/\eps)$, which is worse than the $\Tilde{O}(n/\tau)$ communication complexity of the uniform $\IPP$s (with the same query complexity). With the aim of improving on this gap, we further show that for $\IPP$s over specialised, but large distribution families, such as sufficiently smooth distributions and product distributions, the communication complexity can be reduced to $\Tilde{O}(n/\tau^{1-o(1)})$. In addition, we show that for certain natural families of languages, such as symmetric and (relaxed) self-correctable languages, it is possible to further improve the efficiency of distribution-free $\IPP$s.

\end{abstract}
\thispagestyle{empty}

\newpage
\setcounter{tocdepth}{2}
\tableofcontents
\thispagestyle{empty}

\newpage
\setcounter{page}{1}

\section{Introduction}
Property Testing, initiated in \cite{RS96, GGR98}, is a rich and well-studied research field lying at the heart of many advancements in sublinear algorithms and complexity theory; see \cite{Gol17_book,BY22} for a detailed introduction. Loosely speaking, a testing algorithm for a property $\Pi$ is given oracle access to an input function $f:[n]\rightarrow\{0,1\}$ and should decide whether $f \in \Pi$ using a small \emph{sublinear} number of queries. As  we cannot expect to do so exactly, the tester is required to distinguish between inputs that are in $\Pi$ from those that are $\varepsilon$-far from every function in $\Pi$. Here, distance is typically measured using the relative Hamming distance -- namely, the fraction of outputs of $f$ that need to be changed to reach a member of $\Pi$.

While modeling distance using the relative Hamming distance is natural and convenient, in many settings it may not capture the underlying question (for example, when functions always satisfy a particular format or when some parts in the domain are more important than others). Following the Probably-Approximately-Correct (PAC) learning model, introduced by Valiant in his celebrated work in computational learning theory \cite{Valiant1984}, \textit{distribution-free} algorithms have widely been accepted as a closer abstraction of real-world computational tasks that are required to make decisions based on limited access to the input data. In this spirit, \cite{GGR98} introduced \textit{distribution-free property testing}, where the distance between two functions is with respect to a distribution $\sD$ (over inputs to the function), which is \textit{arbitrary} and \textit{unknown} to the testing algorithm. Since $\sD$ is unknown, in addition to the query oracle to the input $f:[n]\rightarrow \{0,1\}$, the tester can draw independent identically distributed random labelled samples $(i,f(i))$ from a \textit{sample oracle}, where each index $i$ is generated independently from the distribution $\sD$. The tester is required to reject any function that is $\varepsilon$-far\footnote{We say $f:[n]\rightarrow\{0,1\}$ is $\varepsilon$-far from a (non-empty) property $\Pi$ along $\sD$, if for every $f':[n] \rightarrow \{0,1\}$ such that $f' \in \Pi$, it holds that $\mathbb{P}_{i \sim \sD} [f(i) \neq f'(i)]> \varepsilon$.} from $\Pi$ along the unknown distribution $\sD$, and the only access that the tester has to $\sD$ is via the sample oracle. 

The distribution-free model of testing naturally complements the PAC-learning model, and profound bidirectional connections are known between them.\footnote{In particular, in \cite{GGR98}, it is shown that if a class of functions $\mathcal{C}$ has a \textit{proper} PAC-learner using membership queries (where the learner outputs an approximate hypothesis that also belongs to $\mathcal{C}$), then $\mathcal{C}$ has a distribution-free tester that uses roughly the same number of queries and samples as the learner.} Moreover, distribution-free testing is motivated by the fact that it captures the realistic setting where the tester is required to maintain its guarantees despite dealing with data from an unknown environment (i.e., via data samples from some unknown and arbitrary distribution $\sD$). It also deals with situations where not all underlying data points are equally important, e.g., in graphs where certain edges or vertices are more important than others, and one would like to consider distributions that weigh them appropriately.

Following \cite{GGR98}, several distribution-free testing algorithms have been designed for function classes including monotone Boolean functions and low-degree polynomials over finite fields \cite{HalKush}, $k$-juntas \cite{LCSSX18,bshouty19,Belovs19}, conjunctions (monotone or non-monotone) and linear threshold functions \cite{GlasServ, ChenXie}, polynomial threshold functions and decision trees \cite{BFM21}, halfspaces \cite{BFM21,XP22}, and low-degree polynomials on $\R^n$ \cite{dflin,ABFKY23}. Distribution-free testing has also been studied for graph properties including connectivity \cite{HK08}, bipartiteness \cite{goldreich_biparite}, $k$-path and degree regularity \cite{goldreich_2019_graphtest}, as well as for word problems like subsequence-freeness \cite{RR22}.

Despite such strides of progress, our understanding of distribution-free testing is much more limited than that of testing with respect to the uniform distribution. This is due to the multitude of challenges that arise in designing algorithms that need to deal with data samples that can come from any arbitrary distribution, which in turn, makes the model significantly more involved.

This paper aims to bridge the gap between testing over the uniform distribution and distribution-free testing by capitalising on the power of interactive proofs, and delegating the task of handling the challenges imposed by the distribution-free setting to a powerful, but untrusted, prover.

\subsection{Distribution-free Interactive Proofs of Proximity}
\label{sec:dfipps_motivate}
In this work, we initiate the study of \textit{distribution-free interactive proofs of proximity} (distribution-free $\IPP$s), which are distribution-free testers that are augmented with the help of a prover. In the rest of this paper, for convenience, rather than thinking of the input as a function, we view it as a string $x \in \{0,1\}^{n}$ (which can be similarly be viewed as a truth table of a function $f_x : [n] \to \{0,1\}$). Correspondingly, we view a property $\Pi$ of functions as a language $L$ over strings (which may be viewed as truth tables of the functions in $\Pi$).

Thus, distribution-free $\IPP$s are protocols where a \textit{sublinear} time, randomised algorithm, called the verifier, interacts with an untrusted prover to decide whether the given input $x \in \{0,1\}^n$ belongs to the language $L$ or is far from such, where distance is measured with respect to a fixed, but unknown distribution $\sD$ over $[n]$. The verifier is given access to the input $x$ through a query oracle, as well as a sample oracle with respect to $\sD$, while the prover can look at the input entirely. We assume that the prover does not know the queries that the verifier makes to either of its oracles.

We require that for any $x \in L$, there exists an honest prover that interacts with the verifier and convinces it to accept with high probability,
while when $x$ is $\varepsilon$-far from $L$ with respect to the distribution $\sD$, no cheating prover, even computationally unbounded, will make the verifier accept, except with low probability.
Further, we require the distribution-free $\IPP$ to meet these requirements, with respect to the underlying (and unknown) distribution $\sD$ from which the oracle draws samples.

In this setting, the verifier's \textit{query complexity} and \textit{sample complexity}, the number of bits exchanged in the protocol, i.e., the \textit{communication complexity}, and the verifier's running time should all be sublinear in input length. Other complexity parameters of interest are the number of rounds of interaction, and the (honest) prover's running time.

Distribution-free $\IPP$s capture the distribution-free property testing analogue of interactive proofs (for more information, see Section \ref{sec:rel_work}). As such, similar to uniform $\IPP$s, distribution-free $\IPP$s can be alternatively viewed as proof systems where the bounded verifier need only be convinced of the fact that the input is close to the language, by interacting with a more powerful prover. One of the main goals of distribution-free $\IPP$s is to overcome the inherent limitations of distribution-free testing algorithms by showing that for certain properties, verifying proximity over arbitrary distributions is considerably faster with a prover than actually testing it. In particular, we want to design distribution-free $\IPP$s (with sublinear query complexity) for rich families of properties that have no known distribution-free testers.

Of close relevance are the well-studied notion of $\IPP$s over the uniform distribution, which we refer to in this work as Uniform $\IPP$s, that were introduced in \cite{EKR04,RVW} (and are trivially generalised by distribution-free $\IPP$s). Showcasing the power of interaction, \cite{RVW} constructed highly non-trivial uniform $\IPP$s for every language that can be decided in bounded depth (e.g., $\NC$), which was recently made near-optimal by \cite{RR20_batch_polylog} (see \cite{KR15} for the conditional matching lower bound), and strengthened to encompass also bounded space languages \cite{RRR21}. 

Motivated intrinsically and by natural applications to \textit{delegation of computation}, the study of uniform $\IPP$s has drawn much recent attention on its own right \cite{RVW,GR18,KR15,RRR21,GG21_universal,GR22_sample}. Moreover, their study has led to interesting models and applications of sublinear time verification, including non-interactive proofs of proximity (or \textsf{MAP}s) \cite{GR18} (a related model was studied concurrently and independently by \cite{FGL14}), arguments of proximity \cite{KR15}, testing properties of distributions \cite{CG18,HR22}, interactive oracle proofs of proximity \cite{RRR21,BBHR18,RR20,BLNR22}, verifying machine learning tasks \cite{GRSY21}, batch verification for $\mathsf{UP}$ \cite{RRR18,RR20_batch_polylog}, as well as variants involving zero-knowledge \cite{BRV18} and quantum computation \cite{DGMT22}. 

\subsection{Our Results}
\label{sec:results}
Our main contribution is constructing distribution-free $\IPP$s for any language in $\NC$, which for any query vs communication trade-off parameter $\tau \leq \sqrt{n}$, matches the complexity of the best known $\IPP$s for most settings of the proximity parameter $\varepsilon$ -- specifically, when $\varepsilon \geq \tau/n$. We further improve the efficiency of distribution-free $\IPP$s for general $\varepsilon$ (i.e., when $\varepsilon<\tau/n$), under specific distribution families such as ``smooth" and ``learnable" distributions, which are defined below.

In addition, for certain families of languages, such as symmetric and relaxed self-correctable languages, we construct distribution-free $\IPP$s that improve on our general-purpose distribution-free $\IPP$s, then use them to provide separation results that provide further insight into the distribution-free $\IPP$ model.

We elaborate on these results next.

\subsubsection{Distribution-free $\IPP$s for $\NC$}
Our first main result is a sublinear distribution-free $\IPP$ for any language computable by low-depth circuits. In more detail, let (logspace-uniform) $\NC$ be the set of languages computable by (logspace-uniform) Boolean circuits of polynomial size and poly-logarithmic depth. We show that every language in $\NC$ has a distribution-free $\IPP$ with sublinear complexity measures, for almost all values of the proximity parameter $\varepsilon$. We emphasize that this is in stark contrast to distribution-free testers, which are only known for a handful of languages based on their combinatorial or algebraic structure. Indeed, the following theorem shows that distribution-free $\IPP$s capture a much richer class of languages that need not have such special structural properties.

\begin{theorem}[\textbf{Distribution-Free $\IPP$ for $\NC$}]
    \label{thm:informal_dfipp_nc}
    For every language $L$ in logspace-uniform $\NC$ and every trade-off parameter $\tau=\tau(n) \leq \sqrt{n}$, there exists a distribution-free $\IPP$ for $L$ with proximity parameter $\varepsilon\geq \Omega\left(\frac{\log^3 (n)}{n}\right)$, query complexity $\tau+O\left(\frac{1}{\eps}\right)$, sample complexity $\tau+O\left(\frac{1}{\eps}\right)$ and communication complexity $\tilde{O}\left(\frac{n}{\tau}+\frac{1}{\eps}\right)$. 
    
    Moreover, the verifier runs in time $\tilde{O}\left(\frac{n}{\tau} +\frac{1}{\eps}\right)$, the prover runs in time $\poly(n)$ and the round complexity is $\polylog(n)$.
\end{theorem}

Here, $\tau$ denotes the parameter that trades-off between the query and communication complexities of the distribution-free $\IPP$. Note that, for the above values of $\tau$, our distribution-free $\IPP$ has sublinear query and communication complexity even for very small values of the proximity parameter $\varepsilon$ of the form $1/n^{1-\delta}$, where $\delta > 0$. An interesting instantiation of our result is obtained by setting $\tau$ to $\sqrt{n}$, and thus, for every $\eps \geq 1/\sqrt{n}$, the query complexity and sample complexities are $O(\sqrt{n})$, while the communication complexity and verifier running times are both $\Tilde{O}(\sqrt{n})$. 

It is worth noting that, for every $\eps\geq \frac{1}{\tau}$ (and $\tau \leq \sqrt{n}$), this result is conditionally optimal up to poly-logarithmic factors, since \cite{KR15} show a lower bound of $\Omega(n)$ on the product of the query and communication complexities of a uniform $\IPP$ for a language in $\NC^1$, under a strong, but reasonable, cryptographic assumption. Furthermore, for any $\varepsilon$, the query complexity of $\Omega(1/\varepsilon)$ is necessary for any $\IPP$ over non-degenerate languages, even over the uniform distribution (see \cite[Remark 1.2]{RVW}).

\begin{remark}
While Theorem \ref{thm:informal_dfipp_nc} refers to distribution-free $\IPP$s over $\NC$ languages, the theorem is more general (see Theorem \ref{thm:simplerncdfippsizedepth}). In particular, it also yields distribution-free $\IPP$s with sublinear query and communication complexities for languages computable by circuits of sub-exponential size and bounded polynomial depth.

Likewise, in a similar fashion to the known literature on uniform $\IPP$s, we can combine our techniques directly with \cite{RRR21} to get a constant-round distribution-free $\IPP$ for any language that is computable in $\poly(n)$ time and bounded polynomial space.
\end{remark}

\paragraph*{Comparison to Uniform $\IPP$s for $\NC$ \cite{RVW,RR20_batch_polylog}:} For any language in $\NC$, Rothblum, Vadhan and Wigderson \cite{RVW} construct a uniform $\IPP$ for any $\tau=\tau(n)$ and proximity parameter $\eps > 0$, with query complexity $\tau + O(1/\eps)^{1+o(1)}$ and communication complexity $\frac{n}{\tau^{1-o(1)}}$. Rothblum and Rothblum \cite{RR20_batch_polylog} improve on this, by reducing the communication complexity to $\frac{n}{\tau} \cdot \polylog(n)$. In particular, the latter obtains an optimal trade-off, up to poly-logarithmic factors, between the query and communication complexities of a uniform $\IPP$ (conditionally, from \cite{KR15}), for every value of $\tau$ and $\eps \geq 1/\tau$. While these results are stated in \cite{RVW,RR20_batch_polylog} by implicitly setting $\tau = O(1/\eps)$, for any given $\eps$, this $\IPP$ formulation parameterised by $\tau$ is obtained by inspection (see also \cite[Theorem 6.3]{GG21_universal}). For comparison, in this setting, our distribution-free $\IPP$ has the same query (and sample) complexity, while the communication complexity and verifier running times are both $\Tilde{O}(\eps\cdot n + 1/\eps)$.\footnote{In fact, we prove that for every value of the parameter $\tau$ and $\eps$, the distribution-free $\IPP$ from Theorem \ref{thm:informal_dfipp_nc} has communication complexity $\Tilde{O}(\tau+n/\tau+1/\eps)$; thus, setting $\tau = O(1/\eps)$ suffices. An additional point to note is that when $\tau > \sqrt{n}$, the $\IPP$ always has worse communication complexity than its uniform counterpart irrespective of the value of $\eps$, and further, never meets the optimal \cite{KR15} lower bound. As such, we only consider $\tau \leq \sqrt{n}$ as a more interesting regime of study.}

Theorem \ref{thm:informal_dfipp_nc} gives a construction of a \emph{distribution-free} $\IPP$ for any $\NC$ language that matches the query and communication complexities of the uniform $\IPP$ by \cite{RR20_batch_polylog}, when $\eps \geq \tau/n$. Moreover, this obtains the (conditionally) optimal trade-offs between query and communication complexities in the \textit{same regime} of $\eps$, but when $\tau \leq \sqrt{n}$. Indeed, when $\eps \geq 1/\tau$, the product of the query and communication complexities of the distribution-free $\IPP$ from Theorem \ref{thm:informal_dfipp_nc} is $\Tilde{O}(n + \tau^2)$. Our protocol builds on \cite{RVW}, introducing new ideas that allow us to construct $\IPP$s in the more involved distribution-free setting. 

Finally, when the proximity parameter $\varepsilon$ is very small, Theorem \ref{thm:informal_dfipp_nc} suffers a blow-up in the communication complexity compared to the uniform $\IPP$s of \cite{RVW,RR20_batch_polylog}. In more detail, when $\varepsilon \ll \tau/n$, the communication complexity in our distribution-free $\IPP$ is $\Tilde{\Omega}\left(\frac{1}{\eps}\right)$, whereas the communication complexity achieved by the uniform $\IPP$s is $\Tilde{O}\left(\frac{n}{\tau}\right)$ (the query complexity roughly remains the same across all three cases). Thus, our distribution-free $\IPP$ has communication complexity at least $\Omega(n/\tau)$ for every value of $\varepsilon$, whereas the communication complexity of the uniform $\IPP$s is much lower when $\varepsilon \ll \tau/n$.

\subsubsection{$\IPP$s for $\NC$: The case of small $\varepsilon$}
Following the discussion in the last section, we aim to construct distribution-free $\IPP$s that achieve query and communication complexities that match the state-of-the-art uniform $\IPP$ for every value of $\varepsilon$. While we unable to do so in the most general case, we construct such $\IPP$s over \textit{specific families of distributions}, which match the complexities of \cite{RVW} and, in turn, differ from the complexities of \cite{RR20_batch_polylog} only by a factor of $n^{o(1)}$. For these $\IPP$s, while the underlying distribution is still unknown, it is guaranteed to come from the specific family of distributions under consideration.

To describe our results, it will be convenient throughout this section to identify $[n]$ with the elements of an $m$-dimensional tensor of size $k \in \N$ in each dimension, such that $k^m = n$. In such a case, we refer to $[n]$ as $[k]^m$ (by fixing some canonical bijection between them).

\paragraph*{$\rho$-Dispersed Distributions:}  Intuitively speaking, $\rho$-dispersed distributions capture the sense that for a smooth distribution over $[k]^m$, along any dimension, its probability mass on any element in $[k]^m$ is not much larger than the average of the probability masses of its neighbours. $\rho$-dispersed distributions relax this requirement by having the probability mass on any element bounded by $\rho$ times the expected mass on any of its neighbours.\footnote{For example, the uniform distribution is the only $1$-dispersed distribution, i.e., a maximally smooth distribution in this sense. On the other hand, every distribution over $[k]^m$ is trivially a $k$-dispersed distribution.} We refer to Section \ref{sec:def_dispersed} for the formal definition of $\rho$-dispersed distributions and classification of well-studied distributions.

We show that for distributions that are reasonably smooth in this sense, i.e. for $\rho$-dispersed distributions for $\rho \leq k^{o(1)}$, we obtain  $\IPP$s for $\NC$ over such distributions for every $\tau=\tau(n)<n$ and $\varepsilon > 0$, with query complexity $O(\tau + 1/\eps)^{1+o(1)}$, and communication complexity of $\Tilde{O}\left(\frac{n}{\tau}\cdot \tau^{o(1)}\right)$, thus matching the bounds obtained by \cite{RVW}. It is worth noting that $k^{o(1)}$-dispersed distributions are still quite general, e.g. any distribution where the probability mass on any element in $[k]^m$ is in the range $\left[\frac{1}{an}, \frac{a}{n}\right]$, for some $a \leq k^{o(1)}$ is $k^{o(1)}$-dispersed. 

\begin{theorem}[\textbf{$\IPP$ for $\NC$ over $\rho$-dispersed distributions}]
    \label{thm:informal_ipp_dispersed}
    For every language in logspace-uniform $\NC$, every $m,n,k \in \N$ such that $m = \log_k(n)$ (i.e., $k^m = n$) and $\rho \in \R$ such that $\rho \leq k^{o(1)}$, for every proximity parameter $\eps > 0$ and trade-off parameter $\tau>0$,  there exists an $\IPP$ over $\rho$-Dispersed distributions over $[k]^m$ with query and sample complexities $O(\tau+ 1/\eps)^{1+o(1)}$ and communication complexity $\Tilde{O} \left( \frac{n}{\tau^{1-o(1)}}\right)$. 
    
    Moreover, the verifier runs in time $n^{o(1)} \cdot \left(\tau + \frac{n}{\tau} +\frac{1}{\eps}\right)$, the prover runs in time $\poly(n)$ and the round complexity is $\polylog(n)$.
\end{theorem}

Theorem \ref{thm:informal_ipp_dispersed} also holds generally over $\rho$-dispersed distributions, for any $\rho$ (see Theorem \ref{thm:dispersed_ipp_nc}). The query complexity increases with $\rho$, while the communication complexity is \textit{independent of $\rho$}. Theorem \ref{thm:informal_ipp_dispersed} builds on the ideas used for the distribution-free $\IPP$ from Theorem \ref{thm:informal_dfipp_nc} while incorporating new technical insights into the analysis by \cite{RVW} to generalise over $\rho$-dispersed distributions. We leave the task of obtaining $\IPP$s over $\rho$-dispersed distributions that match \cite{RR20_batch_polylog} as future work.

\paragraph{Product Distributions in the White-Box model:}
Note that in the $\IPP$s of Theorems~\ref{thm:informal_dfipp_nc} and~\ref{thm:informal_ipp_dispersed}, the verifier does not learn the underlying distribution $\sD$. Hence, we ask the following question: if we could gain more information about $\sD$, or further, learn a reasonably good approximation for $\sD$, can we improve the query complexity of the $\IPP$s, over general values of $\varepsilon$? We answer this question in the affirmative for product distributions.

We consider the \textit{white-box model} for distribution-free $\IPP$s, where the verifier receives a succinct description of the unknown distribution $\sD$ over $[k]^m$ via a \textit{polynomial-sized} sampling circuit $C$, in addition to query access to the input string. It is worth noting that, for white-box $\IPP$s, the sample complexity is irrelevant since the verifier has a succinct description of the entire distribution. Thus, the main complexity parameters here are the query complexity, communication complexity, and the verifier running time. 

While white-box models have been widely studied in the setting of zero-knowledge proofs \cite{SV97,vadhan_thesis,vadhan06_czk} and in distribution testing (see survey by \cite{GV11_survey}), we use this model to construct $\IPP$s for languages in $\NC$ over a generalised family of product distributions over $[k]^m$, to get improved complexities for general values of $\eps$, compared to the distribution-free $\IPP$ from Theorem~\ref{thm:informal_dfipp_nc}. We call this family as $m$-\textit{product distributions}, and denote any such distribution $\sD$ as $\sD = \sD_1 \times \dots \sD_m$, where each $\sD_j$ is supported on $[k]$ and is independent of any other coordinate distributions. In particular, $\sD(i_1, \dots, i_m)$ is defined as $\prod_{j=1}^m \sD_j(i_j)$ (see Definition~\ref{def:whitebox_df_ipp} for more details about white-box $\IPP$s  and Definition~\ref{def:prod_dist} for product distributions).

\begin{theorem}[\textbf{$\IPP$s for $\NC$ over $m$-product distributions}]
    \label{thm:informal_product_dfipp}
    For every language in logspace-uniform $\NC$, every $\tau=\tau(n)$, $\eps >0$, and $m,n,k \in \N$ such that $m \leq \log(n)$ and $k^m = n$, there exists a white-box $\IPP$ for $L$ over $m$-product distributions over $[k]^m$. The $\IPP$ has query complexity $O(\tau+ 1/\eps)^{1+o(1)}$ and communication complexity $\left(\frac{n}{\tau^{1-o(1)}} \cdot k+k^{2}\right) \cdot \polylog(n)$. Moreover, the verifier runs in time $n^{o(1)}\left(\frac{n}{\tau}\cdot k + \tau + k^2 + \frac{1}{\eps}\right)$ and the round complexity is $\polylog(n)$. 
\end{theorem}

Similar to the previous results, a general version of this $\IPP$ can be found in Theorem \ref{thm:dfipp_product_whitebox}. In particular, when $m$ is large enough (like $m = \log(n)$), then the query and communication complexity trade-off, as well as the verifier running time of the $\IPP$ from Theorem~\ref{thm:informal_product_dfipp} match that of the uniform $\IPP$ from \cite{RVW}, while working in this setting.\footnote{A subtle point here is that while Theorem \ref{thm:informal_product_dfipp} is over product distributions over $[k]^m$, when $m = 2$ (or a small constant), we get sublinear complexities only by considering distributions over biased matrices $[k_1] \times [k_2]$.} 
Theorem \ref{thm:informal_product_dfipp} builds on the framework of Theorem~\ref{thm:informal_dfipp_nc}, and uses several new ideas in the construction of the $\IPP$, as well as its analysis, to improve the complexity. Crucially, it uses that any product distribution has a succinct description to be able to \textit{learn} it in the white-box-setting. 

It is worth stressing that the $\IPP$s from Theorems~\ref{thm:informal_ipp_dispersed} and~\ref{thm:informal_product_dfipp} are incomparable. Indeed, there exist $m$-product distributions $\sD = \sD_1 \times \dots \times \sD_m$ that are poorly dispersed, for eg., $\sD$ is no longer smooth when some $\sD_j$ has a large probability mass over just one element (one row or more generally, a few rows). For such distributions, the $\IPP$ from Theorem~\ref{thm:informal_product_dfipp} provides a much better query and communication trade-off than the $\IPP$ from Theorem \ref{thm:informal_ipp_dispersed}, which is a more general result for smooth distributions.

\subsubsection{On the power of distribution-free $\IPP$s}
Recall that Theorems~\ref{thm:informal_ipp_dispersed} and~\ref{thm:informal_product_dfipp} improve the query and communication complexity trade-off of our general distribution-free $\IPP$ in Theorem~\ref{thm:informal_dfipp_nc}, by considering special families of distributions to design the $\IPP$s over. A natural direction that complements this approach is to ask whether we can use additional information about the \textit{language} $L$ instead, to construct super-efficient distribution-free $\IPP$s. 

In turn, we study distribution-free $\IPP$s for specific problems of interest. On one hand, for certain problems we can hope to improve the various associated complexity parameters over our general distribution-free $\IPP$ by capitalising on the structure of the language. On the other hand, this allows us to obtain complexity-theoretic separations between the power of standard, non-interactive, and interactive distribution-free testers.

\paragraph{Symmetric languages.}
We study the power of distribution-free testers and $\IPP$s for symmetric languages, which are languages that are invariant under permutations. We show that there exist symmetric languages that are hard for distribution-free testers, yet, given interaction with a prover, the symmetrical structure can be leveraged to obtain exponentially faster distribution-free $\IPP$s.

\begin{theorem}[\textbf{Distribution-free $\IPP$s for symmetric languages}]
    \label{thm:informal_ham_sep}
The following statements hold.
\begin{enumerate}
    \item Let $L$ be a symmetric language. Then, there exist a distribution-free $\IPP$ for $L$ with sample complexity $O(1/\varepsilon)$, communication complexity $O(\log^{2}(n)/\varepsilon)$ and $O(\log(n)/\varepsilon)$ round complexity.
    \item There exists a symmetric language $L'$ for every $\varepsilon>0$ such that any distribution-free property tester for $L'$ requires $\Omega(n^{1/3-0.0005})$ queries and labeled samples from the input.
\end{enumerate}

\end{theorem}

\paragraph{(Relaxed) self-correctable languages.}
Next, we show that for languages that admit self-correctability, we can transform any $\IPP$ into a distribution-free $\IPP$ at a negligible cost. In fact, we can deal with a far more general class of languages; namely, languages that are \emph{relaxed locally correctable} \cite{BGHSV04,GRR20_rlcc}. Loosely speaking, these are languages that admit a correcting algorithm that is required to correct the symbol at every location of the codeword, by reading a small number of locations in it, but is allowed to abort if noticing that the given word is corrupted. This family of languages is of central importance in the interactive proofs and probabilistically checkable proofs literature, and in particular, it captures languages of low-degree polynomials, holographic $\IPP$s, and various relaxed locally correctable and decodable languages that were used to prove complexity-theoretic separations (cf. \cite{gur2017locally}).

\begin{proposition}[\textbf{Generic Transformations for $\IPP$s for RLCCs}]
\label{thm:informal_locCorr}
For any subset $L$ of a binary RLCC, $C \subseteq \{0,1\}^n$, if $L$ has an $\IPP$ over the uniform distribution with query complexity $q$ and communication complexity $c$ for proximity $\varepsilon > 0$, then there exists a distribution-free $\IPP$ for $L$ with the same round complexity, communication complexity and query complexity $q+O(\frac{t}{\varepsilon})$, where $t$ is the query complexity of the corrector of $C$.
\end{proposition}

A detailed statement can be found in Theorem \ref{thm:rlcc_dfipp}. As a corollary of Proposition~\ref{thm:informal_locCorr}, we are able to lift complexity-theoretic results concerning uniform $\IPP$s to the setting of distribution-free $\IPP$s. In particular, we obtain strong separations between the power of distribution-free testers, distribution-free non-interactive proofs of proximity ($\MAP$s), and distribution-free $\IPP$s.

\begin{corollary}[\textbf{Complexity separations}]
\label{thm:informal_tensorsum_gaps}
    There exists a language $L$ such the following hold true.
    \begin{enumerate}
        \item \textsf{Property Testing}: The query complexity of distribution-free testing $L$ (without a proof) is $\Theta(n^{0.999\pm o(1)})$.
        \item $\MAP$: $L$ has a distribution-free $\MAP$ with query and communication complexities $\Theta(n^{0.499\pm o(1)})$. Moreover, for every $p \geq 1$, the distribution-free $\MAP$ query complexity of $L$ with respect to proofs of length p is $\Theta\big(\frac{n^{0.999\pm o(1)}}{p}\big)$.
        \item $\IPP$: $L$ has a distribution-free $\IPP$ with query and communication complexities $\polylog(n)$.
    \end{enumerate}
\end{corollary}
See Theorem \ref{thm:formal_tensorsum_gaps} for a more detailed statement. Complementing this Corollary, we prove the existence of languages that can be tested under the uniform distribution with low query complexity (and thus, have a uniform $\IPP$ with low query complexity and no communication), but for which distribution-free $\IPP$s require large query complexity or large communication complexity. This illustrates the difficulty of constructing distribution-free $\IPP$s vs. standard uniform $\IPP$s.

\begin{proposition}[Distribution-free $\IPP$s vs. uniform testing]
\label{prop:seps}
    The following hold true:
    \begin{enumerate}
        \item There exists $\varepsilon>0$ and a language $L$ such that $L$ has a property tester over the uniform distribution with query complexity $O(1/\varepsilon)$ for proximity parameter $\varepsilon$. However, for any distribution-free $\MAP$ for $L$ with proximity parameter $\eps$, query complexity $q$, and proof length $p$, $\max(q,p)=\Omega(\varepsilon\cdot n)$. 
        
        \item Assuming the existence of exponentially hard pseudo-random generators, there exists $\varepsilon>0$ such that for all $q=q(n)\leq n$, there exists a language $L$, such that for any distribution-free $\IPP$ for $L$ with proximity parameter $\eps$, communication complexity $c$, and query complexity $q$, $\max(c,q)=\Omega(\sqrt{\varepsilon\cdot n})$. 
        However, $L$ has a uniform property tester with query complexity $O(1/\varepsilon)$ for proximity parameter $\eps$.
        \end{enumerate}
\end{proposition}

See Section \ref{sec:dfUseparation} for more details. Table \ref{tbl:results} provides a comparison of some of these results with related testing models. It is an interesting open direction to exhibit distribution-free $\IPP$s that improve on the query complexity lower bounds known for distribution-testing functional properties like monotonicity \cite{HalKush}, monotone conjunctions \cite{ChenXie}, or $k$-juntas \cite{Junta}.

\begin{table}[!h]
\begin{tabular}{ | m{5em} | m{2.75cm}| m{1.75cm} | m{2.25cm}| m{5.5cm} | } 
  \hline
  & Property \quad Testing & $\IPP$ &  DF-Property Testing & DF-IPP\\
  \hline
  Languages in $\NC$ & $\Omega(n)$ (e.g., low-degree univariate polynomial) & $\tilde{O}(\sqrt{n})$ \newline \cite{RVW,RR20_batch_polylog} & $\Omega(n)$ similarly & $\tilde{O}(\sqrt{n})$ (arbitrary distributions, for $\varepsilon \geq 1/\sqrt{n}$); see Theorem \ref{thm:informal_dfipp_nc}\newline $n^{1/2+o(1)}$ (smooth distributions); see Theorem \ref{thm:informal_ipp_dispersed} \newline $n^{1/2+o(1)}$  (product distributions); see Theorem \ref{thm:informal_product_dfipp}\\ 
  \hline
  $\TensorSum$ & $\Omega(n^{0.99+o(1)})$ \newline \cite{GR18} & $\polylog(n)$  \newline \cite{GR18} & $\Omega(n^{0.99+o(1)})$\newline Trivially, from \cite{GR18} & $\polylog(n)$; see Corollary \ref{thm:informal_tensorsum_gaps} \\ 
  \hline
  Symmetric Properties & $\Theta(1)$ ($\varepsilon=O(1)$) Folklore \newline & $\polylog(n)$ \cite{RVW} & $\Omega(n^{\frac{1}{3}})$ \newline Theorem \ref{thm:informal_ham_sep}& $\polylog(n)$; see Theorem \ref{thm:informal_ham_sep} \\ 
  \hline
\end{tabular}
\caption{
This is a table of our main results ($\TensorSum$ is defined in Definition \ref{def:tensorsum}). The complexities shown here are those that minimise the sum of the query and communication complexity. Note that while the uniform property tester for symmetric properties is more efficient than the corresponding uniform $\IPP$, this only holds for restricted (constant) values of $\eps$.}
\label{tbl:results}
\end{table}

\subsection{Technical Overview}
\label{sec:techniques}
In this technical overview, we highlight the proofs of Theorems \ref{thm:informal_dfipp_nc}, \ref{thm:informal_ipp_dispersed}, and \ref{thm:informal_product_dfipp}. The general strategy for proving these theorems builds on the Uniform $\IPP$s for $\NC$ from \cite{RVW,RR20_batch_polylog}. However, the setting of distribution-free testing is more involved, and below, we highlight the key challenges encountered in this setting, and our ideas to overcome them. Our distribution-free $\IPP$s are constructed through an interplay of various techniques and tools from interactive proofs, property testing, and distribution testing; see Section \ref{sec:sym}, for further details on the proof strategy of Theorem \ref{thm:informal_ham_sep}. 

Note that, for convenience, we show the construction of the distribution-free $\IPP$ from Theorem \ref{thm:informal_dfipp_nc} in the setting of $\tau = O(1/\eps)$, for any proximity parameter $\varepsilon$, obtaining query complexity $O(1/\eps)$ and communication complexity $\tilde{O}(\varepsilon\cdot n + 1/\eps)$. This can be shown to be equivalent to the statement of Theorem \ref{thm:informal_dfipp_nc} that is parameterised by $\tau$; for more details see Section \ref{sec:dfipp_low-depth}. Similarly, the $\IPP$s for our other results are parameterised in terms of the proximity parameter $\eps$.

\subsubsection{Proof outline of Theorem \ref{thm:informal_dfipp_nc}}
\label{sec:tech_main_dfipp}
The \cite{RVW} protocol (as well as the follow-up work \cite{RR20_batch_polylog}) is centered around a parameterised problem called $\PVAL$. Loosely speaking, the $\PVAL$ language contains all strings, whose encoding under a specific code, called the low degree extension, is equal to given values when projected on to the given coordinates. More precisely, the $\PVAL$ problem is parameterised by a (sufficiently large) finite field $\F$, integers $k,m,n$ such that $k,m < \vert \F \vert$ and $k^m = n$, a set of vectors $J = (j_1, \dots, j_t) \subset \F^m$ of size $t$ and a $t$-length vector $\Vec{v} \in \F^t$. An input $X \in \F^{k^m}$ is in $\PVAL(J,\Vec{v})$ if it holds that $P_X(j_i) = v_i$, for every $i \in [t]$, where $P_X:\F^m \rightarrow \F$ is the $m$-variate low-degree extension ($\LDE$) of $X$.\footnote{Recall that the $m$-variate $\LDE$ $P_X$ is the unique polynomial with individual degree $k-1$ such that $P_X$ agrees with $X$ on $[k]^m$, where we identify $[k]$ with a subset of field elements in some canonical way.}

\paragraph*{The interactive reduction from $\NC$ to $\PVAL$.} Let $L$ be any language in $\NC$ and let $\eps > 0$ be the input proximity parameter. Let $X \in \{0,1\}^n$ be the input to $L$ and $\sD$ be the unknown underlying distribution over which the verifier can access $X$ through a sample oracle. The first step in \cite{RVW} is to show an interactive reduction $\Pi_{\NC}$ from $L$ to (a parameterisation of) $\PVAL$, where the verifier \emph{does not access} the input $X \in \{0,1\}^n$.\footnote{Technically, an interactive proof is specified by a verifier and an honest prover. However, for the sake of exposition we refer to them both together as $\Pi_{\NC}$ in this section.} 

In more detail, let $B_\sD(X)$ (respectively $B_\sU(X)$) be the set of binary strings that are at a distance at most $\eps$ along the distribution $\sD$ (respectively the uniform distribution $\sU$) from $X$. In \cite{RVW}, the verifier in $\Pi_\NC$ generates parameters $(\F,k,m,J,\Vec{v})$ for $\PVAL$, where $J$ is a set of $t$ points in $\F^m$, such that the following hold when $t$ is sufficiently large.
\begin{itemize}
    \item If $X \in L$, then $X \in \PVAL(J,\Vec{v})$.
    \item If $X$ is $\varepsilon$-far from $L$ along $\sU$ then, with high probability over the verifier's randomness, $B_\sU(X)$ and $\PVAL(J,\Vec{v})$ are disjoint. In other words, with high probability, $X$ is $\varepsilon$-far from $\PVAL(J,\Vec{v})$ along $\sU$.
\end{itemize}
Furthermore, the points $J$ output by the reduction $\Pi_\NC$ are \textit{distributed uniformly at random} in $\left( \F^m \right)^t$. Crucially, \cite{RVW} show that the guarantees over the outputs of this reduction \textit{only hold} when $t = O(\log (\vert B_\sU(X) \vert)$ many points are picked in $J$.\footnote{\label{fn:gkr} $\Pi_\NC$ runs $t$ parallel copies of the interactive reduction from $L$ to $\PVAL$ over a single point by \cite{GKR15}, with the guarantee that if the input $X \notin L$, the probability that $X$ is also in $\PVAL$ over $t$ points, is at most $2^{-t}$. Now, if $X$ were instead $\varepsilon$-far from $L$, then a union bound over all the points in $B_\sU(X)$ ensures a small probability for the event that there exists a point in $B_\sU(X)$ that is also in $\PVAL$ over $t$ points. We refer to Section \ref{sec:intuition_reduction_hybrid} for more details.}

Since the size of the set $B_\sU (X)$ is $\binom{n}{\varepsilon n} \leq O(2^{\varepsilon n \log(n)})$, following from the earlier discussion, by setting $t = O(\log (\vert B_\sU(X) \vert) = \tilde{O}(\varepsilon n)$, we ensure that the guarantees of $\Pi_\NC$ hold. An immediate attempt would be try to extend this analysis verbatim to distribution-free testing, by setting $t$ to $O(\log (\vert B_\sD(X) \vert))$ instead, and thus having $\Pi_\NC$ guarantee that $X$ is $\varepsilon$-far from $\PVAL(J,\Vec{v})$ along the distribution $\sD$, for soundness. However, for an arbitrary unknown distribution $\sD$, the size of $B_{\sD}(X)$ can be prohibitively large. For example, when $\sD$ is supported over the first $\log(n)$ indices, for any value of $\varepsilon$, the size of $B_\sD (X)$ blows up to at least $2^{n-\log(n)}$! Thus, for our choice of $t$, we already lose the sublinear time verification and communication complexity, and it is unclear if this reduction can achieve such soundness guarantees for $\PVAL$. 

\begin{figure}
    \centering
    \begin{minipage}{0.45\textwidth}
        \centering
        \includegraphics[width=0.7\textwidth]{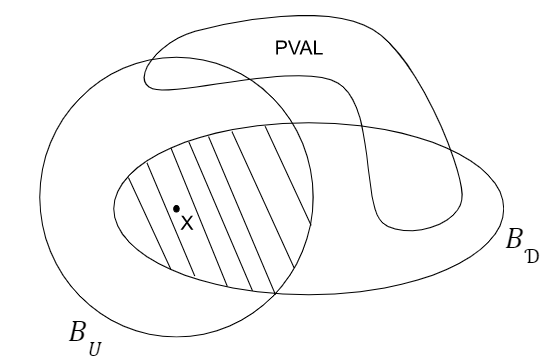}
        \caption{The shaded region ($B_{\mathcal{U}}(X) \cap B_{\mathcal{D}}(X)$) consists of the set of points in $\{0,1\}^n$ that are $\varepsilon$-close to $X$ with respect to both $\sD$ and $\sU$. The soundness promise of the interactive reduction $\Pi'$ ensures that any string in $\PVAL(J,\Vec{v})$ is present in at most one of $B_{\mathcal{U}}(X)$ or $B_{\sD}(X)$, but not in both (shaded region) (with high probability).}
        \label{fig:PVAL}
    \end{minipage}\hfill
    \begin{minipage}{0.45\textwidth}
        \centering
        \includegraphics[width=0.5\textwidth]{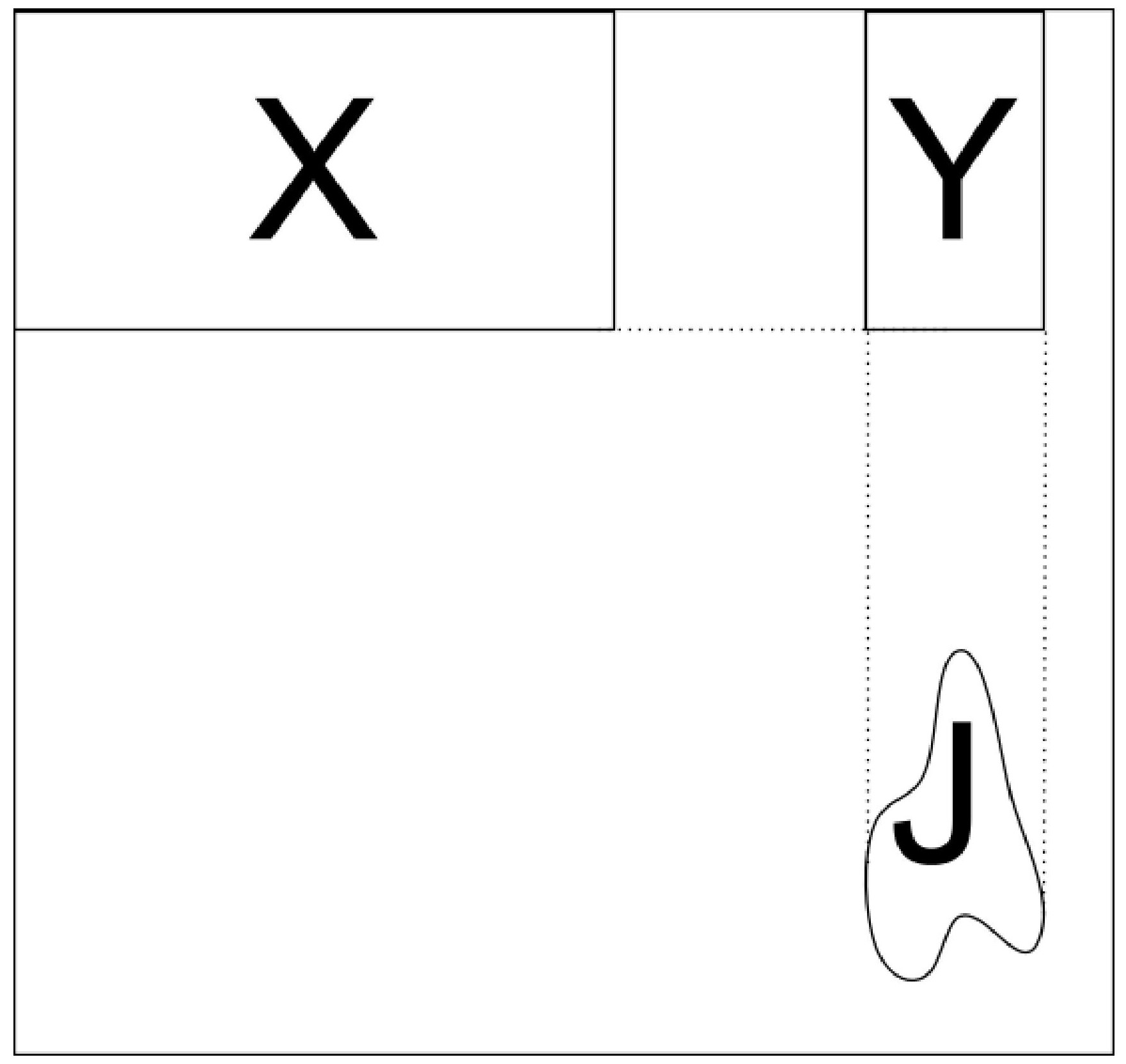}
        \caption{In the uniform $\IPP$ for $\PVAL$, the prover sends the $(m-1)$-variate $\LDE$ of each row of X evaluated on $J_2$ (column indices of $J$), in the form of the purported matrix $Y' \in \F^{k \times t}$. However, to ensure consistency of $Y'$ with respect to $\PVAL(J,\Vec{v})$, for any $j = (a,b) \in J$, the univariate $\LDE$ of the $b^{\text{th}}$-column of $Y'$ evaluated on $a$ is required to be equal to $\Vec{v}[j]$.}
        \label{fig:polyfold_intro}
    \end{minipage}
\end{figure}

\paragraph*{Uniform $\IPP$ for $\PVAL$ is also ``complete" for distribution-free $\IPP$s for $\NC$.} Our key idea for constructing the distribution-free $\IPP$ for $L$, is in fact, an interactive reduction $\Pi'$ to constructing a \emph{uniform} $\IPP$ for $\PVAL$ (with a different parameterisation for $\PVAL$ than that obtained by $\Pi_\NC$). Theorem \ref{thm:informal_dfipp_nc} follows by using the ready-made uniform $\IPP$ for $\PVAL$ by \cite{RR20_batch_polylog}.

Consider a NO input $X \in \{0,1\}^n$ to $L$, that is, an input that satisfies the soundness requirement $d_{\sD} (X,L) > \varepsilon$, over the unknown distribution $\sD$. To start with, $\Pi'$ runs the interactive reduction $\Pi_\NC$ from $L$ to $\PVAL(J,\Vec{v})$ with the same value of $t = \vert J \vert = \Tilde{O}(\varepsilon n)$.  

Setting $t$ to be $O(\log (\vert B_\sD (X) \cap B_\sU (X) \vert)) \leq O(\log(\vert B_\sU (X) \vert)) = \tilde{O}(\varepsilon n)$, we can generalise the guarantees of $\Pi_\NC$ to show that the intersection of $\sB_\sU(X)$ and $\sB_\sD(X)$ is disjoint from $\PVAL(J,\Vec{v})$, with high probability. Indeed, this builds on the earlier argument (and Footnote \ref{fn:gkr}), but over $\sB_\sU(X) \cap \sB_\sD(X)$, alongside the fact that the size of this set is upper bounded by the size of $\sB_\sU(X)$. Thus, $X$ cannot be $\varepsilon$-close to $\PVAL(J,\Vec{v})$ along \textit{both} $\sU$ and $\sD$, or in other words, $X$ is $\eps$-far from every element of $\PVAL$ along at least one of the two distributions (see Figure \ref{fig:PVAL} and Section \ref{sec:intuition_reduction_hybrid} for details).

Following this, assume that $d_\sD (X, \PVAL(J,\Vec{v})) > \varepsilon$. We construct the next stage of $\Pi'$, based on a case analysis whether $X$ is \textit{additionally} $\varepsilon$-far from $\PVAL(J,\Vec{v})$ under the uniform distribution or not. Indeed, suppose that $X$ is $\eps$-far from $\PVAL(J,\vec{v})$ under the uniform distribution. This is the easy case; we can catch this with the uniform $\IPP$ for $\PVAL(J,\Vec{v})$ as usual. 

On the other hand, suppose that instead, $X$ is close to $\PVAL(J,\Vec{v})$ under the uniform distribution, i.e., $d_\sU(X,\PVAL(J,\vec{v})) \leq \eps$. At this point, we observe (following \cite{RR20_batch_polylog}) that when $J$ is distributed uniformly at random, with high probability $\PVAL(J,\Vec{v})$ is a good error correcting code (i.e., with large minimal distance).\footnote{It is worth emphasising that this does not hold for every choice of $J$, for eg., $\PVAL(J,\Vec{v})$ is a bad error correcting code when $J$ consists of $t$ copies of the same point.} Since the output $J$ of $\Pi_\NC$ is distributed uniformly at random, when $X$ is $\varepsilon$-close to $\PVAL(J,\vec{v})$ over the uniform distribution, $\Pi_\NC$ guarantees that $X$ is in fact close to a \textit{unique} element $X'$ in $\PVAL(J,\Vec{v})$. 

To summarize, so far we have that $X$ is $\eps$-close to $X' \in \PVAL(J,\Vec{v})$ along $\sU$, but by our soundness condition, $X$ is $\eps$-far from $\PVAL(J,\Vec{v})$, and in particular from $X'$, along $\sD$. Now, the verifier uses the sample oracle to $\sD$ to generate $O(1/\varepsilon)$ samples, which we denote by $I\subseteq [n]$, and the corresponding values in $X$ given by $X \vert_I$. From the soundness assumption, with high probability there exists an index $i$ in $I$ such that $X_i \neq X'_i$. Combining this with the fact that every element in $\PVAL(J,\Vec{v})$ other than $X'$ is $\varepsilon$-far from $X$ along the uniform distribution, $X'$ is not in $\PVAL((J,I),(\Vec{v},X \vert_I))$, where $\PVAL$ is parameterised over a larger set. In other words, we see that $X$ is $\eps$-far from $\PVAL((J,I),(\Vec{v},X \vert_I))$ along the \textit{uniform distribution} and a uniform $\IPP$ for $\PVAL((J,I),(\Vec{v},X \vert_I))$ suffices.

The argument for completeness trivially holds from the guarantees of $\Pi_\NC$ and definition of an $\LDE$ of $X$, since in this case $X \in \PVAL((J,I),(\Vec{v},X \vert_I))$. We end with a quick note on the complexity of the distribution-free $\IPP$. The query complexity of $O(1/\varepsilon)$ is the same as that of the uniform $\IPP$ by \cite{RR20_batch_polylog}, and the communication complexity is the sum of the number of bits used to send the $O(1/\varepsilon)$ samples in $I$ in addition to the communication by the uniform $\IPP$, which is $\Tilde{O}(\varepsilon n)$. Overall the communication complexity is $\tilde{O}\left(\frac{1}{\varepsilon}+\varepsilon\cdot n\right)$ which matches that in \cite{RR20_batch_polylog} (up to poly-logarithmic factors) whenever $\eps \geq 1/\sqrt{n}$. 

\subsubsection{Proof outlines of Theorems \ref{thm:informal_ipp_dispersed} and \ref{thm:informal_product_dfipp}}
\label{sec:technique_special_dist}
Next, we describe the proof techniques of Theorems \ref{thm:informal_ipp_dispersed} and \ref{thm:informal_product_dfipp} that construct $\IPP$s for $\NC$ over smooth distributions and product distributions, matching the complexities of \cite{RVW} for every value of $\varepsilon$. This improves over the communication complexity of the distribution-free $\IPP$ in Theorem \ref{thm:informal_dfipp_nc} when $\varepsilon \ll 1/\sqrt{n}$ (with roughly the same query complexity). We follow the general strategy by \cite{RVW} and the main technical challenges arise during the analysis with respect to the new promise on the soundness of an $\IPP$ for $\PVAL$. We assume some familiarity with the uniform $\IPP$ construction by \cite{RVW} for this section; see also Section \ref{sec:large_int_fold} for more detailed intuition.

\paragraph*{Uniform $\IPP$ for $\PVAL(J,\Vec{v})$.} We start with a summary of the Uniform $\IPP$ from \cite{RVW}. Let the input $X \in [k]^m$, for $k = \log n$ and $n = k^m$. Further, let $\vert J \vert = t$. 

\cite{RVW} use a divide and conquer approach, by decomposing the $t$ claims about $X$ into new claims for each individual row instance $X_i \in \F^{k^{m-1}}$, for every $i\in [k]$. In more detail, let $J = (J_1,J_2)$, where the first component $J_1 \subset \F$ and $J_2 \subset \F^{m-1}$. The prover sends the matrix $Y' \in \F^{k\times t}$, where each row $Y'_i$ is the purported set of evaluations of the $(m-1)$-variate $\LDE$ (of individual degree $k-1$) of $X_i$ on $J_2$. By the definition of an $m$-variate $\LDE$ on $X$, the prover cannot lie about the consistency of $Y'$ with $\Vec{v}$, since for each $(a,b) \in J$ (where $b \in J_2$), the verifier can easily check if the univariate $\LDE$ of $Y'[\cdot,b]$ (the $b^{\text{th}}$ column of $Y$) evaluated on the coordinate $a$ equals $\Vec{v}[(a,b)]$ (see Figure \ref{fig:polyfold_intro}).

Thus, the initial $\PVAL$ instance is now reduced to $k$ instances $X_i \in \F^{k^{m-1}}$ for $\{\PVAL(J_2,Y'_i)\}$. A natural idea now is for the verifier to send a random vector $z \in \F^k$ to the prover, and ask it back for a ``folded" version $X' \in \F^{k^{m-1}}$, that is purported to be $z \cdot X$.\footnote{The dot product $z \cdot X \in \F^{k^{m-1}}$ between $z \in \F^{k}$ and a matrix $X \in \F^{k \times k^{m-1}}$ is given by $\sum_{i=1}^k z_i X_i$.} Now, the $\IPP$ could recurse on a \textit{single input} $X' \in \F^{k^{m-1}}$ that has shrunk in size by a factor of $k$, to the problem $\PVAL(J_2, z \cdot Y')$. Completeness easily holds, since if $X$ belonged to $\PVAL(J,\Vec{v})$, then the honest prover will just send the ``true" $Y' \in \F^{k \times t}$ and the verifier checks always pass. 

\paragraph{Uniform Distance Preservation Lemma.} However showing soundness is not straightforward. Suppose that $X$ is $\varepsilon$-far from $\PVAL(J,\Vec{v})$ under the uniform distribution. It turns out that the malicious prover has cheated in at least one row of the purported matrix $Y'$ (if not, since $X$ is not in $\PVAL$, there would be at least one column in $Y'$ which would be inconsistent with the corresponding value in $\vec{v}$ and the verifier would catch the prover in the checks made above). 

For any row $X_i \in \F^{k^{m-1}}$ that is a lower-dimensional input instance, let $\varepsilon_i$ be the distance between $X_i$ and $\PVAL(J_2, Y'_i)$. To ensure that the verifier catches the cheating prover, the folded instance $X'$ also needs to be reasonably far from $\PVAL$ on a lower dimension at the end of a recursive step. In order to capture this, \cite{RVW} (implicitly) use a \textit{uniform distance preservation lemma}, which states that if $X$ is $\varepsilon$-far from $\PVAL(J,\Vec{v})$, then $\sum_{i=1}^k \varepsilon_i > k \varepsilon$. 

Using the uniform distance preservation lemma, \cite{RVW} observe that if the prover ended up cheating (roughly) uniformly across all rows in $Y'$, then any row $X_i$ would be roughly $\varepsilon$-far from $PVAL(J_2, z \cdot Y'_i)$, and the $\IPP$ would recurse by picking a single row at random. However, the prover could have cheated across multiple rows of $Y'_i$ and the verifier does not know these rows. To accommodate this, the verifier considers $\log(k)$ many random foldings of $X$, where the Hamming weight of the vectors $z$ used to fold $X$, range across $1$ to $k$ (in powers of $2$). In particular, this results in $O(\log (\log (n)))$ recursive instances in $\F^{k^{m-1}}$. Crucially, they use the uniform distance preservation lemma to generalise the intuition above and show that for at least one of these folded instances, the distance is roughly preserved. Moreover, for such a folded instance, the product of the new distance and the effective query complexity (the number of queries on $X$ to compute the value at any index in $z \cdot X$) is $O(1/\varepsilon)$, along with small but super-constant multiplicative factors.

The $\IPP$ continues to recursively fold the instance dimension-wise by the above process, until the size of each final folded instance becomes $\Tilde{O}(\varepsilon n)$, which happens after $\Omega(\log(n)/\log(\log(n)))$ steps. In such a case, the prover directly sends each final instance. Since there exists an instance $\Tilde{X}^j$ at each level of recursion for which distance is preserved, there exists a final folded instance $\tilde{X}$, such that the verifier catches a cheating prover by uniformly \textit{sampling} a few coordinates of $\Tilde{X}$. Moreover, since the product of the distance and effective query complexities for each $\Tilde{X}^j$ are roughly maintained to be small at each step of the recursion, making $O(1/\varepsilon^{1+o(1)})$ many queries to $\Tilde{X}$ is sufficient to catch the cheating prover (since the total number of recursive instances after the stated number of steps is roughly $n^{o(1)} = 1/\varepsilon^{o(1)}$). The communication complexity is simply the number of bits used to send all the final folded instances, in addition to sending the matrices $Y'$ of size $k \times t$, and thus is $\Tilde{O}(\varepsilon^{1-o(1)} n)$.

\paragraph*{$\IPP$s for $\NC$ under specific distribution families.}
 We now highlight some key ideas which help us construct $\IPP$s over large distribution families like smooth distributions and product distributions. To begin with, on any input $X \in \{0,1\}^{k^m}$, we first reduce $L$ to $\PVAL$ using $\Pi_\NC$. Recall that in the distribution-free setting, $\Pi_\NC$ outputs $(J,\vec{v})$, such that for the soundness promise, with high probability $X$ cannot be $\varepsilon$-close to $\PVAL(J,\Vec{v})$ along both $\sU$ and the unknown distribution from the given family, $\sD$. In other words, $X$ is $\varepsilon$-far from $\PVAL(J,\Vec{v})$ along at least one of $\sU$ or $\sD$. Building on this observation, we design $\IPP$s for $\PVAL(J,\Vec{v})$ over these distribution families, using an intricate case analysis of the soundness condition. 

In more detail, if $X$ is $\varepsilon$-far from $\PVAL(J,\vec{v})$ under the uniform distribution, then we can directly use the uniform distance preservation lemma to catch a malicious prover as seen previously in the uniform $\IPP$. If not, suppose that $d_\sD(X,\PVAL(J,\vec{v})) > \varepsilon$. Next, we briefly describe the soundness analysis, using \textit{specific distance preservation lemmas} for smooth distributions and product distributions. Given this, we build on the strategy of the uniform $\IPP$ above to construct an $\IPP$ for $\PVAL(J,\Vec{v})$ over these distribution families, with the main technical work being that of simultaneously incorporating both the uniform and the respective distance preservation lemmas into the soundness analysis, across the recursive levels.

\paragraph{$\rho$-dispersed distributions.} Recall that $\rho$-dispersed distributions over $[k]^m$ capture the smoothness of a distribution, by 
requiring that the probability mass on any element is bounded by $\rho$ times the average mass on any of its neighbours. Adopting similar notation as above, let $\hat{\sD}$ be the marginal distribution of $\sD$ over $[k]^{m-1}$.

For any row $X_i \in \F^{k^{m-1}}$ that is a lower-dimensional input instance, let $\varepsilon_i$ be the distance between $X_i$ and $\PVAL(J_2, Y'_i)$ over $\hat{\sD}$. Here, we show a distance preservation lemma for $\rho$-dispersed distributions, such that for any distribution $\sD$ that is $\rho$-dispersed, $\sum_{i=1}^k \varepsilon_i > (k\varepsilon)/\rho$.\footnote{Note that the uniform distribution is a $1$-dispersed distribution and we thus generalise the uniform distance preservation lemma.} The idea behind proving this is not obvious immediately; while $\varepsilon_i$ measures the distance along marginal distributions, $\eps$ is the distance from each element of $\PVAL(J,\vec{v})$ over $\sD$ (which could be a joint distribution). However, we crucially use properties about $\rho$-dispersed distributions to prove this distance preservation lemma. 

Using the strategy described earlier, we get an $\IPP$ for $\NC$ over $\rho$-dispersed distributions, having query and sample complexities $\frac{\rho^{\log (1/\varepsilon)/\log\log (n)}}{\varepsilon ^{1+o(1)}}$, while keeping communication complexity the same. In particular, for $\rho = k^{o(1)}$, the query complexity is $1/\varepsilon^{1+o(1)}$ and matches that of the uniform $\IPP$ for all $\varepsilon > 0$. We refer to Section \ref{sec:large_int_fold} for further intuition about this.

\paragraph*{Product distributions.} Let $\sD$ be an $m$-product distribution defined as $\sD = \sD_1 \times \dots \sD_m$ over $[k]^m$, where $k = \log(n)$, and each $\sD_j$ is an independent distribution supported on $[k]$. In particular, $\sD(i_1, \dots, i_m)$ is defined as $\prod_{j=1}^m \sD_j(i_j)$. 

Our main approach here to construct $\IPP$s over such distributions, is to first \textit{learn} the underlying distribution and then use this as an aid to obtain near-optimal complexity parameters. For more context, consider the following $k$-dispersed distribution $\sD$ over $[k]^m$, that is supported on the first row of the first dimension, i.e, exactly on the set of elements of the form $(1,i_2, \dots, i_m)$ for every $(i_2, \dots, i_m) \in [k]^{m-1}$.\footnote{See Section \ref{sec:def_dispersed}; intuitively, for any $i_2, \dots, i_m \in [k]^{m-1}$, $\sD(1,i_2, \dots, i_m)$ is the only element in the set $\{\ell, i_2, \dots, i_m\}_{\ell \in [k]}$ with a non-zero probability mass and thus is $k$-times the average of the probability mass on its neighbourhood.} We see that the $\IPP$ over $k$-dispersed distributions has query complexity $O(1/\varepsilon^2)$. However, if the verifier ``learns'' beforehand that $\sD$ is only supported on the first row, then it can focus its attention on a smaller instance in $\F^{k^{m-1}}$ and potentially obtain much better query complexity, if $\sD$ conditioned on the first row is $\rho$-dispersed, for a small $\rho$.

Our main technical idea here is to show a \textit{learning-augmented} distance preservation lemma for product distributions. Let $\varepsilon_i$ be the distance between $X_i$ and $\PVAL(J_2, Y'_i)$ over $\hat{\sD} = \sD_2 \times \dots \times \sD_m$. Based on an alternative analysis to that of $\rho$-dispersed distributions, we prove that for any product distribution $\sD$, $\sum_{i=1}^k \varepsilon_i > C \varepsilon$, for $C > 1$ that \textit{only depends on} $\sD_1$. Using this key insight, if the verifier ``transformed" $\sD_1$ into the uniform distribution over $[a_0 \cdot k]$, where $a_0 \geq 1$ is a small constant, then we get a similar expression as the uniform distance preservation lemma, i.e., $C = O(k)$, despite still measuring distance according to $\hat{\sD}$ for the lower dimensional instances.\footnote{For consistency, $a_0 = 1$, when $\sD_1$ is just $\sU_k$.} 

We briefly highlight the sequence of tools used to implement the latter idea. The verifier learns the probability vector of $\sD_1$, into an approximation $\sP_1$, using the \textit{parallel set lower bound protocol} \cite{BT06} which requires white-box access to $\sD_1$. Following this, it runs a \textit{``granularising"} algorithm taking $\sP_1$ as input, that outputs the probability vector of a new $8k$-granular distribution $\mathcal{E}_1$ over $[k+1]$ (i.e., for every $i$, $\mathcal{E}_1(i)$ is $b_i/8k$), such that in the soundness case, the distance of the input over $\mathcal{E}_1$ is still $\varepsilon$ (up to constant factors). Finally, this granularity set is used to ``extend" $X$ into a new input instance $X' \in \{0,1\}^{8k \times k^{m-1}}$, by making copies of each row according to it's granularity, and we can thus, equivalently consider the underlying row distribution as the uniform distribution over $[8k]$. The last two steps build on ideas from \cite{Gol20} for testing unknown distributions, while our focus is on the setting of testing with an implicit input.

The details of adapting both distance preservation lemmas and the analysis of the $\IPP$, to handle changing distributions and input sizes across different levels of recursion, is found in Section \ref{sec:dfipp_learnable}.

\subsection{Related Work}
\label{sec:rel_work}

\paragraph{Proofs of Proximity for Distributions.} In a related model, \cite{CG18,HR22} study proofs of proximity for \textit{testing distributions}. In their setting, for a fixed property $\Pi$ of distributions, the verifier receives samples from an unknown distribution $\sD$, and interacts with the prover to decide whether $\sD \in \Pi$ or $\sD$ is $\varepsilon$-far from any distribution in $\Pi$ along the total variation distance. While there are superficial similarities to our model regarding the use of sample oracle, we focus on testing properties (or languages) of strings, where the distribution oracle only provides a means of accessing the input string. In addition, the verifier also has oracle access to the input instance and the distance for the NO instance is measured with respect to the underlying distribution. 

\paragraph{Sample-based $\IPP$s.} Another related model is that of Sample-based $\IPP$s \cite{GR22_sample}, where the verifier can \textit{only} access the input through an oracle that provides labeled samples over the uniform distribution. They show that any language in logspace-uniform $\NC$ has an $\mathsf{SIPP}$ with $\Tilde{O}(\sqrt{n})$ sample and communication complexities, by in fact constructing a reduction protocol from an $\mathsf{SIPP}$ to the query-based $\IPP$ by \cite{RVW}. Our model is more general conceptually, since any protocol in our model needs to be able to test for a language given access to labeled samples over any unknown distribution. On the other hand, to aid with this generality, we also provide the verifier with the more powerful oracle access to the input, which $\mathsf{SIPP}$s do not. 

That being said, we can use the uniform $\mathsf{SIPP}$ by \cite{GR22_sample} within the proof of Theorem \ref{thm:informal_dfipp_nc} (instead of the query-based $\IPP$ by \cite{RR20_batch_polylog}) to obtain a distribution-free $\mathsf{SIPP}$ for $\NC$ where the verifier only accesses the input through labeled samples over $\sU$ and the unknown distribution $\sD$, for any $\varepsilon \geq \tau/n$.\footnote{The uniform \textsf{SIPP} by \cite{GR22_sample} has communication complexity $\Tilde{O}\left(\frac{n}{\tau} + \frac{1}{\eps} \right)$ (for tradeoff  $\tau \leq \sqrt{n})$, and using this still gives us the same communication complexity as the query-based distribution-free $\IPP$ from Theorem \ref{thm:informal_dfipp_nc}.} It is unclear whether we can construct distribution-free $\mathsf{SIPP}$s for general values of $\varepsilon$ (even over smooth or product distributions) that match the complexities of the uniform $\IPP$s and we leave it as future work.

\paragraph{Interactive Proofs for Agnostic Learning.} \cite{GRSY21} study the setting of verifying PAC-learners. There, the verifier has sampling access to an unknown distribution $\sD$ over labeled examples of the form $(i,x_i)$, where $i \sim \sD$ and $x$ is the underlying input. It's goal is to verify whether a hypothesis $h : \{0,1\}^{\log(n)} \rightarrow \{0,1\}$ given by the prover from a fixed hypothesis class, is the best approximation of $\sD$. From the property testing perspective, the prover wants to convince the verifier that $\sD'$ has the property that every hypothesis in the class has error larger than $\varepsilon$ over $\sD$, for some $\varepsilon > 0$ (i.e., the best possible approximation of $\sD$ by the hypothesis class is at least $\varepsilon$).

Similar to the setting of $\mathsf{SIPP}$s, their scenario focuses on the case where the verifier only has access to $x$ via a labeled sample oracle, over an unknown distribution. Furthermore, they focus on testing specific properties pertaining to machine learning, such as closeness to an underlying hypothesis class, with the hope of getting very low sample complexity (with respect to the VC dimension of the hypothesis class). In contrast, we deal with verification of general classes of properties, and in some cases the sample and query complexities are both $\Tilde{O}(\sqrt{n})$.

\section{Preliminaries and definitions}

We denote $[n]=\{1,2,\cdots, n\}$. A language $L$ is defined as $L=\bigcup_{i\in \mathbb{N}}L_{n} \subseteq \{0,1\}^*$, where each $L_n = L \cap \{0,1\}^n$.

Throughout this work, we consider languages computable by logspace-uniform Boolean circuits on $n$ variables of size $S(n)$ (number of gates) and depth $d(n)$ (longest path from the output gate to some input), with $\mathsf{XOR}$ and $\mathsf{AND}$ gates of fan-in two. Of particular interest is the class logspace-uniform $\NC$, which is the class of languages computable by logspace-uniform circuit families of size $\poly(n)$ and depth $O(\log^i (n))$ for some fixed $i\in \N$. In more detail, $L$ belongs to logspace-uniform $\NC$, if there exists $i\in \N$ and a logspace Turing machine $M$ that takes input $1^n$ and outputs the description of an $n$-variate circuit of depth $O(\log^i (n))$, such that for each $x \in \{0,1\}^n$, $C(x) = 1$ if and only if $x \in L$. 

\subsection{Hybrid metrics}
We denote by $\Delta(\Omega_{n})$, the simplex of all possible distributions over a domain $\Omega_{n}$ of size $n\in \mathbb{N}$. Let $\mathbb{F}$ be a finite field and let $x,y$ be vectors in $\mathbb{F}^n$. For any distribution $\mathcal{D}\in \Delta(\Omega_n)$, we define the distance between $x$ and $y$ as
\begin{equation*}
d_{\mathcal{D}}(x,y)=\underset{i\sim D}{\mathbb{P}}\left[x_{i}\neq y_{i}\right].
\end{equation*}
We use $\mathcal{U}_{n}$ to denote the uniform distribution over the set $[n]$, where the size of the support is clear we denote the uniform distribution by $\mathcal{U}$. Note that if the distance is measured according to $\sU$, then this is simply the normalised Hamming distance. 

For any (non-empty) $L\subseteq \mathbb{F}^{n}$ and any vector $x \in \mathbb{F}^{n}$, we similarly define the distance between $x$ and $L$ as:
\begin{equation*}
d_{\mathcal{D}}(x,L)=\min\limits_{y\in L}d_{\mathcal{D}}(x,y).
\end{equation*}
If $d_{\mathcal{D}}(x,L)>\varepsilon$, we say that $x$ is $\varepsilon$-far from $L$ over the $\mathcal{D}$ distribution, otherwise we say it is $\varepsilon$-close. Furthermore, for any $n\in \N$, $\sD\in\Delta(\Omega_{n})$, $\varepsilon>0$, and $X\in \F^{n}$, we denote by $B_{\mathcal{D},\varepsilon}(X)$ as the subset of $\F^n$ that is $\varepsilon$-close to $X$ along $\sD$. In other words,
\begin{equation}
    B_{\sD,\varepsilon}(X)=\{Y\in \F^{n}\vert d_{\sD}(X,Y)<\varepsilon\}.
\end{equation}

The hybrid metric is the maximum over two distances, this increases the size of the set of elements $\varepsilon$-far from an input $X$. The notion of a hybrid metric is key to our proof of Theorem \ref{thm:informal_dfipp_nc}, see Section \ref{sec:dfipp_low-depth} for details. We define the hybrid metric as follows.

\begin{definition}[Hybrid Metrics]
    For any pair of distributions $\mathcal{D}_{1}$, $\mathcal{D}_{2}$ over $[n]$, we define the \textit{($\mathcal{D}_{1}$, $\mathcal{D}_{2}$)-Hybrid Metric} $\mu_{\mathcal{D}_{1},\mathcal{D}_{2}}$ as follows. 
\begin{equation*}
    \mu_{\mathcal{D}_{1}, \mathcal{D}_{2}}(x,y)=\underset{s\in\{\mathcal{D}_{1}, \mathcal{D}_{2}\}}{max}(d_{s}(x,y)).
\end{equation*}
\end{definition}

\begin{remark}
Note that taking the maximum over two metrics is also a metric, as the triangle inequality follows since for some $s\in \{\mathcal{D}_{1}, \mathcal{D}_{2}\}$, it holds that

\begin{equation*}
    \mu_{\mathcal{D}_{1}, \mathcal{D}_{2}}(x,y)=d_{s}(x,y)<d_{s}(x,z)+d_{s}(z,y)\leq \mu_{\mathcal{D}_{1}, \mathcal{D}_{2}}(x,z)+\mu_{\mathcal{D}_{1}, \mathcal{D}_{2}}(z,y).
\end{equation*}

In addition, the definitions of distance of an input string to a language extend in a natural way with respect to $\mu_{\sD_1,\sD_2}$.
\end{remark}

\subsection{Interactive Proofs of Proximity ($\IPP$)}
We refer to the standard textbook \cite{AB09} for the definition of an interactive proof (IP). 
$\IPP$s \cite{EKR04,RVW} are interactive proofs that verify the ``closeness'' of an input string to the given language. In these interactive proofs, the verifier must accept if the input is in the language and reject when it is far with some computation performed by an untrusted prover. The goal is to achieve verification using sublinear queries and communication, by not having to read the input completely. Following previous literature on $\IPP$s, we view the inputs given to the verifier as having two parts: an \textit{implicit} input $X \in \F^n$ and an \textit{explicit} input $w \in \F^*$ ($w$ could be empty), for some finite field $\F$. The verifier can access $X$ only via an oracle (query or sample), but can read $w$ in its entirety. We then refer to $\{L_w\}_{w \in \F^*}$ as a family of \textit{parameterised} languages, each language being parameterised by the explicit input.\footnote{Equivalently, we can view $L$ as a language over \textit{pairs} $(X,w)$ and define each $L_w = \{X \mid (X,w) \in L\}$. The closeness of a string to $L$ is only measured with respect to $X$, the first string in the pair.} At times, we will refer to this family of languages simply as $L$ and take the implicit input as already given to the $\IPP$.

For any language $L_w$, we denote by $(P(X),V^{X})(w, n, \varepsilon)$ as the output of the interaction between a verifier $V$ having query access to an input $X$ of length $n$ and a prover $P$ with explicit access to $X$, when both have full access to the shared inputs $w$, $n$, and $\varepsilon$. An $\IPP$ over the uniform distribution is defined as follows.

\begin{definition}[\textbf{$\IPP$s over the Uniform Distribution \cite{EKR04, RVW}}]
\label{def:uniform_ipp}
For any fixed string $w \in \F^*$, let $L_w \subseteq \F^*$ be a parameterised language. We say that $L$ has an interactive proof of proximity ($\IPP$) if there exists a proof system $(P,V)$ with a (possibly computationally unbounded) prover $P$ and a computationally bounded verifier $V$, such that for every $n$, input $X \in \F^n$ and proximity parameter $\varepsilon > 0$, the following hold.

When $P$ has full access to $X,w,n,\varepsilon$, and when $V$ is given query access to $X$ and full access to $w,n,\varepsilon$, the following hold:
\begin{itemize}
    \item Completeness: If $X \in L_w$, then
    \begin{equation*}
        \underset{V}{\mathbb{P}}\left[ (P(X),V^{X})(w, n, \varepsilon)=1 \right]\geq \frac{2}{3}.
    \end{equation*}

    \item Soundness: If $d_{\mathcal{U}_{n}}(X, L_w)>\varepsilon$, then for every computationally unbounded prover $P^{\ast}$ we have 
    \begin{equation*}
        \underset{V}{\mathbb{P}}\left[(P^{\ast}(X),V^{X})(w,n,\varepsilon)=0\right]\geq \frac{2}{3}.
    \end{equation*}
\end{itemize}

Furthermore, we say that the $\IPP$ has \textit{query complexity} $q = q(n, \vert w \vert, \varepsilon)$, \textit{communication complexity} $c = c(n, \vert w \vert, \varepsilon)$ and \textit{round complexity} $R = R(n, \vert w \vert, \varepsilon)$, if $P$ and $V$ exchange at most $c$ bits in at most $R$ rounds of interaction (having 2 messages per round) and $V$ makes at most $q$ many queries during this process, for every $w$, $X \in \F^n$, and $\varepsilon > 0$.

Additionally, we call this $\IPP$ a Merlin-Arthur proof of proximity ($\MAP$) if over the course of this protocol, the verifier does not send any messages to the prover.
\end{definition}

\noindent Below, we state the main result from \cite{RVW}.
\begin{theorem}[\textbf{$\IPP$ for Low Depth Languages over the Uniform Distribution \cite{RVW}}]
\label{thm:unifRVW}
For every language $L\subseteq \{0,1\}^{n}$ and $\varepsilon\in (0,1]$ computable by log-space-uniform circuits of depth $\Delta_{L}=\Delta_{L}(n)$ and size $S=S(n)$, there exists an interactive proof of proximity for $L$ with perfect completeness and soundness at least $1/2$. 

This $\IPP$ has query complexity $\frac{1}{\varepsilon ^{1+o(1)}}$, communication complexity $\varepsilon\cdot n\cdot \left(\frac{1}{\varepsilon^{o(1)}}\right)+\varepsilon\cdot n\cdot \poly(\Delta_{L})$ and round complexity $O\left(\frac{\log\left(\frac{1}{\varepsilon}\right)}{\log\log (n)}+\Delta_{L}\cdot\log (S)\right)$. In addition, the honest prover runs in time $\poly(S,n)$ and the verifier runs in time $(\frac{1}{\varepsilon})^{1+o(1)}+(\varepsilon\cdot n)^{1+o(1)} \poly(\Delta_{L})$.
\end{theorem}

\section{Distribution-free $\IPP$s}
\label{sec:dfipp_def_separations}

In this section, we define the notion of distribution-free proofs of proximity and provide complexity theoretic insight regarding the power and limitation of the model. We start by defining the notion of a distribution-free proof of proximity and then extensions of this notion to the white-box model and polynomially-samplable distributions.  From these definitions we explore some observations of this model. In Section \ref{sec:corr}, for certain structured languages we show that the existence of an $\IPP$ is equivalent to the existence of a distribution-free $\IPP$. Following this, in Section \ref{sec:sym}, we demonstrate an exponential separation between property testing and $\IPP$s in the distribution-free setting. In section \ref{sec:dfUseparation}, we use a lower bound for $\IPP$ from \cite{KR15} to demonstrate a separation between uniform $\IPP$s and distribution-free $\IPP$s.

Let $\sD = \{\sD_n\}_{n \in \N}$ be a distribution ensemble, where each $\sD_n \in \Delta(\Omega_n)$. Whenever the context is clear, we abuse notation by dropping the support size in $\sD_n$. 

Distribution-free $\IPP$s (DF-$\IPP$s) are interactive proofs that verify the closeness of an input to a language $L$ under any arbitrary distribution. If the input is in the language, the DF-$\IPP$ must accept and if it is far along this distribution, it must reject. Here, the verifier additionally has sample access to an input string $X\in \{0,1\}^{n}$ over an unknown (but fixed) distribution $\mathcal{D}$ over $[n]$, via a \textit{sample oracle} $\mathcal{O}_\mathcal{D}(X)$. The oracle $\mathcal{O}_\mathcal{D}(X)$ returns the tuple $(i,X_i)$, for an index $i$ independently sampled from $\mathcal{D}$. The soundness condition now requires the algorithm to reject strings that are $\varepsilon$-far from the language \emph{along the distribution $\mathcal{D}$}. Additionally, the cheating prover has full access to the distribution, i.e., the prover has access to all of the individual probabilities that constitute the distribution.
 
\begin{definition}[\textbf{Distribution-Free $\IPP$}]
\label{def:df_ipp}
For any fixed string $w \in \F^*$, let $L_w \subseteq \F^*$ be a parameterised language. We say that $L_w$ has a distribution-free $\IPP$ if there exists a proof system $(P,V)$, where $P$ is a (possibly computationally unbounded) prover and $V$ is a computationally bounded verifier $V$, such that for every $n$, input $X \in \F^n$, proximity parameter $\varepsilon > 0$, and for any fixed (but unknown) distribution $\sD_{n} \in \Delta(\Omega_n)$ from a distribution ensemble $\sD = \{\sD_n\}_{n \in \N}$, the following hold.

When $P$ has full access to $X,w,n,\varepsilon$ and $\sD_{n}$, and when $V$ is given query access to $X$, as well as sample access to $X$ via $\mathcal{O}_{\sD_{n}}(X)$ and full access to $w,n,\varepsilon$, the following conditions hold.
\begin{itemize}
   \item Completeness: If $X \in L_w$, then
    \begin{equation*}
        \underset{V, \mathcal{O}_{\sD_{n}}(X)}{\mathbb{P}}\left[ \left( P(X,\mathcal{D}),V^{X, \mathcal{O}_{\sD_{n}}(X)} \right)(w,n,\varepsilon) = 1 \right] \geq \frac{2}{3}.
    \end{equation*}
    In other words, if $X\in L_w$ then the verifier accepts the input with probability at least $2/3$ over its own randomness and the samples from $\mathcal{O}_{\sD_{n}}(X)$.

    \item Soundness: If $d_{\sD_{n}} (X, L_w)>\varepsilon$, then for any computationally unbounded prover $P^{\ast}$ we have 
    \begin{equation*}
        \underset{V, \mathcal{O}_{\sD_{n}}(X)}{\mathbb{P}} \left[ \left( P^{\ast}(X,\mathcal{D}),V^{X, \mathcal{O}_{\sD_{n}}(X)} \right)(w,n,\varepsilon)=0 \right]\geq \frac{2}{3}.
    \end{equation*}
    In other words, if $d_{\sD_{n}}(X,L_w) > \varepsilon$, the verifier rejects with all but $1/3$ probability over its own randomness and the samples from $\mathcal{O}_{\sD_{n}}(X)$, regardless of the cheating prover strategy.
\end{itemize}

The query complexity $q = q(n,\vert w \vert, \varepsilon)$, communication complexity $c = c(n,\vert w \vert, \varepsilon)$ and round complexity $R = R(n, \vert w \vert, \varepsilon)$ of the interactive proof are as defined earlier. In addition, the $\IPP$ has \textit{sample complexity} $s = s(n, \vert w \vert, \varepsilon)$, if $V$ invokes the sample oracle $\mathcal{O}_\sD(X)$ at most $s(.)$ times during its interaction with $P$, for any $w,X$ and $\varepsilon > 0$. 
\end{definition}
\noindent Similar to the non-interactive form of an $\IPP$, we can also define a distribution-free $\MAP$. Moreover, analogous to the PAC-learning setting, we can also consider a fixed set of distributions $\sF$ and define a distribution-free $\IPP$ over $\sF$, by requiring the correctness of the $\IPP$ to hold only over the distributions in $\sF$. It will be necessary for us in Sections \ref{sec:dfipp_low-depth} and \ref{sec:dfipp_learnable} to consider $\IPP$s with a soundness condition over the hybrid metric $\mu$; for simplicity we will still refer to these as $\IPP$s. 

\begin{remark}
    It is worth noting that while Definition \ref{def:df_ipp} provides the honest prover with a full description of $\sD$, most of our protocols enjoy the property that the honest prover does not require the description. 
\end{remark}

\begin{remark}
It is worth noting that for typical properties that are testable given the entire input, both $q(n)$ and $c(n)$ have to be sublinear in $\vert X \vert$ for the $\IPP$ to be non-trivial. Indeed, if $V$ sees all of $X$, it can directly check if $X$ belongs to $L$ or not. On the other hand, if $P$ sends a string $X'$ of length $n$ (purported to be $X$), then $V$ checks if $X' \in L$ and perform an equality test between $X'$ and $X$ over $O(1/\varepsilon)$ samples from $\mathcal{O}_\sD(X)$. Completeness follows when $X' = X$, whereas soundness follows from the distance guarantee of the input $X$. 
\end{remark}

\paragraph{White-Box Distribution-Free $\IPP$.}
In contrast to Definition \ref{def:df_ipp} where the verifier has sample access to the input only via $\mathcal{O}_\sD(X)$, we define distribution-free $\IPP$s in the \textit{white-box model}, where the verifier now gets the sampling device to the distribution, in the form of a circuit (a notion explored by Sahai and Vadhan in \cite{SV97}), in addition to oracle access to the input. The distribution is defined by the output of the circuit when it is fed with a random input of suitable length. 

In more detail, $C$ takes a uniformly random string as input and outputs an index in $[n]$, such that the probability of sampling the index using $C$ is the same as that of $\sD$. We consider distributions that are polynomially samplable, i.e., the circuit $C$ takes $\polylog(n)$ many random bits, outputs an index in $[n]$, and its size is polynomial in the number of its inputs (i.e., the size of $C$ is $\polylog(n)$). More formally,
\begin{definition}[Polynomially samplable distributions]
\label{def:psamp}
Let $\sD = \{\sD_n\}_{n \in \N}$ be a distribution ensemble, where each $\sD_n \in \Delta \left( \Omega_n\right)$. $\sD$ is said to be polynomially-samplable, if there exists a family of circuits $C = \{C_{r(n)}\}_{n \in \N}$ of size $\poly(\log (n))$, where $C_{r(n)} : \{0,1\}^{r(n)} \rightarrow \{0,1\}^{\log (n)}$, such that for each $n \in \N$, the output distribution of $C_{r(n)}$ is the same as $\sD_n$, i.e., for every $i \in [n]$, $\Pr_{x \sim U_{r(n)}} \{ C(x) = i\} = \sD_n(i)$.\footnote{In particular, this implies that $r(n) \leq \polylog(n)$.} 
\end{definition}

\begin{definition}[\textbf{White-Box Distribution-Free $\IPP$ over polynomially samplable distributions}]
\label{def:whitebox_df_ipp}
For any fixed string $w \in \F^*$, let $L_w \subseteq \F^*$ be a parameterised language. We say that $L_w$ has a white-box $\IPP$ over polynomially samplable distributions, if there exists a proof system $(P,V)$, where $P$ is a (possibly computationally unbounded) prover and $V$ is a computationally bounded verifier $V$, such that for every $n$, input $X \in \F^n$, proximity parameter $\varepsilon > 0$, and for any fixed (but unknown) distribution $\sD$ over $[n]$ from a distribution ensemble that is samplable using a polynomial-sized circuit $C : \{0,1\}^{\polylog(n)} \rightarrow \{0,1\}^{\log (n)}$, the following hold. 

When $P$ has full access to $X,w,n,\varepsilon$ and $\sD$, and when $V$ is given the sampling circuit $C$, query access to $X$, and full access to $w,n,\varepsilon$, we have the following conditions.
\begin{itemize}
   \item Completeness: If $X \in L_w$, then
    \begin{equation*}
        \underset{V}{\mathbb{P}}\left[(P(X),V^{X}(w,n,C,\varepsilon) = 1 \right] \geq \frac{2}{3}.
    \end{equation*}
    
    \item Soundness: If $d_{\sD_n} (X, L_w)>\varepsilon$, then for any computationally unbounded prover $P^{\ast}$ we have 
    \begin{equation*}
        \underset{V}{\mathbb{P}} \left[(P^{\ast}(X),V^{X}(w,n,C,\varepsilon)=0 \right]\geq \frac{2}{3}.
    \end{equation*}    
\end{itemize}

The query complexity $q = q(n, \vert w \vert, \varepsilon)$, communication complexity $c = c(n, \vert w \vert, \varepsilon)$ and round complexity $R = R(n, \vert w \vert, \varepsilon)$ of the $\IPP$ are as defined earlier. 
\end{definition}
More generally, we can define white-box $\IPP$s on distributions over $[n]$ samplable by circuit families of size $S(n)$, where $S(n) = 2^{o(\log (n))}$. In this case, the running time of the verifier is given by $T(n, \vert w \vert, \varepsilon, S(n))$, and typically, we require $T(n)$ to be sublinear in $n$.

\begin{remark}
Note that the sample complexity of the verifier is not a useful complexity parameter in the white-box model as the both the prover and the verifier get the entire sampling circuit.  Indeed, the verifier can sample an index from the circuit and query the input value at this index, or it can possibly use the circuit to perform other computations or simulate input access via some other distribution. Any samples made using the circuit for querying $X$ count towards the query complexity of the $\IPP$. Of course, the verifier can go over all possible inputs to the sampler circuit to know the entire distribution exactly, but then its running time is no longer sublinear. 
\end{remark}

\subsection{DF-$\IPP$s for (Relaxed) Correctable Languages}
\label{sec:corr}

In this section, we show a generic transformation from an $\IPP$ to a distribution-free $\IPP$ for any subset of a relaxed locally correctable code (RLCC), while maintaining the round complexity, query complexity and communication complexity, in Theorem \ref{thm:rlcc_dfipp}. This is an extension.

We show for this large and natural family of languages you can obtain a distribution-free $\IPP$ from a uniform $\IPP$. This includes properties of polynomials or any locally correctable codes as well as many regularly studied problems in the literature for probabilistically checkable proof. These RLCCs have exponential better parameters that locally correctable codes and have had significant recent developments in complexity theory.

For any field $\F$, a code is a subset $C\subseteq\F^{n}$ with relative distance $\delta>0$, if the relative minimum distance between any two elements in the code is at least $\delta$, in other words

\begin{equation*}
    w_{1}, w_{2}\in C \implies d_{\mathcal{U}}(w_{1}, w_{2})\geq \delta.
\end{equation*}

\begin{definition}[Locally Correctable Codes]
For any field $\F$, let $C \subseteq \F^n$ be an error correcting code with relative distance $\delta$. 
We say that C is locally correctable if there exists a correcting radius $\delta_{r}< \delta/2$ and an algorithm $A$, called the \textit{corrector}, such that when $A$ is give oracle access to an implicit input $w\in \F^{n}$ and an explicit input $i \in [n]$, the following hold.
\begin{enumerate}
    \item $w\in C \implies \mathbb{P}[A^w(i)=w_{i}]=1$.
    \item if $\exists c\in C$ such that $d_{\sU}(c,w) \leq \delta_{r}\implies\mathbb{P}_A \left[ A^w(i)=c_{i} \right] \geq \frac{2}{3}.$

\end{enumerate}

We say $A$ has query complexity $t$ if it uses at most that many queries to perform the correction, on any inputs.
\end{definition}

We now generalise this notion to relaxed locally correctable codes. In this setting, the corrector is allowed a third possible output $``\bot"$, indicating that it has aborted.

\begin{definition}[Relaxed Locally Correctable Codes (RLCCs)]
Let $C$ be an error correcting code with relative distance $\delta$. We say that C is locally correctable if there exists a $\delta_{r}< \delta/2$ and a corrector A, such that when $A$ is given oracle access to an implicit input $w\in\{0,1\}^{n}$ and explicit access to the input $i \in [n]$, the following hold.

\begin{enumerate}
    \item $w\in C\implies \mathbb{P}\left[A^w(i)=w_{i}\right]=1$.
    \item if $\exists c\in C$ such that $d_{U}(c,w)<\delta_{r}\implies\mathbb{P}_A \left[ A^w(i)\in\{c_{i},\bot\} \right] \geq \frac{2}{3}.$
\end{enumerate}
$\bot$ is a special abort symbol. We say $A$ has query complexity $t$ if it uses at most that many queries to perform the correction, on any inputs.
\end{definition}

\begin{remark}
    There is also a third condition that says that in the second case, there are a constant number of coordinates $i\in [n]$ for which $\mathbb{P}_{A}[A^{w}(i)=c_{i}]>\frac{2}{3}$. This follows from a transformation from \cite{BGH06} given the first two conditions and given that the algorithm requires only constant queries.
\end{remark}

The following theorem states that there exists an $\IPP$ for any subset of a relaxed locally correctable code, The proof follows by a reduction from relaxed correcting to distribution-free $\IPP$s. This builds on a result from \cite{HalKush} which states that there is a distribution-free property tester for correctable languages that are testable.

\begin{theorem}
\label{thm:rlcc_dfipp}
For any $n \in \N$, let $C \subseteq \{0,1\}^n$ be an RLCC (i.e., a binary RLCC) with a corrector $C_{\mathsf{cor}}$ having query complexity $t(n)$ and correcting radius $\delta_{r}$. Then for any language $L \subseteq C$ and every $0 < \varepsilon \leq \delta_{r}$, if $L$ has an $\IPP$ over the uniform distribution with query complexity $q(n)$, communication complexity $c(n)$ and round complexity $R(n)$ on inputs of length $n$, there exists a distribution-free $\IPP$ for $L$ with query complexity $O(q(n)+ \frac{t(n)}{\varepsilon})$, communication complexity $c(n)$ and round complexity $r(n)$.
\end{theorem}

\begin{proof}
Let $(V_0, P_0)$ be the $\IPP$ for $L$ over the uniform distribution. Recall that the corrector $C_{\mathsf{cor}}$ takes inputs $X \in \{0,1\}^{n} $ and an index $i \in [n]$, and returns the corrected value of $X$ at $i$.  Then we construct a distribution-free $\IPP$ $(V_{df}, P_0)$ for $L$ in the following way.

\begin{algorithm}
\label{pcl:rlcc_dfipp}
\floatname{algorithm}{Protocol}
\caption{Distribution-free $\IPP$ for a subset of an RLCC $C$, with corrector $C_{\mathsf{cor}}$.}
\label{RLCCDFIPP}
\begin{enumerate}
    \item $V_{df}$ runs the uniform $\IPP$ between $V_0$ and $P_0$ on the input $X \in \{0,1\}^n$. It rejects, if $V_{0}$ rejects.
    \item \label{it:4repsrlcc} Repeat O(1) times:
        \begin{enumerate}
            \item \label{it:samplerlcc} $V_{df}$ samples $\frac{1}{\varepsilon}$ points from $\sO_\mathcal{D}(X)$. Let $S$ be this set.
            \item $V_{df}$ checks if there exists $i \in S$, $C_{\mathsf{cor}}(X, i) \neq X_{i}$. If so, it rejects. 
    \end{enumerate}
\item $V_{df}$ accepts otherwise.
\end{enumerate}
\end{algorithm}

For completeness of $(V_{df}, P_0)$, if $X\in L$, by definition, the honest prover $P_0$ convinces $V_0$ to accept with probability at least $\frac{2}{3}$ and by the perfect completeness of the corrector, for every sample $i$, $C_{\mathsf{cor}}(X,i)=X_{i}$, therefore $V_{df}$ must accept with probability at least $\frac{2}{3}$.  

For soundness, suppose $X$ is $\varepsilon$-far from $L$ along $\sD$. Either we have that $X$ is $\varepsilon$-far from $L$ along both $\sD$ and $\sU$, or it is far only along $\sD$. In the first case, if $d_{\mathcal{U}}(X,L)>\varepsilon$, then $X$ is rejected by the uniform $\IPP$ with probability at least $\frac{2}{3}$. For the other case, suppose that the uniform $\IPP$ accepts $X$ since it is $\varepsilon$-close to $L$ under the uniform distribution. Let $c$ be the closest codeword to the input $X$ along $\sU$, i.e., $d_{\sU}(X,c)\leq \varepsilon$. If $V_{df}$ does not reject, either for each sampled point in each copy of $S$, $c$ coincides with $X$, or there exists an $i$ in some $S$ such that $c_i \neq X_i$, but $C_{\mathsf{cor}}$ failed in correcting $X_{i}$ to the value $c_i$.

Take any iteration of Step \ref{it:4repsrlcc}. The probability that there exists no sample $i \in S$ for which $X_{i} \neq c_{i}$ is at most $(1-\varepsilon)^{\frac{1}{\varepsilon}}\leq\frac{1}{e}$. On the other hand, the probability there exists an $i \in S$ on which the corrector fails (i.e., $C_{cor}(X,i) = X_i \neq c_{i}$) is at most $\frac{1}{3}$. By a union bound, the probability that $V_{df}$ does not reject in any iteration of this step, is at most $\frac{1}{e}+\frac{1}{3}$ and we can achieve the required soundness by $O(1)$ repetitions. 

Clearly, the communication and round complexities are unchanged as the only interactions with the prover are in the uniform $\IPP$. Moreover, the query complexity is just $q(n) + \frac{O(t(n))}{\varepsilon}$.
\end{proof}

\subsubsection{Complexity separations via correctability}

A language of interest here is the parameterised sub-tensor sum property denoted as \textsf{TensorSum}, that was defined in \cite{GR18}. This language is an example of a subset of a correctable code for which we have $\MAP$ lower and upper bounds, testing lower bounds and an $\IPP$ upper bound which means we can use this language to demonstrate separations between these complexity classes using these uniform results and our result for RLCCs. The protocols for this problem capture the sumcheck protocol which is one of the most important tools in the field of interactive proofs.
\begin{definition}[$\TensorSum_{\F,m,d,H}$]
\label{def:tensorsum}
Let $\mathbb{F}$ be a finite field and let $H \subset \mathbb{F}$. Let $P:\mathbb{F}^{m}\rightarrow \mathbb{F}$ be a polynomial of individual degree $d$. Then, $P$ belongs to $\TensorSum_{\mathbb{F},m,d,H}$ iff
\begin{equation*}
    \sum_{x\in H^{m}}P(x)=0.
\end{equation*}
\end{definition}

We now state the query complexity gaps for $\TensorSum$ which as we observed earlier, is a subset of all low-degree polynomials, that in turn is an RLCC. We use various results from \cite{GR18} to prove these gaps along with Theorem \ref{thm:rlcc_dfipp}.
 
\begin{theorem}[Query complexity gaps for distribution-free testing $\TensorSum$]
\label{thm:formal_tensorsum_gaps}

For any field $\F$, any $m,d \in \N$ such that $d < \vert \F \vert$, and any sub-field $H \subseteq \F$, the following hold true for the language $\TensorSum_{\F,m,d,H}$.
\begin{enumerate}
    \item For every $\varepsilon \in \left( 0, 1 - \frac{dm}{|\mathbb{F}|} \right)$, if $d \geq 2(|H| - 1)$, then every distribution-free $\MAP$ for $\TensorSum$ (with respect to proximity parameter $\varepsilon$) that has proof complexity $p \geq 1$ must have query complexity $q = \Omega \left( \frac{|H|^{m}}{p\log |\mathbb{F}|} \right)$. 
    \item If $dm < \frac{|\mathbb{F}|}{10}$, then, for every $\ell \in \{0, . . . , m\}$, $\TensorSum_{\mathbb{F},m,d,H}$ has a distribution-free $\MAP$ with proof complexity $(d+1)^{\ell} \log(|\mathbb{F}|)$ and query complexity $|H|^{m-\ell}(dm^{2} \log |H|)\cdot poly(1/\varepsilon).$ 
    \item If $dm < \frac{|\mathbb{F}|}{10}$, then there exists an $m$-round distribution-free $\IPP$ for $\TensorSum_{\mathbb{F},m,d,H}$ with communication complexity $O(dm \log |\mathbb{F}|)$, and query complexity $O(dm \cdot \poly(1/\varepsilon))$.
\end{enumerate}
\end{theorem}

The property testing lower bound here follows from the $\MAP$ lower bound from \textbf{Lemma 3.15} in \cite{GR18} and from the fact that distribution-free testing is a more general setting which encompasses uniform testing and so will require at least as many queries.

The $\MAP$ and $\IPP$ upper bound come from \textbf{Lemma 3.14} and \textbf{Theorem 3.22} respectively from \cite{GR18}.  They also require Theorem \ref{thm:rlcc_dfipp} to reduce to the uniform setting.

\noindent Theorem \ref{thm:informal_tensorsum_gaps} follows from this. 

\begin{proof}[Proof Sketch of Theorem \ref{thm:informal_tensorsum_gaps}]
We sketch the proof of the three claims as follows.
\begin{enumerate}
    \item $\Theta(n^{0.999\pm o(1)})$ distribution-free query complexity for testing $\TensorSum$:
    
    This holds from Item 1 of Theorem \ref{thm:formal_tensorsum_gaps}, by setting $p = 1$.
    
    \item $\TensorSum$ df-$\MAP$ query and communication complexity $\Theta(n^{0.499\pm o(1)})$ and for proof complexity $p\geq 1$, query complexity $\Theta(\frac{n^{0.999\pm o(1)}}{p})$:
    
    For $d=\Theta(\vert H \vert)=n^{o(1)}$, $m=\log_{|H|}(n)$, $l=\log_{d+1}(\frac{p}{\log |\mathbb{F}|})$, and $\varepsilon=O(1)$, we obtain the upper bound for query complexity of the $\MAP$ in terms of the proof length, from Item 2 of Theorem \ref{thm:formal_tensorsum_gaps} above. The corresponding lower bound is achieved by setting $|\mathbb{F}|=\sqrt{n}$ and $|H|^{m}=n$, in Item 1 of Theorem \ref{thm:formal_tensorsum_gaps}. 

    In particular, setting $p = O(\sqrt{n})$ optimises the sum of proof and query complexities, thus proving Item 2.
    
    \item Distribution-free $\IPP$ with query and communication complexity $\polylog(n)$:
    
    This $\IPP$ follows from Item 3 of Theorem \ref{thm:formal_tensorsum_gaps}, when $|\mathbb{F}|=\sqrt{n}$, $m=\log(n)$ and $d=2$.
\end{enumerate}    
\end{proof}

\subsection{Symmetric Languages}
\label{sec:sym}

In this section, we show an $\Omega(n^{1/3-0.0005})$ query complexity lower bound for any distribution-free testers for $\HAM(w(n))$, for some fixed $w(n)$. Following this, we construct a distribution-free $\IPP$ with $\polylog(n)$ query and communication complexity for any symmetric language. In fact, this is shown by a straight-forward interactive reduction to $\HAM$, and constructing a distribution-free $\IPP$ for this problem. Put together, we prove Theorem \ref{thm:informal_ham_sep}, thus demonstrating an exponential advantage in the query complexity given interaction with a prover in the distribution-free setting. We define the problem as follows.

\begin{definition}[The Hamming weight language]
For any $X\in \{0,1\}^{n}$, let $\Hamming(X)$ be the Hamming weight (the number of non-zero entries) of $X$. 

Let $w : \N \rightarrow \N$ be a weight function such that for any $n \in \N$, $1 \leq w(n) \leq n$. We define the language $\HAM(w) = \{\HAM_n (w(n))\}_{n \in \mathbb{N}}$, such that for any $X\in \{0,1\}^{n}$, $X\in \HAM_n(w(n)) \iff \Hamming(X)=w(n)$.
\end{definition}

The (parameterised) Hamming weight language is one of a general class of languages that are \textit{symmetric}, i.e., those which depend only on the Hamming weight of the input string.
\begin{definition}[Symmetric Languages]
A language $L = \{L_n\}$ is called Symmetric if and only if for every $n \in \N$, there exists a predicate $\mathcal{S}_n : \{0, \dots, n\} \rightarrow \{0,1\}$ such that
\begin{equation*}
    L_n = \{ X \in \{0,1\}^n \mid \mathcal{S}_n(\Hamming(X)) = 1 \}
\end{equation*}
\end{definition}

\subsubsection{Testing Lower Bound}
\begin{theorem}
\label{thm:hamming_lb}
For any $\varepsilon>0$, $n\in \mathbb{N}$, there exists $w = w(n) > 0$ such that any distribution-free property tester for $\HAM(w)$ requires $\Omega(n^{1/3-0.0005})$ queries.
\end{theorem}

We prove this theorem via the following steps:
\begin{enumerate}
    \item We construct two pairs $(\mathcal{D}_{1}, X),(\mathcal{D}_{2}, Y)$ of a distribution and an input to $\HAM(w)$, where $w = \Hamming(Y)$. 
    The first pair is a NO instance, i.e., $d_{\mathcal{D}_{1}}(X, \HAM(w))>\varepsilon$ and a tester should reject this. On the other hand, the second pair is a YES instance, where $Y$ must be accepted irrespective of the access to the sample oracle with respect to any distribution. 
    
    \item We next show that the inputs $X,Y$ are close along $\sU$. In other words, with high probability, $o(n^{1/3-0.0005})$ queries made by the tester along the uniform distribution do not help it distinguish between $X$ and $Y$.
    
    \item Further, we show that $o(n^{1/3-0.0005})$ values in $X$ along samples from $\mathcal{D}_{1}$ and those in $Y$ along samples from $\mathcal{D}_{2}$ will be distributed in the same way $(\mathbb{P}_{i\sim \mathcal{D}_{1}}\left [X_{i}=1\right]=\mathbb{P}_{i\sim \mathcal{D}_{2}}\left [Y_{i}=1\right])$. 
    
    By our construction, with high probability, there are no collisions between the samples from $\sU$ and $\sD_1$, or $\sU$ and $\sD_2$. This means that the tester can't distinguish between the 2 distribution-input pairs by sampling along any of these distributions, with high probability. 
    
    \item To conclude, we show a transformation from \textit{any query-optimal tester} for $\HAM$, to one that only uses uniformly sampled queries to the input, along with samples from the underlying distribution oracle (i.e., along $\sD_{1}$ for testing $X$ or $\sD_{2}$ for testing $Y$). 
    Putting this together with the fact that the uniform sampled queries are distinct (with high probability) from the samples along $\sD_1$ or $\sD_2$, we get query lower bound of $\Omega(n^{1/3-0.0005})$ for testing $\HAM$.
\end{enumerate}
    
We first define the pairs $(\mathcal{D}_{1}, X),(\mathcal{D}_{2}, Y)$. Consider the partition of $[n]$ into the following 3 intervals. 
\[
\begin{split}
I_{1}&=[1,n-n^{\frac{2}{3}}-n^{\frac{2}{3}-0.001}].\\ I_{2}&=[n-n^{\frac{2}{3}}-n^{\frac{2}{3}-0.001}+1, n-n^{\frac{2}{3}-0.001}].\\ I_{3}&=[n-n^{\frac{2}{3}-0.001}+1,n].
\end{split}
\]

The distributions $\mathcal{D}_{1}, \mathcal{D}_{2}\in \Delta(\Omega_{n})$ are defined as follows.
\begin{center}
\begin{tabular}{ c c }
 $\forall i_{0} \in I_{1}: \underset{i\sim \mathcal{D}_{1}}{\mathbb{P}}\left [i=i_{0}\right ]=\frac{1-20\varepsilon}{|I_{1}|}.$ & $\forall i_{0} \in I_{1}: \underset{i\sim \mathcal{D}_{2}}{\mathbb{P}}\left [i=i_{0}\right ]=\frac{1-20\varepsilon}{|I_{1}|}.$ \\
 $\forall i_{0} \in I_{2}: \underset{i\sim \mathcal{D}_{1}}{\mathbb{P}}\left [i=i_{0}\right ]=\frac{12\varepsilon}{|I_{2}|}.$ & $\forall i_{0} \in I_{2}: \underset{i\sim \mathcal{D}_{2}}{\mathbb{P}}\left [i=i_{0}\right ]=\frac{8\varepsilon}{|I_{2}|}.$ \\ 
 $\forall i_{0} \in I_{3}: \underset{i\sim \mathcal{D}_{1}}{\mathbb{P}}\left [i=i_{0}\right ]=\frac{8\varepsilon}{|I_{3}|}.$ & $\forall i_{0} \in I_{3}: \underset{i\sim \mathcal{D}_{2}}{\mathbb{P}}\left [i=i_{0}\right ]=\frac{12\varepsilon}{|I_{3}|}.$ 
\end{tabular}
\end{center}

Finally, we define $X, Y\in\{0,1\}^{n}$ as having a fixed proportion of bits in each interval set to $1$, as follows.
\begin{itemize}
    \item $\forall i\in I_{1}$, $X_{i}=Y_{i}=1$.
    \item For $I_{2}$:
    \begin{itemize}
        \item $\Hamming(X\vert_{I_{2}})=\frac{|I_{2}|}{3}$, i.e. $X$ takes value $1$ for $\frac{1}{3}$ of these indices. 
        
        In particular, $\forall i \in [n-n^{\frac{2}{3}}-n^{\frac{2}{3}-0.001}+1, n-\frac{2}{3}n^{\frac{2}{3}}-n^{\frac{2}{3}-0.001}+1]$, $X_{i}=1$, for all other $i\in I_{2}$, $X_{i}=0$.
        
        \item $\Hamming(Y\vert_{I_{2}})=\frac{|I_{2}|}{2}$, i.e. $Y$ takes value $1$ for $\frac{1}{2}$ of these indices. 
        
        In particular, $\forall i \in [n-n^{\frac{2}{3}}-n^{\frac{2}{3}-0.001}+1, n-\frac{1}{2}n^{\frac{2}{3}}-n^{\frac{2}{3}-0.001}+1]$, $Y_{i}=1$, for all other $i\in I_{2}$, $Y_{i}=0$.
    \end{itemize} and 
    \item For $I_{3}$:
    \begin{itemize}
        \item $\Hamming(X\vert_{I_{3}})=\frac{|I_{3}|}{2}$, $X$ takes value $1$ for $\frac{1}{2}$ of these indices
        
        That is, $\forall i \in [n-n^{\frac{2}{3}-0.001}+1, n-\frac{1}{2}n^{\frac{2}{3}-0.001}+1]$, $X_{i}=1$, for all other $i\in I_{3}$, $X_{i}=0$.
        
        \item $\Hamming(Y\vert_{I_{3}})=\frac{|I_{3}|}{3}$, $Y$ takes value $1$ for $\frac{1}{3}$ of these indices.
        
        That is, $\forall i \in [n-n^{\frac{2}{3}-0.001}+1, n-\frac{2}{3}n^{\frac{2}{3}-0.001}+1]$, $Y_{i}=1$, for all other $i\in I_{3}$, $Y_{i}=0$.
    \end{itemize} 
\end{itemize}

We first show that with $n^{\frac{1}{3}-0.0005}$ uniformly sampled indices, we expect to only receive elements of $I_{1}$ on which $X$ and $Y$ are identical. 
\begin{lemma}\label{usamps}
With high probability, $t = o(n^{\frac{1}{3}})$ uniform samples will only return indices $i$ for which $X_{i}=Y_{i}=1$. In other words:

\begin{equation}
    \mathbb{P}_{i\sim \mathcal{U}^{t}}\left [\forall j\in [t]:X_{i_{j}}=Y_{i_{j}}=1\right ]=1-o(1)
\end{equation}
\end{lemma}

\begin{proof}
The probability of each of these samples having this property is 
\[
\begin{split}
    \mathbb{P}_{i\sim \sU_{n}^{t}}\left [\forall j\in [t]: X_{i_j}=Y_{i_j}=1\right ] &\geq (\mathbb{P}_{i\sim \sU_{n}}\left [i_j\in I_{1}\right ])^{t}\\
    &=\bigg(\frac{n-2n^{-\frac{2}{3}}}{n}\bigg)^{t}\\
    &=\bigg(1-\frac{2}{n^{\frac{1}{3}}}\bigg)^{t}\\
    &\geq 1-o(1)\\
\end{split}
\]
\end{proof}

Then we show that sampling along the corresponding distributions returns an index on which the probability that the input is $1$ is the same for both instances. 
\begin{lemma}\label{dsamps}
$\mathbb{P}_{i\sim \mathcal{D}_{1}}\left [X_{i}=1\right ]=\mathbb{P}_{i\sim \mathcal{D}_{2}}\left [Y_{i}=1\right ]$
\end{lemma}

\begin{proof}
We evaluate both probabilities by summing over each interval as follows
\[
\begin{split}
\underset{i\sim \mathcal{D}_{1}}{\mathbb{P}}\left [X_{i}=1\right ]&=\sum_{j\in\{1,2,3\}}\underset{i\sim \mathcal{D}_{1}}{\mathbb{P}}\left [X_{i}=1\vert i\in I_{j}\right ]\underset{i\sim \mathcal{D}_{1}}{\mathbb{P}}\left [i\in I_{j}\right ]\\
&=\frac{1-20\varepsilon}{|I_{1}|}|I_{1}|+\frac{12\varepsilon}{|I_{2}|}|I_{2}|\frac{1}{3}+\frac{8\varepsilon}{|I_{3}|}|I_{3}|\frac{1}{2}\\
&=1-12\varepsilon
\end{split}
\]

\[
\begin{split}
\underset{i\sim \mathcal{D}_{2}}{\mathbb{P}}\left [Y_{i}=1\right ]&=\sum_{j\in\{1,2,3\}}\underset{i\sim \mathcal{D}_{2}}{\mathbb{P}}\left [Y_{i}=1\vert i\in I_{j}\right ]\underset{i\sim \mathcal{D}_{2}}{\mathbb{P}}\left [i\in I_{j}\right]\\
&=\frac{1-20\varepsilon}{|I_{1}|}|I_{1}|+\frac{8\varepsilon}{|I_{2}|}|I_{2}|\frac{1}{2}+\frac{12\varepsilon}{|I_{3}|}|I_{3}|\frac{1}{3}\\
&=1-12\varepsilon.
\end{split}
\]
\end{proof}

We then show that for each pair, the probability of sampling the same element twice from the distribution oracle and from uniform distribution is very small. The following notation is useful. For any $t\in \N$, by $\sU^t \times \sD^t$ we denote the $2t$-length tuple of indices that consists of $t$ i.i.d. samples from $\sU$ and $t$ i.i.d. samples from $\sD$, drawn independently of each other.

\begin{lemma}\label{colls}
For each $j\in \{1,2\}$, sampling $t=o(n^{\frac{1}{3}-0.0005})$ times along $\mathcal{D}_{j}$ or along $\sU$ results in distinct indices(no collisions) with high probability. In other words:
\begin{equation}
    j\in\{1,2\}\implies\mathbb{P}_{(i_{1},\cdots, i_{2t})\sim\sU^{t}\times\mathcal{D}_{j}^{t}} \left[ \exists \ell,k\in [2t]:\ell\neq k\wedge i_{\ell}=i_{k} \right] \geq 1-o(1)
\end{equation}
\end{lemma}

\begin{proof}
We prove this statement for $\sD_1$ and the other case is analogous to this. For any $m>0$, $o(\sqrt{m})$ uniform samples from $[m]$ will be distinct with probability at least $1-o(1)$ (this follows from the birthday paradox). Each of $\sqrt{|I_{1}|}$, $\sqrt{|I_{2}|}$, and $\sqrt{|I_{3}|}=\Omega(n^{\frac{1}{3}-0.0005})$ and therefore with high probability, samples from $\sD_1$ within any of these intervals will not produce a collision. This is because restricted to each of the intervals $I_{1}, I_{2}, I_{3}$, $\mathcal{D}_{1}$ is uniform and therefore even if all the samples were just on one of these intervals, then with high probability all of the indices sampled will be distinct. 

Similarly, with high probability the $t$ samples along the uniform distribution over $[n]$ won't collide with each other. Additionally the probability that they don't collide with the samples from $\sD_1$ is at least $\left (1-\frac{o(n^{\frac{1}{3}})}{n}\right)^{o(n^{\frac{1}{3}})} \geq 1-o(1)$.
\end{proof}

Finally, we need the following result which shows that the only relevant queries for property testing $\HAM$, are those that are uniformly sampled. Intuitively, this holds since $\HAM$ is a symmetric language: we can permute the indices of the input and the Hamming weight remains unchanged. 

\begin{lemma}\label{lem:PTSamplingHam}
For any $w\in \mathbb{N}$ and any distribution-free testing algorithm $\sA$ for $\HAM(w)$ with query complexity $T$ and sample complexity $S$ such that $T + S = o(\sqrt{n})$, there exists an equivalent distribution-free testing algorithm $\sA'$ with $O(T)$ queries and $O(S)$ samples, such that $\sA'$ only makes input queries along indices sampled from $\mathcal{U}$.
\end{lemma}

\begin{proof}
Fix some $w \in \N$ and a distribution-free testing algorithm $\mathcal{A}$ for $\HAM(w)$ using $T$ queries to $X$, some of which are possibly adaptive. Without loss of generality, suppose that $\sA$ gets all its labelled samples from $\sO_\sD$ at the beginning. Following this, we know that $\sA$ must have the same output (with high probability) for any pair of inputs which are consistent on these indices sampled from $\sD$. 

Let $i$ be the last query that $\sA$ makes, which is not a uniformly sampled index. Let $I$ be the set of indices already queried and $J$ be the set of indices sampled from $\mathcal{D}$. Let $j$ be a uniformly random sample from $[n]$. In particular, $j \notin I \cup J$ with probability $\frac{n-o(\sqrt{n})}{n}=1-o(\frac{1}{\sqrt{n}})$.

For any $X\in\{0,1\}^{n}$, define $X^{(i,j)}$ as an input string with the values of $X$ swapped at indices $i$ and $j$, and similarly, $\mathcal{D}^{(i,j)}$ where the probabilities of sampling $i$ and $j$ are swapped. Let $\hat{\sA}$ be the algorithm that follows the same queries as $\mathcal{A}$, however querying $j$ instead of $i$, and whose outcome is the same as $\mathcal{A}$ on the input $X^{(i,j)}$ (i.e., $\sA$ gets the value $X_j$ when it queries $i$). This follows as long as $j\not\in I\cup J$, which as we saw earlier, happens with high probability.

Next, we have the following fact for any $\mathcal{D}\in \Delta([n])$ and any $X \in\{0,1\}^{n}$: 
\begin{equation}
\label{eq:ij_swap}
    \mathbb{P}_{\mathcal{D}, \hat{\mathcal{A}}}\left[\mathcal{\hat{\sA}}(X)=1\right]\geq\mathbb{P}_{\mathcal{D}^{(i,j)},\mathcal{A}}\left[\mathcal{A}(X^{(i,j)})=1\right]-o\left(\frac{1}{\sqrt{n}}\right).
\end{equation}

Indeed, if $\mathcal{A}$ accepts $X^{(i,j)}$ when given a set of samples from $\sD^{(i,j)}$ and coin flip outcomes, $\hat{\mathcal{A}}$ must accept $X$ when sampling from $\sD$ and having the same coin flip outcomes so long as no collision occurs which is accounted for by the $o \left( \frac{1}{\sqrt{n}} \right)$ term . In other words,

\[
\begin{split}
    \mathbb{P}_{\mathcal{D}, \hat{\mathcal{A}}}\left[\mathcal{\hat{\sA}}(X)=1\right] &\geq \mathbb{P}_{\mathcal{D}, \hat{\mathcal{A}}}\left[\mathcal{\hat{\sA}}(X)=1 \wedge j\not\in I\cup J\right]\\
    &\geq\mathbb{P}_{\mathcal{D}^{(i,j)},\mathcal{A}}\left[\mathcal{A}(X^{(i,j)})=1\right]-o\left(\frac{1}{\sqrt{n}}\right).
\end{split}
\]

For each pair $(X^{(i,j)},\mathcal{D}^{(i,j)})$, $\mathcal{A}$ returns the correct result with probability at least $\frac{2}{3}$. Therefore, from Equation \ref{eq:ij_swap}, $\hat{\sA}$ also returns the correct value on $(X,\mathcal{D})$ with the same probability up to an additive factor of $o\left(\frac{1}{\sqrt{n}}\right)$. In other words, repeating this at most $T$ times for all arbitrary queries, results in a new tester $\sA'$. This must be a distribution-free property tester for $\HAM$ using $T$ queries which are all uniformly sampled and $S$ samples from $\sD$, with success probability at least $\frac{2}{3}-o\left(\frac{T}{\sqrt{n}}\right)=\frac{2}{3}-o(1)$. By repeating this tester $O(1)$ times, we can recover the original $\frac{2}{3}$ success probability.
\end{proof}

\noindent Given these lemmas, we now prove the main result of this section.
\begin{proof}[Proof of Theorem \ref{thm:hamming_lb}]
Let $\Hamming(Y) = w$. By Lemmas \ref{usamps}, \ref{dsamps} and \ref{colls} we know that the labeled samples along $\sD_{1}$ for $X$, or $\sD_{2}$ for $Y$, along with queries made along $U$ are not enough to distinguish whether $X$ is the input tested against distribution $\mathcal{D}_{1}$, or $Y$ is the input tested against distribution $\mathcal{D}_{2}$. Due to Lemma \ref{lem:PTSamplingHam}, we know that any distribution-free property tester can only use those queries to distinguish the two cases. 

Therefore, it suffices to show that $X$ is $\varepsilon$-far from $\HAM(w)$ along $\mathcal{D}_{1}$ as then it becomes impossible to distinguish $(X, \mathcal{D}_{1})$ from $(Y, \mathcal{D}_{2})$ with $o(n^{\frac{1}{3}-0.0005})$ queries to the input. The Hamming distance between $X$ and $Y$ is $\frac{1}{6}(|I_{2}|-|I_{3}|)=\frac{n^{2/3}}{6}(1-n^{-0.001})$. This is the number of bits of $X$ that need to be flipped from $0$ to $1$ to have Hamming weight same as $Y$. This can be done by flipping either the elements of $I_{2}$ or $I_{3}$ (note that $X$ and $Y$ have the same values in the set $I_{1}$). Furthermore, along $\sD_{1}$, the weight of elements in $I_{2}$ is less than $I_{3}$; so in total the distance $d_{\sD_{1}}(X, w)=\frac{n^{2/3}}{6}(1-n^{-0.001}) \frac{12\varepsilon}{|I_{2}|}>\varepsilon$.
\end{proof}
\subsubsection{Upper Bound for Symmetric Languages}
\label{sec:hamming}
We next present a distribution-free $\IPP$ for $\HAM(w(n))$, for every weight function $w(n) \leq n$. 
\begin{theorem}
\label{thm:ham_dfipp}
For every $n \geq 2$ and every $w \leq n$, there exists a distribution-free $\IPP$ for $\HAM(w)$ with perfect completeness and soundness $\frac{1}{3}$. This protocol has round complexity $O\left(\frac{2\log(n)}{\varepsilon}\right)$, communication complexity $O \left( \frac{\log(n)^{2}}{\varepsilon} \right)$ and sample complexity $\frac{1}{\varepsilon}$. Furthermore, the $\IPP$ does not make any queries to the input and only uses samples from the distribution.
\end{theorem}
 
There are $O\left(\frac{1}{\varepsilon}\right)$ rounds in this $\IPP$, in each round the verifier samples an index $i\sim \mathcal{D}$ and then performs a binary split across the indices in [n] as follows. Set $I=[n]$, $I_{0}=[\lfloor n/2\rfloor ]$ and $I_{1}=[\lfloor n/2 \rfloor + 1,n]$, the prover sends the Hamming weight of both $I_{0}$ and $I_{1}$. The protocol iterates by updating $I$ to be set to whichever of $I_{0}$, $I_{1}$ contains $i$ and iterate thereon. We iterate until $I=\{i\}$, the verifier queries $X_{i}$. 

\begin{protocol}
\caption{Interactive Proof of $\varepsilon$-proximity  for $\HAM(w)$}
\label{pcl:hamub}
\begin{enumerate}
    \item Repeat $O\left(\frac{1}{\varepsilon}\right)$ times:\label{erep}
\begin{enumerate}
    \item The verifier samples $(i,X_i)$ from $\sO_\sD$. Let $i=(i_{0}, \cdots i_{\lceil \log(n)-1\rceil})_{2}$.
    \item Set $I=[n]$, $v=w$, $L=0$, $U=n$, $r=0$ and $s= \emptyset$.
    \item Repeat the following until $\vert I\vert =1$. \label{binSplit}
    \begin{enumerate}
        \item The verifier partitions $I$ into two parts $I_{0}=\left[ L,\lfloor\frac{U+L}{2}\rfloor \right], I_{1}=\left(\lfloor\frac{U+L}{2}\rfloor , U\right]$.
        \item\label{hSent} Let the Hamming weight of $X$ over $I_{0}, I_{1}$ is $h_{s0}, h_{s1}$ respectively. The prover sends $h'_{s0},h'_{s1}$ which is the purported value of $h_{s0}$ and $h_{s1}$ respectively.
        \item \label{ConsCheck}The verifier rejects if $h'_{s0}+h'_{s1}\neq v$, or if either of $h'_{s0} \notin [0,\vert I_0 \vert]$ or $h'_{s1} \notin [0,\vert I_1 \vert]$.
        \item The verifier sends $i_{r}$ to the prover. If $i_{r}=0$ then reassign $I=I_{0}$, $U=\lfloor\frac{U+L}{2}\rfloor$,  $s=s0$, $v=h'_{s0}$ otherwise reassign $I=I_{1}$, $L=\lfloor\frac{U+L}{2}\rfloor+1$, $s=s1$, $v=h'_{s1}$. 
        
        In both cases set $r=r+1$.
    \end{enumerate}
    \item \label{samples} Finally, $X_{i}\neq v$, the verifier rejects.
\end{enumerate}
\item The verifier accepts otherwise.
\end{enumerate}
\end{protocol}

\begin{proof}
Fix some $w \leq n$. Let $X \in \{0,1\}^n$ be the input to $\HAM(2)$. For any $i \in [n]$, we write $i=(i_{0}, \cdots i_{\lceil \log(n-1)\rceil})_{2}$ as its binary expansion. The full $\IPP$ for $\HAM(w)$ is given in Protocol \ref{pcl:hamub}. Clearly, its sample complexity is $\frac{2}{\varepsilon}$ and it makes no additional queries to $X$. Below, we prove its correctness.

\textbf{Completeness}: If $X\in \HAM(w)$, completeness follows from the fact that the honest prover will only provide the genuine values for $h_{s0}, h_{s1}$ in each iteration of the binary search. At the end of each of the $O\left(\frac{1}{\varepsilon}\right)$ rounds, the final value in $v$ is the Hamming weight of $X$ on the interval $I=\{i\}$ which is $X_{i}$, therefore the verifier will not reject.

\textbf{Soundness}: Suppose $d_{\mathcal{D}}(X,\HAM(w))>\varepsilon$. For each round of Step \ref{erep}, we show that the probability that the verifier catches the prover is at least $\varepsilon$. Let $Y=h'_{00\cdots0}h'_{00\cdots1}\cdots h'_{11\cdots1}$ be the $n$-length Boolean string concatenating all possible values of $h'_{i}$ that the prover sends the verifier in the final iteration of the binary search process (here the $i$ in $h'_{i}$ is represented in binary). 

It is sufficient to prove that an optimal dishonest prover strategy (which is rejected with the lowest probability at the latest possible point) will have the property that $\Hamming(Y)=w$. By our assumption, we have that $d_{\mathcal{D}}(X,Y)\geq d_{\mathcal{D}}(X,\HAM(w))>\varepsilon$. Thus, with probability at least $\varepsilon$ over the choice of the sample $i$ from $\sD$, the verifier rejects in a single round. In other words, repeating $O\left(\frac{1}{\varepsilon}\right)$ times will achieve the desired soundness probability.

\begin{claim}
For all values of $n\geq 2$, if we assume in the soundness case that the cheating prover strategy maximises the probability the verifier accepts, $\Hamming(Y)=w$.
\end{claim}
\begin{proof}
We proceed by induction on the length of the binary representation of $n$, denoted by $R =\lceil \log(n)\rceil$.

For the base case, suppose $R=\lceil \log(n)\rceil=1 \iff n=2$. Here the input length $n =2$ and thus, $ w \leq 2$. In the one round of the binary split the prover sends $h'_{0}, h'_{1}$ for which $Y=h'_{0}h'_{1}$ so that the value of $\Hamming(Y)$ has to be equal to $w$ as if the prover sends values such that $h_{0}'+h_{1}'\neq w$, the verifier will reject before it checks its sample in step \ref{samples} with probability $1$, so the optimal prover strategy will send values such that $\Hamming(Y)=w$.

For some $k \in \N$, suppose that $\Hamming(Y)=w$ for $R=k-1$, let us take $R=k$. In the first round of the binary split: we have that $h'_{0}+h'_{1}=w$, as otherwise the verifier would immediately reject with probability $1$, let $Y[0]$ be the restriction of $Y$ to $I_{0}$, define $Y[1]$ similarly. 

By the inductive assumption as $\lfloor \log(\vert I_{0}\vert )\rfloor=k-1$, we have $\Hamming(Y[0]) =h'_{0}$, and similarly $\Hamming(Y[1]) =h'_{1}$. 
Therefore, $\Hamming(Y)=\Hamming(Y[0])+\Hamming(Y[1])=h'_{0}+h'_{1}=w$.
\end{proof}

\textbf{Round complexity}: This is the number of times the prover sends $h_{s0}, h_{s1}$ which is $O\left(\frac{\log (n)}{\varepsilon}\right)$. 

\vspace{0.1in}
\textbf{Communication complexity}: The prover sends the values of $h_{s0}$ and $h_{s1}$ at each iteration of the binary search. This happens $\lceil \log (n) \rceil$ many times for each sample, and there are $O\left(\frac{\log (n)}{\varepsilon}\right)$ samples chosen. The verifier only sends the prover $O\left(\frac{\log (n)}{\varepsilon}\right)$ bits of information over the course of the $\IPP$. Therefore in total, the communication complexity is $O \left( \frac{\log(n)^{2}}{\varepsilon} \right)$.
\end{proof}

\begin{remark}
     It is worth noting that the distribution-free $\IPP$ from Theorem \ref{thm:ham_dfipp} offers a quadratic improvement in the total ``access" complexity (sample complexity and query complexity) over both its distribution-free tester \cite{CANETTI199517}, as well as what we get by using the generic distribution-free $\IPP$ from Theorem \ref{thm:dispersed_ipp_nc} (since any symmetric language is computable in $\NC^1$). 
\end{remark}

Finally, we observe that the distribution-free $\IPP$ for $\HAM(w)$ easily extends to any symmetric language in the following corollary.
\begin{corollary}
\label{cor:sym_dfipp}
Any symmetric language $L$ has a distribution-free $\IPP$ with round complexity $O\left(\frac{\log (n)}{\varepsilon}\right)$, communication complexity $O \left( \frac{\log(n)^{2}}{\varepsilon} \right)$ and sample complexity $O\left(\frac{1}{\varepsilon}\right)$.
\end{corollary}

\begin{proof}
For any input $X \in \{0,1\}^n$, this is clear from the fact that by having the prover send the $w' = \Hamming(X)$, we reduce this problem to an instance of $\HAM(w')$, so that the completeness and soundness conditions follow from its distribution-free IPP. The round complexity only increases by $1$ and the sample complexity is unchanged but the communication complexity increases by $\log (n)$ for the number of bits sent to specify $w'$. This leaves the communication complexity unchanged at $O(\frac{\log^{2} (n)}{\varepsilon})$.
\end{proof}

\subsection{Separation between $\IPP$s and Distribution-Free $\IPP$s}\label{sec:dfUseparation}

In this section, we demonstrate that not all languages with a uniform $\IPP$ have a distribution-free $\IPP$ with the same query and communication complexity. We do so by showing the existence of a language $L$ that has an efficient $\IPP$ but for which any distribution-free $\IPP$ requires a very large complexity. As a matter of fact, the result is even stronger as the $\IPP$ for $L$ is essentially just a property tester - that is, the prover is not required.

The proof is rather straightforward - we construct a language $L$ in which a very small portion of the input consists of a very hard property. A standard property tester can safely ignore this part of the input (since it is small), but for a distribution-free $\IPP$, the distribution could be entirely concentrated on this small portion.

\begin{proposition}\label{prop:Sep_DFIPP_U}
    Suppose that there exists $\varepsilon>0$ s.t. for all $q=q(n)\leq n$, there exists a language $L$ and a function $\ell(n)$ s.t.  for any uniform $\IPP$ (respectively uniform $\MAP$) for $L$, with proximity parameter $\varepsilon$, query complexity $q=q(n)$, and communication complexity $c=c(n)$, $\max(q,c) \geq \ell(n)$, then the following is true.

    For all $q=q(n)<\frac{\eps}{2} \cdot n$, there exists a language $L'=\bigcup_{n\in\mathbb{N}} L'_{n}$ which has the following properties:
    
    \begin{enumerate}
        \item The property testing query complexity for $L'$ with proximity parameter $\varepsilon$ is $O(1/\varepsilon)$.
        \item For any distribution-free $\IPP$ (respectively $\MAP$) for $L'$ with proximity parameter $\varepsilon$, query complexity $q(n)$, and communication complexity $c(n)$, it holds that $\max(q(n),c(n))\geq \ell(n')$, where $n' = (\eps/2) \cdot n$.
    \end{enumerate} 
\end{proposition}

Proposition \ref{prop:seps} follows from Proposition \ref{prop:Sep_DFIPP_U} using known results. Specifically, \cite[Theorem 4]{KR15} shows a language for which the conditions of Proposition \ref{prop:Sep_DFIPP_U} hold for $\ell(n)=\Omega(\sqrt{n})$ (assuming the PRGs from the hypothesis of Proposition~\ref{prop:seps}). The second item of Proposition~\ref{prop:seps} follows from  \cite[Theorem 4]{GR18}, which shows there exists a language with no efficient $\MAP$ satisfying the criterion for $\ell(n)=\Omega(n)$. 

\begin{proof}[Proof of Proposition \ref{prop:Sep_DFIPP_U}]
    We prove the theorem w.r.t. to $\IPP$s (i.e., where both the assumption and conclusion are for $\IPP$s). The result for $\MAP$s is proved in exactly the same way.

    Suppose there exists $\varepsilon>0$ as in the statement of Proposition \ref{prop:Sep_DFIPP_U}. Let $q=q(n) \leq \frac{\varepsilon}{2}\cdot n$ and $L=\bigcup_{n\in\mathbb{N}} L_{n}$ be a language such that every $\IPP$ for $L$ with wrt proximity parameter $\varepsilon$ with query complexity $q$ has communication complexity $c$ such that $\max(q,c) \geq \ell(n)$.

    We construct a language $L'_{n}$ as follows.
    \begin{equation*}
        L'_{n}= \big\{(0^{(1-\varepsilon/2) \cdot n},x) \;\vert\; x\in L_{(\varepsilon/2) \cdot n} \big\}
    \end{equation*}
    Thus, inputs to $L'_n$ consists of an input of length $(\varepsilon/2) \cdot n$ of $L$, which is preceded by $(1-\varepsilon/2) \cdot n$ zeros. The high level idea is that since the vast majority of the input is always $0$, the language $L'$ is easy to test. However, in the distribution-free setting, we can concentrate all of the ``weight'' of the distribution on the suffix. Details follow.

    Indeed, in the uniform case, there is a trivial property tester for $L'$ that queries the input $O(1/\varepsilon)$ times sampled uniformly over the first $\left(1-\frac{\varepsilon}{2}\right) \cdot n$ bits and accepts if and only if all the queries return $0$.
    For any $I\subseteq [n]$, we denote by $X\vert_{I}$ the substring of $X$ restricted to $I$.  Completeness follows as the tester will accept for $X\in L'$ as all queries return $0$ by the definition of $L'$. Soundness follows as for $d_{\sU}(X,L')>\varepsilon$, by the triangle inequality, $d_{\sU}(X\vert_{\left[n\cdot\left(1-\frac{\varepsilon}{2}\right)\right]},0^{n-(\varepsilon/2)n})>\varepsilon-\varepsilon/2\geq\varepsilon/2$ and the tester will query a non-zero entry with high probability after $O(1/\varepsilon)$ queries.

    We proceed to the distribution-free setting. Let $\sD$ be the distribution which is uniform over the last $(\eps/2) \cdot n$ bits. 
    Suppose there exists an $\IPP$ for $L'$ with query complexity $q$ and communication complexity $c$ along $\sD$ with proximity parameter $\varepsilon$. This $\IPP$ immediately yields yields a uniform $\IPP$ for $L$ for input size $\frac{\varepsilon}{2}\cdot n$ and proximity parameter $\varepsilon$. This follows as the distance of $X$ from $L'$ is now the uniform distance from $X$ restricted to $\left[n-\frac{\varepsilon}{2}\cdot n+1, n\right]$ to $L$. In other words,
    \begin{equation*}
        d_{\sD}(X,L')=d_{\sU}(X\vert_{\left[n-\frac{\varepsilon}{2}\cdot n+1, n\right]},L).
    \end{equation*}
    
    Thus, for any distribution-free $\IPP$ for $L'_{n}$ with query complexity $q(n)$ and communication complexity $c(n)$, there is a uniform $\IPP$ for inputs of length $n' = \frac{\varepsilon}{2}\cdot n$ for $L$ with query complexity $q'(n')=q(n)$ and communication complexity $c'(n')=c(n)$. By the definition of the language $L$, $\max(q(n), c(n))\geq \ell(n')$ by the property of uniform $\IPP$s for $L$.
\end{proof}

\section{Distribution-Free $\IPP$s for $\NC$}
\label{sec:dfipp_low-depth}
Recall that the class $\NC$ consists of languages computable by a sequence of Boolean circuits $\{C_n\}_{n \in \mathbb{N}}$ of polynomial size and poly-logarithmic depth. In this section, we construct a distribution-free $\IPP$ for any language in logspace-uniform $\NC$ that uses $O(1/\varepsilon)$ queries, as well as $O(1/\varepsilon)$ samples from the input, and $O\left(n\cdot\varepsilon\cdot \polylog(n) + \frac{\log(n)}{\varepsilon}\right)$ bits of communication, on inputs of length $n$ and proximity parameter $\varepsilon$. 

Our approach extends the strategy in \cite{RVW} to the more involved setting of distribution-free testing. In Section \ref{sec:NCPVAL}, we show that the polynomial evaluation problem $\PVAL$, defined in Section \ref{sec:pval_def}, is not only `complete' for constructing uniform $\IPP$s for languages computable by low-depth circuits, but also for constructing \textit{distribution-free} $\IPP$s in the following sense: we reduce the task of proving proximity to any such language over an unknown (but fixed) distribution, to proving proximity to $\PVAL$ with respect to a new, \emph{hybrid-metric} notion of soundness. This interactive reduction is used in all our distribution-free $\IPP$s in this paper. 

It is worth noting that analysis over this hybrid metric provides for a better exposition for Sections \ref{sec:laconic_dfipp_special_dist} and \ref{sec:dfipp_learnable}, and we also use it here to maintain overall consistency. In essence, it captures the case analysis explained in Section \ref{sec:tech_main_dfipp}. 

Finally, in Section \ref{sec:proof_main_dfipp} we prove that in fact, constructing distribution-free $\IPP$s for any language computable in low-depth reduces to constructing an $\IPP$ for $\PVAL$ \textit{only} over the \textit{uniform distribution}. We restate the main result of this section.

\begin{theorem}[\textbf{Distribution-Free $\IPP$ for languages computable in low depth}]\label{thm:informal_dfipp_nc_restated}
    For every language $L$ in logspace-uniform $\NC$ and every $\varepsilon > \frac{200\log^{3}(n)}{n}$, there exists a distribution-free $\IPP$ for $L$ with proximity $\varepsilon$, having query complexity $O\left( \frac{1}{\varepsilon} \right)$, sample complexity $O\left( \frac{1}{\varepsilon} \right)$ and communication complexity $O\left( \frac{\log (n)}{\varepsilon}\right) + \varepsilon \cdot n\cdot \polylog(n)$. Moreover, the verifier runs in time $\tilde{O}\left( \frac{1}{\varepsilon}+\varepsilon\cdot n \right)$, the prover runs in time $\poly(n)$ and the round complexity is $\polylog(n)$.
\end{theorem}

Note that this parameterisation is different from the one stated in Theorem~\ref{thm:informal_dfipp_nc} and we use it for more convenience in our analysis. The argument that Theorem \ref{thm:informal_dfipp_nc_restated} implies Theorem~\ref{thm:informal_dfipp_nc} follows. For any input length $n \in \N$, let the trade-off parameter $\tau$ be such that $\tau \leq \sqrt{n}$. Applying Theorem~\ref{thm:informal_dfipp_nc_restated} with the proximity parameter $\eps'=min(\eps, 1/\tau)$, we obtain a distribution-free $\IPP$ that has query and sample complexities $O(1/\eps') \leq \tau+O(1/\eps)$, communication complexity $O\left( \frac{\log (n)}{\eps} +\log (n) \cdot \tau\right) + (n/\tau)\cdot \polylog(n)$, and verifier runtime $\tilde{O}\left( \tau+\frac{n}{\tau}+\frac{1}{\eps} \right)$, as stated in Theorem \ref{thm:informal_dfipp_nc}. On the other hand, the converse simply follows by setting $\tau = O(1/\eps)$ in Theorem~\ref{thm:informal_dfipp_nc}.

In particular, for the case when $\varepsilon \geq 1/\sqrt{n}$, we obtain a distribution-free $\IPP$ for $\NC$ that achieves the optimum of the trade-off between the sum of query and communication complexities. Formally,
\begin{corollary}
    \label{cor:dfipp_for_sqrt_error}
    For every language $L$ in logspace-uniform $\NC$ and every $\varepsilon \geq 1/\sqrt{n}$, there exists a distribution-free $\IPP$ for $L$ with proximity parameter $\varepsilon$, having query and sample complexities $O(\sqrt{n})$, with the communication complexity and verifier running time being at most $\tilde{O}(\sqrt{n})$.
\end{corollary}

\subsection{The Polynomial Evaluation Problem}
\label{sec:pval_def}
Let $\F$ be a finite field, and $[k]$ be identified with a subset of $\F$ of size $k\in \N$ via some bijection. Let $m,n \in \N$ be integers such that $n=k^{m}$. We start by defining the low-degree extension of a string in $\F^{k^m}$.
\begin{definition}[Low-degree Extension ($\LDE$)]
\label{def:lde}
For any $X \in \F^n$ $\left( \text{or alternatively } X \in \F^{k^{m}} \right)$, we define $P_X : \F^m \rightarrow \F$ as the unique $m$-variate polynomial over $\F$ with individual degree at most $k - 1$, such that it evaluates to $X$ on $[k]^{m}$.
\end{definition}

We next define the polynomial evaluation problem $\PVAL$. For any fixed set of points $J \subset \F^m$ and a vector of values $\Vec{v} \in \F^{\vert J \vert}$, $\PVAL$ parameterised by $J,\Vec{v}$, is the problem of deciding whether the $\LDE$ of the given input $X \in \F^n$ is consistent with $\Vec{v}$ on the points in $J$. Formally, the language $\PVAL$ is defined as follows.

\begin{definition}[$\PVAL$, \cite{RVW}]
\label{def:pval}
$\PVAL$ is a language parameterised by $(\F,k,m,J,\Vec{v})$, where $J \subset \F^m$ and $\Vec{v} \in \F^{\vert J \vert}$, such that an input $X \in \F^{k^m}$ belongs to $\PVAL(\F,k,m,J,\Vec{v})$ if and only if for every $j\in J$, $P_X(j)=\vec{v}_{j}$. Equivalently, this can be stated as saying that $P_X(J) = \Vec{v}$. 

When it is clear from context, we drop $\F,k,m$ from the explicit input (i.e., the string of input parameters) and consider $X$ as an instance of $\PVAL(J,\Vec{v})$.
\end{definition}

\subsection{Distribution-Free Interactive reduction from \texorpdfstring{$\NC$}{} to \texorpdfstring{$\PVAL$}{}}
\label{sec:NCPVAL}
We first demonstrate an interactive protocol that reduces the problem of constructing a distribution-free $\IPP$ for any low-depth computable language $L$ to an instance of $\PVAL$, but over a \textit{hybrid} distance metric. In more detail, we depart from the interactive reduction by \cite{RVW} for the uniform setting, by showing that in the soundness case where the input is $\varepsilon$-far along $\sD$ from $L$, this generalised reduction produces an instance for $\PVAL$ that only satisfies a \textit{weaker promise} with respect to a new distance measure, instead of guaranteeing distance with respect to $\sD$. Complementing this, our new task will be to design $\IPP$s for $\PVAL$ which are powerful enough to reject a larger set of inputs that satisfy this weaker soundness promise, which we show in Section \ref{sec:proof_main_dfipp} (see also Sections \ref{sec:IPP} and \ref{sec:final_wb_ipp_prod}).

The soundness condition in our case is based on a hybrid metric defined across the underlying distribution $\sD_n$ and the uniform distribution $\sU_n$. We say that on a given input $X$, if for every $Y\in \PVAL$: $d_{\sU}(X,Y)\geq \varepsilon$ or $d_{\sD}(X,Y)\geq \varepsilon$, then the $\IPP$ should reject the input $X$ (this is exactly when $\PVAL$ does not intersect the shaded region of Figure \ref{fig:PVAL}). The set of inputs $X$ $\varepsilon$-far along $\sD$ from $\PVAL$ is a subset of this collection of inputs we need to reject. We explain the reason for reducing to this metric in the following section. 

\subsubsection{Protocol Intuition for the Interactive Reduction}
\label{sec:intuition_reduction_hybrid}

Let $L$ be a language in $\NC$. For any input $X \in \{0,1\}^n$ for $L$, and any arbitrary but fixed distribution $\mathcal{D}\in \Delta([n])$, the interactive reduction runs $t=O(\varepsilon \cdot n \cdot \log (n))$ instances of the GKR protocol on $L$ \cite{GKR15}. Here we identify $X$ as a (Boolean-valued) vector in $\F^{k^{m}}$, where $k^m = n$ as noted above.

The GKR protocol is an interactive protocol whose output is $(j,v)$ on input $X$. If $X\in L$, then $P_X(j)=v$ with probability $1$, and if $X\not\in L$, then $P_X(j) = v$ with probability at most $1/2$. This can be thought of as $X\in L$ being reducible to $X\in \PVAL(\{j\}, \{v\})$.

Each instance of the GKR protocol generates a pair $(j,v)$, and all of these pairs are collected together into $(J,\vec{v})$ which defines an instance of $\PVAL$. The completeness of this reduction follows immediately from that of the GKR protocol, as each pair $(j,v)$ represents the true statement that $P_X(j)=v$. On the other hand, for any $X'$, if $X'\not\in L$, then for any malicious prover, the probability that $X'\in \PVAL(J,\vec{v})$ will be at most $2^{-t}$ by the soundness of the GKR protocol, i.e., $X' \notin \PVAL(J,\vec{v})$ with high probability. By taking a union bound over all $X'$ considered close enough to $X$ (all of which are not in $L$ by the soundness condition of an $\IPP$) and setting $t$ to be large enough, with high probability, all such close element are not in $\PVAL$ implying distance between $X$ and $\PVAL$.

A natural idea is to now try to argue that, by setting $t$ to be sufficiently large, the input $X$ is also far along $\mathcal{D}$ from $\PVAL$.  However, as discussed in Section \ref{sec:tech_main_dfipp}, the difficulty here is that while the size of the uniform $\varepsilon$-ball\footnote{Recall that for a distribution $\mathcal{D}$, the $\varepsilon$-ball along $\mathcal{D}$ around a point $X$, $B_{\mathcal{D}, \varepsilon}(X)$, is the set of elements $\varepsilon$-close to $X$ along $\mathcal{D}$} $B_{\sU,\varepsilon}(X)$ around a single element is relatively small, a similar ball defined by $\mathcal{D}$ can be extremely large. In more detail, the uniform ball around a single element has size $O(n^{\varepsilon n})$, so we need to take $t = O(\log (n^{\varepsilon n}))=O(\varepsilon n\log (n))$ repetitions of the GKR protocol in order to apply the union bound in uniform setting. On the other hand, suppose that the underlying distribution $\sD$ is supported only over the first $\log(n)$ indices, then $\sB_{\sD,\varepsilon}(X)$ contains at least $2^{n-\log(n)}$ elements. We now need $\log (2^{n-\log (n)})\approx n$ many instances of the GKR protocol, which already means that the resulting $\IPP$ no longer as sublinear time verification and communication complexity.

From hereon, it is not obvious how to proceed or even if there is a reduction to showing an $\IPP$ for $\PVAL$. To this end, our solution is to  instead, reduce to $\PVAL$ over a \textit{new} soundness constraint. Beginning with the initial promise on the input, $d_{\mathcal{D}}(X, L)>\varepsilon$, we take the intersection of $B_{\sD,\varepsilon}(X)$ and $B_{\sU,\varepsilon}(X)$. This set is of course, bounded above by the uniform ball $B_{\sU,\varepsilon}(X)$, but also does not contain any elements of $L$ (by the soundness condition) as shown in Figure \ref{fig:PVAL}. Taking a similar union bound over this intersection, we now have a weaker soundness condition for $\PVAL$, $\mu_{\mathcal{D}, \mathcal{U}}(X, \PVAL)>\varepsilon$, signifying that no element in the intersection of those two balls is in $\PVAL$. In other words, with high probability, we have 
\begin{equation*}
    d_{\mathcal{D}}(X,L)>\varepsilon\implies \mu_{\mathcal{D}, \mathcal{U}}(X,\PVAL) = \min_{Y\in \PVAL} \left(\max (d_{\mathcal{U}}(X,Y), d_{\mathcal{D}}(X,Y)) \right) > \varepsilon.
\end{equation*} 

\subsubsection{Interactive Reduction from $\NC$ to $\PVAL$}
\label{sec:df_nc_to_pval}
Below, we prove our reduction from verifying the proximity of an instance to a language in $\NC$ to verifying the proximity of an instance to $\PVAL$, in the distribution-free sense.
\begin{theorem}
\label{thm:df_nc_pval_reduction}
For any $\varepsilon>0$ and any language $L$ computable by logspace-uniform Boolean circuits of depth $\Delta_{L} = \Delta_{L}(n)$, size $S = S(n)$, and fan-in $2$, any $k \in \N$ and $m = \log_{k}(n)$, the following holds. Let $\F$ be a finite field of size $|\F|= \Omega( k\cdot\Delta_{L}\cdot \log (S))$. 

There exists an interactive protocol $(P_{\NC},V_{\NC})$ with input $X \in \{0,1\}^n$, whose output is a coordinate set $J \subseteq \mathbb{F}^{m}$ of size $t = 4\varepsilon \cdot n \cdot \log (n)$, and a vector $\vec{v} \in \mathbb{F}^{\vert J\vert }$ defining an instance of $\PVAL$, such that:
\begin{enumerate}
    \item \textbf{Completeness}: If $X \in L$, then $(J,\vec{v})$ is such that $X \in \PVAL(J,\vec{v})$ (with probability 1). In other words,
    \begin{equation*}
        X\in L\implies \underset{{V_{\NC}}}{\mathbb{P}}[X\in \PVAL(J,\Vec{v})]=1
    \end{equation*}
    
    \item \textbf{Soundness}: $\forall \sD\in\Delta([n])$, if $d_{\mathcal{D}}(X,L)>\varepsilon$, then for any cheating prover strategy with probability at least $1/2$ over the verifier's coins, $\mu_{\mathcal{D}, \mathcal{U}_{n}}(X,\PVAL(J,\vec{v})) > \varepsilon$. In other words,

    \begin{equation*}
        d_{\sD}(X,L)>\varepsilon \implies \underset{{V_{\NC}}}{\mathbb{P}}[\mu_{\sD,\mathcal{U}}(X, \PVAL(J,\Vec{v}))>\varepsilon]\geq \frac{1}{2}
    \end{equation*}
    
\end{enumerate}

The prover runs in time $\poly(S, \log |\F|)$, and the verifier runs in time $\varepsilon \cdot n \cdot \poly(k,\Delta_{L},\log (S),\log |\F|)$ (the verifier does not need to access the input $X$). The communication complexity is $\varepsilon \cdot n \cdot \poly(k,\Delta_{L}, \log (S), \log |\F|)$, and the number of rounds is $O(\Delta_{L} \cdot \log (S))$. Moreover, $J$ is a uniformly random set of $t$ points from $\F^m$.
\end{theorem}

\begin{remark}
Note that this interactive reduction does \emph{not} transform the input $X$ or even \textit{access} it. Instead, we invoke an interactive protocol that outputs a concise description of a different language (namely, a parameterisation of the $\PVAL$ language) to which the distance from $X$ is preserved (with high probability) over a different metric.
\end{remark}

The proof of Theorem \ref{thm:df_nc_pval_reduction}  relies on a result by Goldwasser, Kalai and Rothblum \cite{GKR15}, which states that there is an interactive protocol reducing low depth languages to $\PVAL$ on a single point (i.e., $\vert J\vert =1$).

\begin{theorem}[\cite{GKR15}]
\label{thm:gkr_red}
Let $L$ be a language computable by logspace-uniform  Boolean circuits of depth $\Delta_{L}=\Delta_{L}(n)$, size $S=S(n)$, and fan-in 2, for any $k\in \N$ and $m=\log_{k}(n)$, then the following holds. Let $\F$ be a finite field of size $|\F|= \Omega(k\cdot \Delta_{L}\cdot\log (S))$. 

There exists an interactive protocol $(P_{\GKR}, V_{\GKR})$, with input $X \in \{0,1\}^n$, that outputs a coordinate $j\in \mathbb{F}^{m}$ and a value $v\in \mathbb{F}$, such that:
\begin{itemize}
    \item \textbf{Completeness} - if $X\in L$ then with probability $1$, the m-variate low degree extension with individual degree $k-1$ evaluated at $j$ is $v$: $P_X(j)=v$ with probability $1$. In other words,
    \begin{equation*}
        X\in L \implies \underset{{V_{\GKR}}}{\mathbb{P}}[P_X(j)=v]=1.
    \end{equation*}
    
    \item \textbf{Soundness} - if $X\not\in L$ then  for every prover strategy, the probability that $P_X(j)=v$ is at most $1/2$ over the verifier's randomness. In other words,
    \begin{equation*}
        X\not\in L\implies \underset{{V_{\GKR}}}{\mathbb{P}}[P_X(j)=v]\leq \frac{1}{2}.
    \end{equation*}
\end{itemize}

The verifier runs in time $\poly(k,\Delta_{L},\log (S),\log |\F|)$, the prover runs in time $\poly(S,\log|\F|)$, the communication complexity is $\poly(k,\Delta_{L}, \log (S),\log |\F|)$, and the number of rounds is $O(\Delta_{L} \cdot \log(S))$. Moreover, the coordinate $j$ is a uniformly random point from $\F^m$.
\end{theorem}

\noindent Equipped with Theorem \ref{thm:gkr_red}, we are now able to prove our reduction.

\begin{proof}[Proof of Theorem \ref{thm:df_nc_pval_reduction}]
Let $(P_{\NC}, V_{\NC})$ be the protocol for which $t=2\varepsilon n (\log (n) +\log \vert \mathbb{F}\vert )\leq 4\varepsilon n \log (n)$ iterations of $(P_{\mathsf{GKR}}, V_{\mathsf{GKR}})$ from Theorem \ref{thm:gkr_red}, are run in parallel. This implies that the round complexity is the same as that of a single iteration of this protocol (see, e.g., \cite[Appendix A]{GR17_round_hierarchy} for additional details). This yields $t$ pairs of the form $(j,v)\in \mathbb{F}^{m}\times \mathbb{F}$. We collect all of these terms into a set of coordinates $J\subseteq \mathbb{F}^{m}$ and a set of claims of values on the set $\vec{v}\in \mathbb{F}^{t}$.

The running times, round complexity and communication complexity follow from this construction as it is $t=O(\varepsilon n (\log (n) + \log \vert \mathbb{F} \vert ))$ times the complexity of a single run of $(P_{\mathsf{GKR}}, V_{\mathsf{GKR}})$. By the perfect completeness of each run of $(P_{\mathsf{GKR}}, V_{\mathsf{GKR}})$, each invocation produces a pair $(j,v)$ for which $P_X(j)=v$, we have that the collection of these $t$  parallel runs that generate an instance of $\PVAL$ also satisfy that $P_X(J)=\vec{v}$. 

It remains to prove the soundness of this protocol. Suppose that $d_{\mathcal{D}}(X,L)>\varepsilon$. Now, the probability that any $X'\not \in L$ is in $\PVAL$ is at most $2^{-t}<1/(2n\cdot \vert \F\vert )^{\varepsilon \cdot n}$; this is the probability that it is consistent with the outputs of each run of $(P_{\mathsf{GKR}}, V_{\mathsf{GKR}})$. 

It suffices to prove that, with high probability, no element of $\PVAL$ is in the intersection $\sB_{\sD,\varepsilon}(X)$ and $\sB_{\sU,\varepsilon}(X)$. The number of elements in the intersection of these two balls is bounded by the size of the uniform ball, in other words,

\begin{equation*}
    \vert B_{\mathcal{D},\varepsilon}(X)\cap B_{\mathcal{U}_{n},\varepsilon}(X)\vert\leq\vert B_{\mathcal{U}_{n},\varepsilon}(X)\vert\leq n^{\varepsilon n}.
\end{equation*}
Therefore taking the union bound over the entire intersection we have the probability of any element in that intersection satisfying $\PVAL$ is less than $n^{\varepsilon n} \cdot \frac{1}{2^{4 \varepsilon n \log (n)}}<1/2$ as all of those elements are not in $L$, this  implies that $\mu_{\mathcal{D}, \mathcal{U}_{n}}(X,\PVAL)>\varepsilon$ with high probability.
\end{proof}

\subsection{Proof of Theorem \ref{thm:informal_dfipp_nc_restated}}
\label{sec:proof_main_dfipp}
In this section, we prove our main result that constructs distribution-free $\IPP$s for any language computable by logspace-uniform circuits of low-depth. 

\paragraph*{High level sketch of the proof:} Consider the (soundness) case where the input $X \in \{0,1\}^n$ to an $\NC$-language $L$ is such that $d_{\sD}(X,L)>\varepsilon$. Our main goal is to reduce the construction of a distribution-free $\IPP$ for $L$ to a uniform $\IPP$ for $\PVAL$ over a \textit{larger} index set $(J \cup I)$, for which we can use a pre-existing $\IPP$ from \cite{RR20_batch_polylog} (which is a quantitative improvement over a prior $\IPP$ for $\PVAL$ from \cite{RVW}), and get the stated query and communication complexities.

To this end, for the output $(J,\vec{v})$ for which the interactive reduction from $\NC$ to $\PVAL$ from Theorem \ref{thm:df_nc_pval_reduction} guarantees that $\mu_{\sD, \mathcal{U}}(X,\PVAL(J,\Vec{v})) > \varepsilon$. If $X$ is far from $\PVAL$ along the uniform distribution, the uniform $\IPP$ will reject and so we can assume that $X$ is close to $\PVAL(J,\vec{v})$ uniformly (i.e., $d_{\sU}(X,\PVAL(J,\vec{v}))\leq \varepsilon$). 
At this point we observe that since $\PVAL$ is a good error correcting code (i.e., with large minimal distance), the input $X$ must be close to a \emph{unique} element $X' \in \PVAL(J,\vec{v})$. However, by our soundness assumption over $\mu$, we know that $d_{\sD}(X,X')>\varepsilon$.

Now, the verifier generates $O(1/\varepsilon)$ samples $I$ from $\sD$ (and the corresponding values in $X$). Let $\PVAL'(I, X \vert_I)$ be the set of strings in $\PVAL$ which agree with $X$ on $I$. Alternatively, $\PVAL'(I, X \vert_I) = \PVAL((J, I), (\Vec{v},X \vert_I))$. Since $d_\sD(X,\PVAL(J,\Vec{v})) > \varepsilon$ from our initial assumption, we see that with high probability there exists an index in $I$ on which $X'$ and $X$ disagree on. In other words, $X'$ is not in $\PVAL'(I, X \vert_I)$ and using the above properties of $X'$, we see that $X$ is $\varepsilon$-far from $\PVAL((J, I), (\Vec{v},X \vert_I))$ along the \textit{uniform distribution}. Thus, by applying the uniform $\IPP$ we can catch the cheating prover.

\begin{theorem}[Theorem \ref{thm:informal_dfipp_nc_restated} restated]
\label{thm:simplerncdfippsizedepth}

For every $n \in \N$, let $L \subseteq \{0,1\}^{n}$ be a language computable by logspace-uniform  circuits with depth $\Delta_{L}=\Delta_{L}(n)\geq \log (n)$ and size $S=S(n)$. Then, for  $\varepsilon > \frac{200\log^{3}(n)}{n}$, there exists a distribution-free interactive proof of proximity for $L$ with perfect completeness and soundness at least $1/2$. 

This protocol has query complexity $O\left( \frac{1}{\varepsilon} \right)$, sample complexity $O\left( \frac{1}{\varepsilon} \right)$, and communication complexity $O\left(\frac{\log(n)}{\varepsilon} +\varepsilon\cdot n\cdot \poly(\Delta_{L})\right)$. In addition, the honest prover runs in time $\poly(n,S)$ and the verifier runs in time $\tilde{O}\left( \frac{1}{\varepsilon}+\varepsilon\cdot n\cdot\poly(\Delta_{L}) \right)$. Finally, the round complexity of the protocol is $\polylog(n)+O(\Delta_{L}\cdot\log (S))$.
\end{theorem}

\noindent To prove this, we need the following $\IPP$ for $\PVAL$ over the uniform distribution.
\begin{theorem}[Uniform $\IPP$ for $\PVAL$ \cite{RR20_batch_polylog}]
\label{thm:uniform_ipp_pval_rvw}
    Let $n\in \N$, $\varepsilon \geq \frac{200\log^{3}(n)}{n}$, and $\F$ be a finite field of characteristic $2$ of size $|\F|=\Theta(n^{3}\varepsilon^{2}\log^{4}(n))$.
    
    Then, for any set $J\in (\F^{m})^{t}$ of size $O(n\cdot \varepsilon\cdot \log (n))$ and $\vec{v}\in \F^{t}$, there exists a uniform $\IPP$, $(P_{\mathsf{Unif}}, V_{\mathsf{Unif}})$, for $\PVAL(J,\vec{v})$. This protocol has perfect completeness, soundness $1/2$, query complexity $O(1/\varepsilon)$, communication complexity $\tilde{O}(n\cdot \varepsilon)$. Moreover, the honest prover runs in $\poly(n)$ time, the verifier runs in time $\tilde{O}\left(\frac{1}{\varepsilon}+\varepsilon\cdot n\right)$, and the number of messages communicated between them is $\polylog(n)$.
\end{theorem}

\begin{proof}[Proof of Theorem \ref{thm:simplerncdfippsizedepth}]
We construct a distribution-free $\IPP$ with perfect completeness and a constant soundness error (specifically $4/5$) which can reduced to, say $1/3$, by repetition. This distribution-free $\IPP$ for $\NC$ is given in Protocol \ref{pcl:simple}. 
\begin{protocol}[h!]
\caption{Distribution-free $\IPP$ for any language $L$ computable by circuits of size $S(n)$ and depth $\Delta_{L}(n)$.}
\label{pcl:simple}
\begin{algorithmic}{}
\vspace{0.1in}
\STATE \textbf{Input:} The verifier $V_{\mathsf{df}}$ gets implicit input $X\in \{0,1\}^{n}$ that is accessible through a query oracle, as well as the sample oracle $\sO_\sD(X)$, for some unknown distribution $\sD$. The verifier also gets explicit access to $\varepsilon > 0$. The prover $P_{\mathsf{df}}$ gets direct access to $X$ and $\varepsilon$.
\vspace{0.1in}
\STATE \textbf{The distribution-free $\IPP$:}

    \begin{enumerate}
    \item Let $(P_{\NC}, V_{\NC})$ be the interactive reduction from Theorem \ref{thm:df_nc_pval_reduction} with proximity parameter $\eps$. $P_{\mathsf{df}}$ and $V_{\mathsf{df}}$ run $(P_{\NC},V_{\NC})$ on $X$, to output a set $J \subset \F^m$ of size $t=4\varepsilon\cdot n\cdot\log (n)$ and $\Vec{v}\in \F^{t}$, using parameters $k = 2$ and $m = \log(S)$.
            
    \item \label{step:dfNCsimplesamples} $V_{\mathsf{df}}$ sets $T= 3/\varepsilon$ and picks $T$ fresh samples $I=((i_{1},X_{i_1}), \dots, (i_{t},X_{i_T}))$ from $\sO_\sD(X)$. Let $z\in \{0,1\}^{T}$ $z_{j}=X_{i_{j}}$, for every $j \in [T]$.
    The verifier    $V_{\mathsf{df}}$ sends $(I,z)$ to $P_{\mathsf{df}}$.

    \item \label{step:IPP_unif}$P_{\mathsf{df}}$ and $V_{\mathsf{df}}$ run the uniform $\IPP$ $(P_{\mathsf{Unif}}, V_{\mathsf{Unif}})$ from Theorem \ref{thm:uniform_ipp_pval_rvw} for $\PVAL((J, I), (\Vec{v},\Vec{z}))$ on input $X$, using parameters $m=\log (n)$ and $r=\log(1/\varepsilon)$.
    
    \item $V_\mathsf{df}$ accepts if and only if $V_\mathsf{Unif}$ accepts.
    \end{enumerate}
\item
\end{algorithmic}
\end{protocol}

The complexities follow from inspection for the chosen parameters. In particular, the query complexity comes from the $\IPP$ for $\PVAL(J,\vec{v})$ from Theorem \ref{thm:uniform_ipp_pval_rvw} and the $O(1/\varepsilon)$ samples in Step \ref{step:dfNCsimplesamples} of the protocol make up the sample complexity. The running times of the prover and the verifier, along with the round complexity come from the sums of the respective values from Theorems \ref{thm:df_nc_pval_reduction} and \ref{thm:uniform_ipp_pval_rvw}. The communication complexity follows similarly, but also includes the $T$ samples sent to the prover in addition.

The perfect completeness of Protocol \ref{pcl:simple}, follows from the combination of the completeness guarantees of $(P_{\NC}, V_{\NC})$, as well as $(P_{\mathsf{Unif}}, V_{\mathsf{Unif}})$. On the other hand, for soundness, suppose $d_{\sD}(X,L)>\varepsilon$. Firstly, we have the following result from \cite{RR20_batch_polylog}.
\begin{lemma}[Follows from Proposition 5.4 in \cite{RR20_batch_polylog}]
\label{lem:rr20_min_dist_random_pval}
Let $\F$ be any field and $m,n\in \N$. Let $d_{\min}(\PVAL(J,\Vec{v}))$ represent the relative minimum Hamming distance between any pair of $n$-length strings in $\PVAL(J,\Vec{v})$. For any $t \geq 2\varepsilon\cdot n (\log (n)+\log |\F|)+4$, we have
\begin{equation*}
    \underset{J\sim \sU_{(\F^{m})^{t}}}{\mathbb{P}}[d_{\min}(\PVAL(J,\vec{v}))<2\varepsilon \cdot n]<2^{-4}<1/10.
\end{equation*}
\end{lemma}

Using this, we prove the following lemma that establishes the constraints satisfied by the output $(J,\Vec{v})$ of the protocol $(P_\NC, V_\NC)$ from Theorem \ref{thm:df_nc_pval_reduction}. 
\begin{lemma}
\label{lem:PVALasaCodeSoundness}
For any $n\in\mathbb{N}$, $\varepsilon>0$, distribution $\sD$ over $[n]$, $X\in\{0,1\}^{n}$, and language $L\subseteq \{0,1\}^{n}$ computable by logspace-uniform  circuits with depth $\Delta_{L}=\Delta_{L}(n)$ and size $S=S(n)$, if $d_{\sD}(X,L)>\varepsilon$, the output $(J,\vec{v})$ of the protocol $(P_\NC,V_\NC)$ from Theorem \ref{thm:df_nc_pval_reduction} satisfies the following conditions:
\begin{itemize}
    \item $\underset{\mathcal{V}_{\NC}}{\mathbb{P}}[\mu_{\sD,\mathcal{U}}(X,\PVAL(J,\vec{v}))>\varepsilon]>0.5$.
    \item $\underset{\mathcal{V}_{\NC}}{\mathbb{P}}[\exists X_1 \neq X_2, \text{ such that } X_1, X_2 \in \sB_{\sU,\varepsilon}(X) \text{ and } X_{1},X_{2} \in \PVAL(J,\vec{v})] < 0.1$.
\end{itemize}
\end{lemma}
While the first condition maintains the soundness guarantee along the hybrid metric promised by Theorem \ref{thm:df_nc_pval_reduction}, the second condition implies that there is at most one element close to $X$ uniformly that is in $\PVAL$.

\begin{proof}[Proof of Lemma \ref{lem:PVALasaCodeSoundness}]
    The first item is satisfied from the soundness guarantees of $(P_\NC,V_\NC)$ by Theorem \ref{thm:df_nc_pval_reduction}. 
    
    To prove the second item, we first observe that the probability (over the internal randomness of $V_\NC$) that there exist two distinct strings that are $\varepsilon$-close to $X$ along the uniform distribution in $\PVAL$ is at most the probability that there exist two distinct strings in $\PVAL$ that are $2\varepsilon$-close.

    \[
    \begin{split}
        \underset{\mathcal{V}_{\NC}}{\mathbb{P}}[\exists X_1 \neq X_2, \text{ such that } &X_1, X_2 \in \sB_{\sU,\varepsilon}(X) \text{ and } X_{1},X_{2} \in \PVAL(J,\vec{v})]\\ 
        &\leq \underset{\mathcal{V}_{\NC}}{\mathbb{P}}[\exists X_1 \neq X_2, \text{ such that } d_{\mathcal{U}}(X_1, X_2)<2\varepsilon \text{ and } X_{1},X_{2} \in \PVAL(J,\vec{v})]\\
        &=\underset{\mathcal{V}_{\NC}}{P}[d_{\min}(\PVAL(J,\vec{v}))<2\varepsilon \cdot n] \\ 
        &< 0.1. \\
    \end{split}
    \]

    The first transition follows from the triangle inequality for Hamming distances as the distance between such an $X_{1},X_{2}$ is at most the sum of their distances to $X$. In turn, this probability is equal to that of the minimum distance of $\PVAL$ being less than $2\varepsilon$, as seen in the next line. Since $J$ is distributed uniformly at random and generated using $V_\NC$'s internal randomness, and its size is $4\varepsilon\cdot n \log (n)\geq 2\varepsilon\cdot n\cdot (\log (n) +\log |\F|)+4$, this probability can be upper bounded using Lemma \ref{lem:rr20_min_dist_random_pval}.
\end{proof}
    
    In the first step of each repetition of Protocol \ref{pcl:simple}, we have $\mu_{\sD,\sU}(X,\PVAL(J,\vec{v}))>\varepsilon$ with probability at least $1/2$ by the first item of Lemma \ref{lem:PVALasaCodeSoundness}. Suppose that $d_{\sU}(X,\PVAL(J,\vec{v}))>\varepsilon$, then $V_{\mathsf{Unif}}$ rejects with probability at least $1/2$ in Step \ref{step:IPP_unif}. 
    
    On the other hand, suppose that $d_{\sU}(X,\PVAL(J,\vec{v}))\leq\varepsilon$. Then, from the second item of Lemma \ref{lem:PVALasaCodeSoundness}, observe that with probability at least $9/10$, there exists at most one $W\in \PVAL(J,\vec{v})$ such that $d_{\sU}(X,W)<\varepsilon$. 
    
    Further, using the guarantee from the first item of Lemma \ref{lem:PVALasaCodeSoundness}, we see that $d_\sD(X,\PVAL(J,\vec{v})) > \varepsilon$, and in particular, $d_{\sD}(X,W)>\varepsilon$. In turn, this implies that $W \not\in \PVAL((J, I), (\Vec{v},\Vec{z}))$ with probability at least $9/10$, since at least one of the entries of $z$ will be an index on which $X$ and $W$ disagree. More precisely,
\[
    \begin{split}
        \mathbb{P}[W\not\in \PVAL((J,I), (\Vec{v},\Vec{z}))]&\geq\mathbb{P}[\exists i\in I: X_{i}\neq W_{i}]\\
        &\geq 1-(1-\varepsilon)^{t}\\
        &=1-(1-\varepsilon)^{3/\varepsilon}\\
        &\geq 1-e^{-3}\\
        &>9/10.
    \end{split}
\]
  
    Put together, in each repetition, with probability at least $2/5$ (over the internal randomness of $V_\NC$ and the choice of $I$), $d_{\sU}(X,\PVAL((J, I), (\Vec{v},\Vec{z})))>\varepsilon$. Indeed, this is the probability that both the items of Lemma \ref{lem:PVALasaCodeSoundness} hold (more precisely, Item 1 and the complement of Item 2), times the probability that $W \notin \PVAL((J,I), (\Vec{v},\Vec{z}))$.

    Thus, at the end of each round, $V_{\mathsf{df}}$ rejects $X$ with probability at least $1/5$, by the soundness guarantee of Theorem \ref{thm:uniform_ipp_pval_rvw}. Soundness error $1/2$ can now be achieved by standard soundness amplification.
\end{proof}

\section{$\IPP$s over Dispersed Distributions}
\label{sec:laconic_dfipp_special_dist}
In Section \ref{sec:dfipp_low-depth}, we showed the construction of a distribution-free $\IPP$ for any $\NC$ language that uses $\tau + O\left (\frac{1}{\varepsilon}\right)$ queries and $\tilde{O}(\frac{n}{\tau} + \frac{1}{\varepsilon})$ bits of communication, for every $\tau \leq \sqrt{n}$ and $\eps$. For $\eps> \tau/n$ this matches the best $\IPP$s for the uniform distribution. However, for $\eps < \tau/n$, the communication complexity is at least $\Tilde{\Omega}(n/\tau)$. Compare this to \cite{RVW}, where for any $\varepsilon=o \left( \frac{\tau}{n} \right)$, the communication complexity is $\Tilde{O}(n/\tau)$, still with $O(1/\varepsilon)$ query complexity.

While we do not know how to overcome this problem in general, in this section, we introduce a classification of distributions over $[k]^m$ (and $k^m = n$), which we call $\rho$-dispersed distributions (for $1 \leq \rho \leq k$), to capture the ``closeness" of the behaviour of an underlying distribution to the Uniform distribution. Under such a definition, the larger the value of $\rho$, the lesser the distribution behaves like the uniform distribution in this sense. In particular, $\rho = 1$ captures the uniform distribution, while every distribution is $k$-dispersed. For $\IPP$s over the set of $\rho$-dispersed distributions for a small enough $\rho$, we match the result by \cite{RVW}. 

For our main result of this section, we construct $\IPP$s for $\NC$ languages over $\rho$-dispersed distributions, which give a smooth trade-off between $\rho$ and the query complexity for fixed communication complexity. In particular, for small enough values of $\rho$ (i.e., distributions behave like the uniform distribution), these $\IPP$s achieve better communication complexity than Theorem \ref{thm:informal_dfipp_nc}, for roughly the same query complexity when the proximity parameter $\varepsilon $ is less than $1/\sqrt{n}$.

\paragraph{Section Organization.} 
First, we define $\rho$-Dispersed Distributions in Section \ref{sec:def_dispersed} and then  state the main theorem of this section which is an $\IPP$ for $\NC$ over such distributions. Subsequently, we show the construction of an $\IPP$ for $\PVAL$ over (hybrid metrics for) $\rho$-Dispersed Distributions in Sections \ref{POLYFOLD} and \ref{sec:IPP}, that builds on certain structural properties of such distributions. The final $\IPP$ from Theorem \ref{thm:dispersed_ipp_nc} is obtained by combining the reduction from Theorem \ref{thm:df_nc_pval_reduction} with this $\IPP$ for $\PVAL$, and its formal details are provided in Section \ref{sec:IPP}.

\subsection{$\rho$-Dispersed Distributions}
\label{sec:def_dispersed}
Below, we define $\rho$-Dispersed distributions. These are distributions over $[k]^{m}$ for some $k,m\in \N$ over which the probability of any element is at most $\rho$ times the average probability taken over any single dimension. Formally,

\begin{definition}[$\rho$-Dispersed Distributions]\label{def:dispersed}
    Let $\rho\in \mathbb{R}$ be such that $1 \leq \rho \leq k$. We say that a distribution $\sD\in \Delta([k]^{m})$ is $\rho$-Dispersed, if for every $j\in [m]$ and for every $(i_1, \dots, i_m) \in [k]^{m}$, we have 
    \begin{equation*}
        \sD(i_1, \dots, i_m) \leq \rho \cdot \underset{t\sim U_{k}}{\mathbb{E}}[\sD(i_1, \dots, t, \dots, i_m)],
    \end{equation*}
    where $\sD(\cdot)$ denotes the probability mass function.
    
\end{definition}
In other words, for any element $i$ in the support of a $\rho$-Dispersed distribution $\sD$ and for every dimension $j \in [m]$, $\sD(i)$ is at most $\rho$ times the average of $\sD$ along the $j^{\text{th}}$-dimension, keeping the rest of the coordinates of $i$ fixed. 

For any $1 \leq \rho \leq k$, we consider $\rho$-Dispersed distributions as a way of capturing distributions that behave closely to the uniform distribution. In particular, observe the following simple facts.
\begin{itemize}
    \item The uniform distribution is a $1$-Dispersed distribution as this implies that no index has weight which is greater that the average.
    \item  Moreover, any distribution over $[k]^{m}$ is trivially a $k$-Dispersed distribution as $k$ times the average is the total over that dimension for which every weight in that column is at most that value. 
    \item Let $\sD$ be a distribution, such that for some $i_{1},\cdots, i_{m-1}\in [k]$, there is only one element in the set $\{(i_{1},\cdots, i_{m-1},s)\}_{s\in [k]}$ for which $\sD$ has non-zero weight (and the rest of the distribution can behave arbitrarily). Then, $\sD$ is $k$-dispersed but not $(k-\delta)$-dispersed for any $\delta>0$. Indeed, suppose the element with non-zero weight in this set is $(i_1, \dots, i_{m-1},s_{0})$. Then, 
    
    \begin{equation*}
    \sD(i_{1},\cdots, i_{m-1},s_{0})=\sum_{s=1}^{k}\sD(i_{1},\cdots, i_{m-1},s)=k\underset{s\sim U_{k}}{\mathbb{E}}[\sD(i_{1},\cdots, i_{m-1},s)].
    \end{equation*}
    
    \item $\alpha$-log Lipschitz distributions, introduced by \cite{AFK13}, define distribution families which are locally smooth, in the sense that if points that are close to each other in Hamming distance cannot have vastly different probability masses under $\sD$. More formally, we define $\alpha$-log Lipschitz distributions as follows:

    \begin{definition}
    A distribution $\sD$ is $\alpha$-log Lipschitz if $\forall x,x'\in [k]^{m}$ that differ in only one value, the following holds.

    \begin{equation*}
        |\log(\sD(x))-\log(\sD(x'))|\leq \log (\alpha)
    \end{equation*}
    Equivalently stated, we have for every $x,x'\in [k]^{m}$ that differ in only one value, $\frac{\sD(x)}{\sD(x')}\leq \alpha$.
    \end{definition}

    This notion captures a wide variety of popularly studied distributions and has been studied in several different contexts (cf. \cite{AFK13} for more references). Examples include the uniform distribution ($\alpha = 1$) or product distributions over $[k]^{m}$, where the probability of sampling each element is in the interval $\left[ \frac{1}{k-1+\alpha}, \frac{\alpha}{k-1+\alpha} \right]$. 
    
     We observe that any $\alpha$-$\log$-Lipschitz distribution is also $ \left( \frac{\alpha\cdot k}{\alpha+k-1} \right)$-dispersed.\footnote{Note that this inclusion is strict, since $\alpha$-dispersed distributions need not be supported on all elements, unlike $\alpha$-$\log$-Lipschitz distributions, which by definition have non-zero measure everywhere.} This holds because, the log-Lipschitz condition of $\sD$ implies that for the index $(i_{1},\cdots, i_{m})\in [k]^{m}$ and value $t\in [m]$ for which the ratio $\frac{\sD(i_{1},\cdots i_{m})}{\underset{t\sim [k]}{\mathbb{E}}[\sD(i_{1},\cdots,t,\cdots, i_{m})]}$ is maximised, the following is true. 
     
     \[
     \begin{split}
     \max\left(\frac{\sD(i_{1},\cdots i_{m})}{\underset{t\sim [k]}{\mathbb{E}}[\sD(i_{1},\cdots,t,\cdots, i_{m})]}\right)&=\max\left(\frac{k\sD(i_{1},\cdots i_{m})}{\sum_{t\in [k]}\sD(i_{1},\cdots,t,\cdots, i_{m})}\right)\\
     &\leq \frac{k\sD(i_{1},\cdots i_{m})}{\sD(i_{1},\cdots i_{m})+\frac{k-1}{\alpha}\sD(i_{1},\cdots i_{m})}\\
     &=\frac{\alpha\cdot k}{\alpha+k-1}
     \end{split}
     \] 
     
     and thus, it is $\left( \frac{\alpha\cdot k}{\alpha+k-1} \right)$-dispersed. In particular, we know that $k^{o(1)}$-$\log$-Lipschitz distributions are $k^{o(1)}$-dispersed.

    \item Let $\hat{\sD}$ be an $m$-product distribution over $[k]^{m}$ defined as $\hat{\sD} = \sD \times \dots \times \sD$, where $\sD$ is a distribution over $[k]$ with minimum weight $p_{min}$ and maximum $p_{max}$, and the product is taken $m$ times. Then, $\hat{\sD}$ is $\left( \frac{kp_{max}}{p_{max}+(k-1)p_{min}} \right)$-dispersed. To see this, observe that, $p_{max} = \left( \frac{kp_{max}}{p_{max}+(k-1)p_{min}}\right) \cdot \underset{t\sim\sD}{\mathbb{E}}[\sD(t)] $ as $\underset{t\sim\sD}{\mathbb{E}}[\sD(t)]= \frac{p_{max}+(k-1)p_{min}}{k}$. 
\end{itemize}

For any distribution $\sD$ over $[k]^{m}$, for any $p\in [m]$ and for any $(i_{1},\cdots,i_{m-p})\in [k]^{m-p}$, we define the marginal distribution over $p$ dimensions as  $\mathcal{D}^{(m-p)}\in \Delta(\Omega_{k^{m-p}})$ as follows.
\begin{equation*}
    \sD^{(m-p)}(i_{1},\cdots, i_{m-p})=\sum_{t\in [k]}\sD(i_{1},\cdots i_{m-p},t).
\end{equation*}
\noindent Note that for $p_{1},p_{2}\in [m]$, $\sD^{(m-p_{1})^{(m-p_{2})}}=\sD^{(m-p_{1}-p_{2})}$.

\begin{lemma}\label{lem:fold_dispersed}
    For any $m,k,\rho\in \mathbb{N}$, and any distribution $\sD$ over $[k]^{m}$, $\sD$ is a $\rho$-Dispersed distribution implies $\sD^{(m-1)}$ is $\rho$-Dispersed.
\end{lemma}

\begin{proof}
    We first restrict our attention to the $\rho$-Dispersed condition on the first index and the same analysis extends to this condition on any index in $[m-1]$. For any $(i_1, \dots, i_{m-1}) \in [k]^{m-1}$, we have
    \[
    \begin{split}
        \sD^{(m-1)}(i_{1},\cdots, i_{m-1})&=\sum_{j\in [k]}\sD(i_{1},\cdots, i_{m-1},j)\\
        &\leq \sum_{j\in [k]}\rho\underset{{l\sim [k]}}{\mathbb{E}}[\sD(l,i_{2},\cdots, i_{m-1}, j)]\\
        &\leq \rho\underset{{l\sim [k]}}{\mathbb{E}}\left[\sum_{j\in [k]}\sD(l,i_{2},\cdots, i_{m-1}, j)\right]\\
        & \leq\rho \underset{{l\sim [k]}}{\mathbb{E}}[\sD^{(m-1)}(l,i_{2},\cdots, i_{m-1})]\\
    \end{split}
    \]
    The first and last lines follow by the definition of $\sD^{(m-1)}$. 
\end{proof}

\noindent We now state the main theorem of this section. Again, for convenience, we stick to the setting where the query vs communication complexity trade-off parameter is set to $O(1/\eps)$. 
\begin{theorem}[Formal statement for Theorem \ref{thm:informal_ipp_dispersed}]
\label{thm:dispersed_ipp_nc}
Let $n\in \N$, and set $k=\log (n)$ and $m=\log_{k}(n)$ (such that $k^m=n$), let $L \subseteq \{0,1\}^{n}$ be a language computable by logspace-uniform  circuits with depth $\Delta_{L}=\Delta_{L}(n)$ and size $S=S(n)$. Then, for $\varepsilon > 0$, $\rho\in \mathbb{R}$, there exists an interactive proof of proximity over $\rho$-Dispersed distributions over $[k]^{m}$ for $L$ with perfect completeness and soundness at least $1/2$.

This protocol has query complexity $\frac{\rho^{\log (1/\varepsilon)/\log\log (n)}}{\varepsilon ^{1+o(1)}}$, sample complexity $\frac{\rho^{\log (1/\varepsilon)/\log\log (n)}}{\varepsilon ^{1+o(1)}}$, communication complexity $
\varepsilon^{1-o(1)}\cdot  n\cdot \log^{2}(n)+\varepsilon\cdot n\cdot \poly(\Delta_{L})$ and round complexity $O\left(\frac{\log\left(\frac{1}{\varepsilon}\right)}{\log\log (n)}+\Delta_{L}\cdot\log (S)\right)$. In addition, the honest prover runs in time $\poly(S,n)$ and the verifier runs in time
\[n^{o(1)} \cdot \left( \frac{\rho^{\log (1/\varepsilon)/\log\log (n)}}{\varepsilon}+\varepsilon\cdot n\cdot\poly(\Delta_{L}) \right).
\]
\end{theorem}

It is worth noting that for $\rho=k^{o(1)}$, we have that $\rho^{\log (1/\varepsilon)/\log\log (n)}=k^{o(r)}=1/\varepsilon^{o(1)}$ and so we match the query and communication complexities of the uniform $\IPP$ from \cite{RVW} (up to poly logarithmic factors). More importantly, $\rho$ \textit{does not contribute} to the communication complexity of the $\IPP$.

\subsection{Polynomial Folding Protocol}\label{POLYFOLD}
We now demonstrate an interactive protocol that accepts any input in $\PVAL(J,\Vec{v})$, while rejecting any input that is far along the hybrid metric for any $\rho$-Dispersed $\sD$ from Definition \ref{def:dispersed}. The idea is to reduce an instance of $\PVAL$ to a set of $\PVAL$ instances, each on one lesser variable. We then generalise the protocol analysis in \cite{RVW}, to be able to handle this new condition for soundness over the hybrid metric. 

We require the following notation for this protocol intuition. We define column marginals over $[k_{1}]\times [k_{2}]$ for $k_{1},k_{2}\in \N$. For which we define the marginal as follows.

\begin{equation*}
    \forall j\in [k_{2}]: \sD^{c}(j)=\sum_{i\in[k_{1}]}\sD(i,j).
\end{equation*}

We define a hybrid metric $\mu_{\mathcal{D}^{c}, U_{k_{2}}}$ over the column marginals of $\mathcal{D}$ and the uniform distribution. For any $x,y \in \F^{k_2}$, we have
\begin{equation*}
    \mu_{\mathcal{D}^{c}, U_{k_{2}}}(x,y) = max(d_{\mathcal{D}^{c}}(x,y), d_{\mathcal{U}_{k_2}}(x,y)).
\end{equation*}

\noindent Whenever the usage of $\mathcal{D}$ is clear from the context, we refer to this metric as $\mu^{c}$.

\subsubsection{Protocol Intuition for Polynomial Folding}
\label{sec:large_int_fold}
\begin{figure}
    \centering
    \includegraphics[width=0.4\textwidth]{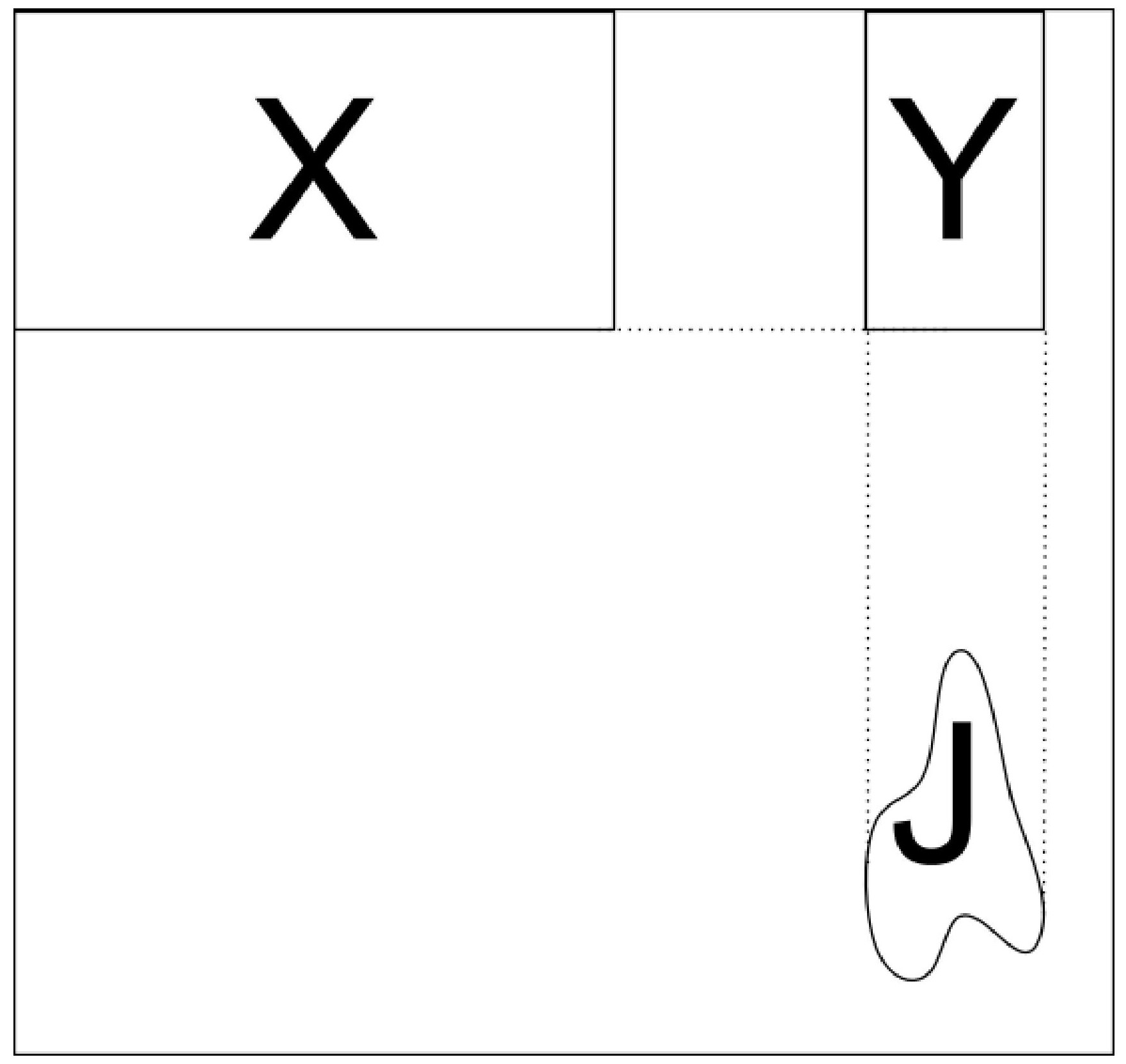}
    \caption{During the polynomial folding protocol, the prover sends the univariate $\LDE$ of each row of X evaluated on the columns of $J$, collected in the matrix $Y \in \F^{k_1 \times t}$. For any $j = (j_1,j_2) \in J$, the univariate $\LDE$ of the $j_2^{\text{th}}$-column of $Y$ restricted to $j_1$ is equal to $\Vec{v}[j]$.}
    \label{fig:polyfold}
\end{figure}

Let $\mathbb{F}$ be a finite field of size $\max\{\ell_{1}, \ell_2\} = \poly(k_1,k_2)$, where $k_1<\ell_{1}$, $k_2 < \ell_{2}$. We start with the two dimensional case, by viewing the input $X$ as an element in $\mathbb{F}^{k_{1} \times k_2}$ and defining $P_{X}: \mathbb{F}^2 \rightarrow \mathbb{F}$ as its bivariate low-degree extension ($\LDE$). We reduce the problem of checking proximity of $X$ to $\PVAL$ on bivariate $\LDE$s along $\mu_{\mathcal{D},\mathcal{U}_{n}}$ to checking proximity of strings in $\mathbb{F}^{k_{2}}$ to $\PVAL$ defined on \textit{univariate} $\LDE$s (of degree $k_2-1$) along $\mu_{\mathcal{D}^{c}, \mathcal{U}_{k_{2}}}$. This idea naturally extends to the $m$-variate case, where $X \in \mathbb{F}^{k^{m}}$ as we can reduce the dimensionality by 1 repeatedly, by taking $X$ to be a $k \times k^{m-1}$ matrix of values over $\mathbb{F}$. 

In more detail, $\PVAL$ is parameterised by $(J,\vec{v})$ with input $X$, where $J\subseteq \mathbb{F}^{l_{1}\times l_{2}}$ is a set of $t$ coordinates, and $\PVAL$ is satisfied if and only if the evaluations of $P_X$ on these points is the vector $\vec{v}$. Let $J_{2}=\{i_{2} \mid \exists i \text{ s.t. } (i,i_{2})\in J\}$, i.e., $J_2$ is the projection of the elements in $J$ onto its second coordinate. For every $i_{2}\in J_{2}$, we define a new set of coordinates $Q_{i_{2}}=\{(i,i_{2}) \mid i\in [k_{1}]\}$. Note that, we can interpolate the value at any coordinate $(i_{1},i_{2}) \in J$ using the univariate $\LDE$ over $Q_{i_2}$, of degree $k_1-1$. The condition for completeness is for $X$ to be in $\PVAL(J,\vec{v})$, whereas for soundness we would like to have that $X$ is $\varepsilon$-far from $\PVAL(J,\vec{v})$ along the ($\mathcal{D}$, $\mathcal{U}$)-hybrid metric, i.e., $\mu_{\mathcal{D}, \mathcal{U}_{n}}(X,\PVAL(J,\vec{v}))>\varepsilon$. 

The protocol proceeds as follows, for each $i_{2}\in J_{2}$, the prover sends the evaluations of $P_{X}(Q_{i_{2}})$. There will with high probability be $t$ such values of $i_{2}$ (as $\vert J_{2}\vert\approx\vert J \vert$=t). Since each $Q_{i_{2}}$ is of size $k_{1}$, in total, the honest prover sends a $k_1 \times t$ matrix $Y$ of evaluations of $P_X$ (as shown in Figure \ref{fig:polyfold}), the verifier receives $Y'$. The verifier then checks that these values are consistent with $\vec{v}$ at each point $(i_{1},i_{2})$ in $J$ using the univariate low degree extension over $Q_{i_2}$. 

The \cite{RVW} protocol works as follows. For each $j\in [k_{1}]$, let $X_{j}$ be the $j^{\text{th}}$ row of $X$. We now have $t$ new conditions on any row $X_j$; the low degree extension of $X_{j}$ restricted to $J_{2}$ is the $j^{\text{th}}$ row of $Y'$, denoted by $Y'_j$. This corresponds to a new instance of $\PVAL(J_{2},Y'_j)$ for each $j \in [k_1]$, which is defined on the univariate low degree extension of $X_{j}$. 

If the previous check succeeds, the verifier sends a uniformly random vector $z\in \mathbb{F}^{k_{1}}$ to the prover and the new case of $\PVAL$ will be $\PVAL\bigg(J_{2}, z \cdot Y'\bigg)$ for which we want to test membership of $w = z\cdot X$.\footnote{For any $z \in \mathbb{F}^{k_1}$ and any matrix $A \in \mathbb{F}^{k_1 \times k_2}$, the dot product $z \cdot A \in \mathbb{F}^{k_2}$ is the linear combination of the rows of $A$ whose coefficients come from $z$.} Completeness of any such instance of $\PVAL$ follows from the linearity of polynomial interpolation. 

On the other hand, we generalise the soundness analysis in the following way. It is worth emphasising that the distance of a vector in $\mathbb{F}^{k_{2}}$ from $\PVAL(J_2,Y'_j)$ for any $j$, is taken along the marginal distribution of columns in $\mathcal{D}$. This is the distribution of $i_{2}$ returned from sampling $(i_{1}, i_{2})\sim \mathcal{D}$, so we can test against this distribution by sampling from $\mathcal{D}$.

At this point, to pass the verifier's checks and make it accept, the prover has to ``lie'' on a certain set of rows by pretending that the input is 
$X'\in \{0,1\}^{n}$ which satisfies $\PVAL$. We first look at the case that the prover lies in just one row and how a uniformly random $z$ will assist the verifier in catching the prover. We then extend this intuition to the case where the prover lies on any number of rows and show how repeating this process with random $z$ of varying Hamming weights will catch the prover.

Suppose first that the prover only lies about one row $i^{\ast}$. The distance of that row from satisfying $\PVAL(J_2,Y'_{i^\ast})$ is now $\varepsilon$ along one of $\mathcal{U}_{k_{2}}$ or $\mathcal{D}^{c}$ because of the original soundness condition. When the verifier picks a uniformly random $z$ from $\mathbb{F}^{k_{1}}$, with high probability $z_{i^{\ast}}$ is non-zero. For some $X'$ that belongs to $\PVAL(J,\vec{v})$, on every column that $X_{i^{\ast}}$ differs from $X'_{i^{\ast}}$, $w$ differs from $z \cdot X'$ (note that $X'$ is consistent with $X$ on all the other rows). Since $X_{i^{\ast}}$ is $\varepsilon$-far from $\PVAL({J_2,Y'_{i^\ast}})$ along $\mu^{c}$, the $\LDE$ of the corresponding $z \cdot X$ is far from satisfying $z\cdot Y'$  along $\mu^{c}$. This means that $w$ is far from this new folded instance of $\PVAL$.

For when the prover cheats on multiple rows, we prove that there is some $m^{\ast}\in [\log (k_{1})]$ such that by sampling a random set of $\frac{k_{1}}{2^{m^*}}$ rows, with high probability at least one of these rows will be $\Omega(\varepsilon\cdot k_{1}/\rho \cdot 2^{m^{\ast}})$-far from satisfying the corresponding row of $\PVAL$ along $\mu_{c}$ as $\sD$ is $\rho$-Dispersed. Since the verifier does not know the value of $m^{\ast}$, the verifier looks at each $m\in [\log (k_{1})]$ and uniformly samples a $z_{m}$ of Hamming weight $2^{m}$. Here, we use a lemma on distances between linear subspaces, which is a generalisation of an analogous lemma in \cite{RVW}. This lemma states that for our metric, if $S$ and $T$ are linear subspaces then a point in $S$ far from $T$ implies a uniformly random element of $S$ will be far from $T$ with high probability. This implies that $z_{m^{\ast}}\cdot X$ will be $(\varepsilon/\rho 2^{m^{\ast}})$-far from $\PVAL\bigg(J_{2}, z_{m^{\ast}}\cdot Y'\bigg)$ along $\mu^{c}$ with high probability. 

There are $\log (k_{1})$ different instances of $\PVAL$ and one of these is far from the corresponding $z_{m}\cdot X$. For each $m\in [\log (k_{1})]$ the prover sends the verifier $z_{m}\cdot X$  and the verifier checks if each of this is consistent with $\PVAL\bigg(J_{2}, z_{m}\cdot Y'\bigg)$. The next stage is for the verifier to check that each $w'_{m}$ that the prover purports to be $z_{m}\cdot X$ is close to the correct value. The verifier does this by sampling columns of $X$ along the $\mathcal{D}^{c}$ and the $\mathcal{U}_{k_{2}}$ distributions and computing the projection of $z_{m}\cdot X$ onto these samples, then querying the entire corresponding columns of $X$. If either consistency checks fail then the verifier rejects. Completeness follows immediately, but for soundness we have that $\mu_{\mathcal{D}^{c}, \mathcal{U}_{k_{2}}}(z_{m^{\ast}}\cdot X,\PVAL(J_{2}, z_{m^{\ast}}\cdot Y'))>\varepsilon/2^{m^{\ast}}\rho$ and therefore sampling $ z_{m}\cdot X$ will catch the cheating prover after $O(\frac{\rho}{2^{m^{\ast}}\varepsilon})$ samples from $z_{m}\cdot X$. Each query to $z_{m}\cdot X$ requires $2^{m}$ queries to $X$.

The total query complexity here is $\tilde{O}(\frac{\rho}{2^{m^{\ast}\varepsilon}})\cdot 2^{m^{\ast}} =\tilde{O}(\rho/\varepsilon)$ and the total sample complexity is $O(\rho/\varepsilon)$, which is a blowup of $\rho$ from the original uniform case in \cite{RVW}. This happens as the distance from $X$ to $X'$ that differ on a single element, (i,j) is $\mathcal{D}(i,j)$ originally but when we consider distance on a row vector, it becomes $\sum_{i'\in k_{1}}\mathcal{D}(i', j)$. In the uniform case this corresponds to multiplying by $k_{1}$, however, when the distribution is $\rho$-Dispersed, we multiply by $k_{1}/\rho$. The communication complexity is unchanged only sending $Y'$ and $\log(k_{1})$ different folded rows to total $O(|J|k_{1}+k_2 \log (k_{1}))$. For example, we can still achieve sublinear complexity even for $\rho=k_{1}$ for $k_{1}=n^{1/4}$, $k_{2}=n^{3/4}$, $\varepsilon=n^{-1/2}$ and $|J|=n\varepsilon\log (n)$. In this case the communication and query complexity are both $\tilde{O}(n^{3/4})$. In this case, we require $k_{1}\neq k_{2}$ as otherwise we do not have sublinear communication complexity and query complexity. In that case the query complexity would be $\tilde{O}(\sqrt{n}/\varepsilon)$ and the communication would be $\tilde{O}(n^{3/2}\varepsilon)$ where they can't both be sublinear.

\subsubsection{Polynomial Folding Proof}

For the following protocol, we take $X\in \F^{k\times k^{p}}$ for some $p\in \N$. For any $i\in [k]$, we define $X[i,\cdot]\in \F^{k^{p}}$ to be the $i^{\text{th}}$ row of $X$ such that 

\begin{equation*}
    \forall (i_{1},\cdots, i_{p})\in [k]^{p}:X[i,\cdot]_{(i_{1},\cdots, i_{p})}=X_{i,i_{1},\cdots, i_{p}}. 
\end{equation*}

For $i\in [t]$, $Y\in \F^{|J|\times k^{p}}$, we define $Y[i,\cdot]$ and $Y'[i,\cdot]$ similarly. Note that for $j\in [k]^{p}$, $Y'[\cdot,j]$ refers to the $j^{\text{th}}$ column of this object such that 
\begin{equation*}
    \forall i\in [k]:Y'[\cdot,j]_{i}=Y_{i,j}. 
\end{equation*}

$P_{X}$ is the $p+1$-variate $\LDE$ of $X$ to an $[\ell_{1}]\times [\ell_{1}]^{p}$ hypercube containing $X$ for sufficiently large $\ell_{1}=\poly(k)$. We sometimes identify this $p+1$-dimensional hypercube as a two dimensional $\ell_{1}\times \ell_{2}$ matrix for $\ell_{2}=\ell_{1}^{p}$). Additionally, we sometimes treat $X$ as a $k \times k_{2}$ matrix of elements of a field $\mathbb{F}$ for $k_{2}=k^{p}$, and $J$ as a subset of a larger $l_{1}\times l_{2}$ matrix and each $j \in J$ as $j = (j_{1}, j_{2}) \in [\ell_{1}] \times [\ell_{2}]$. We define $J_{2}$ to be the set of columns that contain elements of $J$ in other words
\begin{equation*}
    J_{2}= \{j_{2}\in \F^{k^{p}} : (j_{1}, j_{2}) \in J\}.
\end{equation*} 

\begin{protocol}
\caption{Polynomial Folding Protocol}
\label{TensRedAlg}
The protocol, $(P_{1}, V_{1})$ has explicit input $(\F,k,p,J,\Vec{v},\kappa)$, for soundness amplification parameter $\kappa>0$ and implicit input $X\in\F^{k\times k^{p}}$ that the prover has no access to. This protocol proceeds in two rounds:

\begin{enumerate}
    \item Prover sends Verifier: for each row $i \in [k]$ of $X$, send its encoding by $P_{X[i,\cdot]}$ (the $(p-1)$-variate $\LDE$ of the $i^{\text{th}}$ row of $X$) restricted to coordinates $J_{2}$. We call this matrix $Y\in\F^{k\times |J|}$.

    Verifier: receive $Y^\prime\in \F^{k\times \vert J\vert }$, reject if for some $(j_{1},j_{2})\in J$, the univariate low degree extension of the $j_{2}^{\text{th}}$ column of $Y'$ (i.e., $P_{Y'[\cdot,j_{2}]}$) on $j_{1}$ is not equal to the correct value in $\vec{v}$. In other words reject if 

    \begin{equation*}
        \exists (j_{1}, j_{2})\in J: P_{Y'[\cdot,j_{2}]}(j_{1})\neq \Vec{v}[j_{1}, j_{2}].
    \end{equation*}
    
    \item \label{VPrandvecpolyfold} Verifier sends Prover: for each $a \in [log(k/\kappa) + 1]$, send a uniformly random vector $\vec{z_{a}} \in \F^{k}$ of Hamming weight $2^{a}\kappa$.
\end{enumerate}
The output is $(log(k/\kappa) + 1)$ tuples $\{(a, \vec{z_{a}}, J_{2}, \vec{v_{a}}=\vec{z_{a}}\cdot  Y^\prime)\}_{a\in[log(k/\kappa)+1]}$.
\end{protocol}

\begin{theorem}
\label{TensRed}
For any $\kappa>0$, $\rho$-Dispersed distribution $\sD^{(p+1)} $ over $[k]\times [k^{p}]$, the polynomial folding protocol (Protocol \ref{TensRedAlg}), $(P_{1}, V_{1})$ on shared input  $(J, \vec{v})$ and prover input $X$ produces an output of $(log(k/\kappa) + 1)$ tuples $\{(a, \vec{z_{a}}, J_{2}, \vec{v_{a}}=\vec{z_{a}}\cdot  Y^\prime)\}_{a\in[\log(k/\kappa)+1]}$ and obeys the following conditions: 

\textbf{Completeness}: If $X$ satisfies $\PVAL(J,\vec{v})$ and we have an honest prover, the verifier does not reject and $\forall a\in [\log(k/\kappa)+1]$, $\Vec{z}_{a}\cdot X\in\PVAL(J_{2},\vec{v}_{a})$. 
\newline

\textbf{Bounded Locality}: $\forall a\in [\log(k/\kappa)]$, in the a-th output of the interactive protocol $\{(a, \vec{z_{a}}, J_{2}, \vec{z_{a}}\cdot  Y^\prime)\}_{a\in[log(k/\kappa)+1]}$, each coordinate of $\Vec{z}_{a}\cdot X$ is a linear combination of $\tau_{a}=2^{a}\cdot \kappa$ coordinates of $X$.
\newline

\textbf{Soundness}:  For $\sD^{(p+1)}\in \Delta([n])$, if $X$ is $\varepsilon$-far from $\PVAL(J, \vec{v})$ along $\mu_{\mathcal{D}^{(p+1)},\mathcal{U}}$, then for any cheating prover $P'$, with all but $((\vert \F\vert -1)^{-1} + e^{-\kappa/(4 \log (k))})$ probability over $V'$s coins, either $V$ rejects, or there exists some $a^{\ast}\in [\log(k/\kappa) + 1]$ s.t. 

\begin{equation*}\mu_{\mathcal{D}^{(p)}, \mathcal{U}_{k^{p}}}(\Vec{z}_{a^{\ast}}\cdot X,\PVAL(J_{2}, \vec{v}_{a^{\ast}}))>\frac{\varepsilon \cdot 2^{a^{\ast}}}{4\rho}.
\end{equation*}

In other words,
\begin{equation*}
    \mathbb{P}_{V_{0}}\left[\mu_{\mathcal{D}^{(p)}, \mathcal{U}_{k^{p}}}(\Vec{z}_{a^{\ast}}\cdot X,\PVAL(J_{2}, \vec{v}_{a}))>\frac{\varepsilon \cdot 2^{a^{\ast}}}{4\rho}\right]\geq 1-((\vert \F\vert -1)^{-1} + e^{-\kappa/(4 \log (k))})
\end{equation*}

This protocol has communication complexity $O(\vert J\vert  \cdot  k \cdot \log (k) \cdot \log \vert \F\vert )$ and one round of communication. The honest prover runs in time $\poly(k^{t}, \log \vert \F\vert )$, and the verifier runs in time $\poly(\vert J\vert, k, \log \vert \F\vert )$.  
\end{theorem}
\noindent Note that the verifier never accesses $X$ in this protocol. 

Now, $\mu_{\mathcal{D}^{(p)}, \mathcal{U}_{k^{p}}}$ is the new distance measure between the individual row vectors. This distance is equivalently obtained under the process of sampling from $\mathcal{D}^{(p+1)}$ and ignoring the last part of the index, and doing the same for $\mathcal{U}$.

Let $\PVAL_{i}=\PVAL(J_{2},Y'[i,\cdot])$ be the set of vectors $Z\in\mathbb{F}^{k^{p}}$, such that the low degree extension of $Z$ restricted to the points in $J_{2}$ are equal to the values in $Y'[i,\cdot]$. Further, for each $i \in [k]$, define $\varepsilon_{i}=\mu_{\mathcal{D}^{(p)}, \mathcal{U}_{k^{p}}}(X[i,\cdot],\PVAL_{i})$.

The soundness of Theorem \ref{TensRed} now relies on the following sequence of lemmas. To begin with, we state the ``distance preservation lemma" to prove that the distance from $X$ to $\PVAL(J,\Vec{v})$ is maintained for the sum of column marginal distances between $X[i,\cdot]$ and $\PVAL_i$ across all the $k$ rows.

\begin{lemma}
\label{epsilons}
\begin{equation*}
    \mu_{\mathcal{D}^{(p+1)}, \mathcal{U}_{n}}(X,\PVAL(J,\vec{v}))>\varepsilon\implies \sum_{i=1}^{k}\mu_{\mathcal{D}^{(p)}, \mathcal{U}_{k^{p}}}(X[i,\cdot],\PVAL(J_{2}, Y'[i,\cdot]))>\frac{k}{\rho}\varepsilon
\end{equation*}
\end{lemma}

\begin{proof}
We proceed by choosing $X'$ which allows us to relate the $\sD^{(p+1)}$-distance between $X$ and $\PVAL(J,\Vec{v})$, with the distances between the individual rows of $X$ to their corresponding lower-dimensional $\PVAL$ instances, but with respect to the marginal $\sD^{(p)}$. We then see that this sum is maximal over both the $\mathcal{D}^{(p+1)}$ and the uniform distribution. 

More precisely, for each $i\in[k]$, let $X'[i,\cdot]$ be the element of $\PVAL_{i}$ that minimises the hybrid distance, $\mu_{\mathcal{D}^{(p)},\mathcal{U}_{k^{p}}}(X[i,\cdot], X'[i,\cdot])$. We set $X'$ such that for each $i \in [k]$, it's $i^{\text{th}}$ row is $X'[i,\cdot]$. Put together with the fact that the verifier has not rejected at the end of step 1, we immediately observe that $X' \in \PVAL(J,\vec{v})$. Indeed, this holds as the univariate $\LDE$ of each column of $Y'$ restricted to $J_1$, gives us exactly $\Vec{v}$.

In what follows, we focus our calculations to distance over the $\sD^{(p+1)}$ distribution. The same calculations hold for distances over $\sU$ as well ($\sU$ is always $1$-Dispersed).
\begin{equation*}
    \begin{split}
        \sum_{i=1}^{k}\mu_{\mathcal{D}^{(p)},\mathcal{U}_{k^{p}}}(X[i,\cdot], \PVAL_{i})&=\sum_{i=1}^{k}\mu_{\mathcal{D}^{(p)}, \mathcal{U}_{k^{p}}}(X[i,\cdot], X'[i,\cdot]) \\
        &\geq \sum_{i=1}^{k}d_{\sD^{(p)}}(X[i,\cdot], X'[i,\cdot])\\
        &= \sum_{i=1}^{k} \underset{j'\sim \sD^{(p)}}{\mathbb{P}}\left[X_{ij'}\neq X'_{ij'}\right]\\
    \end{split}    
\end{equation*}
This follows as the first expression is the maximum over the distances with respect to the two distributions under consideration. Further,
\begin{equation*}
    \begin{split}
        \sum_{i=1}^{k} \underset{j'\sim \sD^{(p)}}{\mathbb{P}}\left[X_{ij'}\neq X'_{ij'}\right]&= \sum_{i=1}^{k}\sum_{l=1}^{k}\underset{(i',j')\sim \mathcal{D}^{(p+1)}}{\mathbb{P}}\left [X_{ij'}\neq X'_{ij'}\wedge (i'=l)\right]\\
        &\geq \frac{k}{\rho}\sum_{i=1}^{k} \underset{(i',j')\sim \sD^{(p+1)}}{\mathbb{P}}\left[X_{ij'}\neq X'_{ij'}\wedge i'=i\right]\\
        &=\frac{k}{\rho}\underset{(i',j')\sim \sD^{(p+1)}}{\mathbb{P}}\left[X_{i'j'}\neq X'_{i'j'}\right]\\
        &=\frac{k}{\rho}d_{\sD^{(p+1)}}(X,X')
    \end{split}
\end{equation*}
This follows from the definition of the marginal distribution and the definition of $\rho$-Dispersed distributions. We note that this inequality can be tight in certain cases.\footnote{For $\rho=k$, consider a distribution is only supported on one row $i_{0}$, if $i=i_{0}$ the following holds: $\underset{(i',j')\sim \ast}{\mathbb{P}}\left[X_{ij'}\neq X'_{ij'}\right]=\underset{(i',j')\sim \mathcal{D}^{(p+1)}}{\mathbb{P}}\left[X_{i_{0}j'}\neq X'_{i_{0}j'}\wedge (i'=i_{0})\right]$. In the case that $i\neq i_{0}$, both sides of this equation are $0$.} 

Therefore, as observed earlier, we have that $\sum_{i=1}^{k}\mu_{\mathcal{D}^{(p)},\mathcal{U}_{k^{p}}}(X[i,\cdot], \PVAL_{i})\geq \frac{k}{\rho}d_{\sD^{(p+1)}}(X,X')$, as well as $\sum_{i=1}^{k}\mu_{\mathcal{D}^{(p)},\mathcal{U}_{k^{p}}}(X[i,\cdot], \PVAL_{i})\geq \frac{k}{\rho}d_{\mathcal{U}_{n}}(X,X')$. Thus, 
\begin{equation*}
    \sum_{i=1}^{k}\mu_{\mathcal{D}^{(p)},\mathcal{U}_{k^{p}}}(X[i,\cdot], \PVAL_{i})\geq \frac{k}{\rho}\mu_{\mathcal{D}^{(p+1)},\mathcal{U}_{n}}(X,X')\geq \frac{k}{\rho}\mu_{\mathcal{D}^{(p+1)},\mathcal{U}_{n}}(X, \PVAL).
\end{equation*}
\end{proof}

We next have the following lemma on the distance between subspaces on arbitrary metrics that satisfy certain invariance constraints.
\begin{lemma}
\label{linSub}
Let $d : \mathbf{V} \times \mathbf{V} \rightarrow \mathbb{R}$ be a metric defined over a vector space $\mathbf{V}$ on a field $\mathbb{F}$, such that the following invariance conditions hold on $d$.
\begin{enumerate}
    \item For every $X,Y\in\mathbf{V}, a\in \mathbb{F}$, $d(aX,aY)=d(X,Y)$.
    \item For every $X,Y,Z\in\mathbf{V}$, $d(X+Z,Y+Z)=d(X,Y)$.
\end{enumerate}

Let $S$ and $T$ be two linear subspaces of $\mathbb{F}^{n}$ (for any finite field $\mathbb{F}$ and $n\in \mathbb{N}$). Suppose that there exists some point $s\in S$ such that $d(s,T)>\varepsilon$. Then, with all but $\frac{1}{\vert \mathbb{F}\vert -1}$ probability over the choice of a uniformly random point $r$ from $S$, $d(r,T)>\frac{\varepsilon}{2}$.
\end{lemma}

\begin{proof}
We sample a uniformly random vector in $S$ by first taking uniformly random $r\in S$ and then, if $r=s$ return $s$, and if not, take a uniformly random sample from the line between $r$ and $s$, excluding $s$.

If $r=s$, then the condition is fulfilled, if not then we look at the line along $r,s$. We claim that there can be at most one element along this line whose distance from $T$ is less than $\frac{\varepsilon}{2}$. If this is true, the theorem follows as the probability of sampling an element close to $T$ is less than $\frac{1}{|\mathbb{F}|-1}$. 

Suppose otherwise for contradiction and we have $r_{1}, r_{2}$ on the line and $t_{1}, t_{2}\in T$ such that $d(r_{1}, t_{1})< \frac{\varepsilon}{2}$ and $d(r_{2}, t_{2}) < \frac{\varepsilon}{2}$. As these points are on the same line, for some $a\in \mathbb{F}$, we have  $s=r_{1}+a(r_{2}-r_{1})$ and $t_{s}=t_{1}+a(t_{2}-t_{1})$. Therefore:
\[
\begin{split}
    d(s,t_{s})&=d(r_{1}+a(r_{2}-r_{1}),t_{1}+a(t_{2}-t_{1}))\\
    &= d((1-a)r_{1}+ar_{2},(1-a)t_{1}+at_{2})\\
    &\leq d((1-a)r_{1}+ar_{2},(1-a)t_{1}+ar_{2})+d((1-a)t_{1}+ar_{2},(1-a)t_{1}+at_{2})\\
    &\leq d((1-a)r_{1},(1-a)t_{1})+d(ar_{2},at_{2})\\
    &\leq d(r_{1},t_{1})+d(r_{2},t_{2})\\
    &< \frac{\varepsilon}{2}+\frac{\varepsilon}{2}\\
    &< \varepsilon
\end{split}
\]

This contradicts the initial assumption that $s$ is far from $T$, therefore this lemma follows by contradiction.
\end{proof}

Note that $\mu_{\mathcal{D}^{(p)},\mathcal{U}_{k^{p}}}$ as the maximum of two metrics is a metric. We next prove that we can apply Lemma \ref{linSub} with this distance measure, since it satisfies the constraints required.
\begin{lemma}
The metric $\mu_{\mathcal{D}^{(p)},\mathcal{U}_{k^{p}}}$ satisfies both the invariance properties in Lemma \ref{linSub}.
\end{lemma}

\begin{proof}
\begin{itemize}
    \item Invariance under vector addition:
\[
\begin{split}
    \mu_{\mathcal{D}^{(p)},\mathcal{U}_{k^{p}}}(X+Z,Y+Z)&=max(d_{\mathcal{D}^{(p)}}(X+Z,Y+Z),d_{U_{k^{p}}}(X+Z,Y+Z))\\
    &=max(d_{\mathcal{D}^{(p)}}(X,Y),d_{U_{k^{p}}}(X,Y))\\
    &=\mu_{\mathcal{D}^{(p)},\mathcal{U}_{k^{p}}}(X,Y)
\end{split}
\]
    \item Invariance under scalar multiplication:
\end{itemize}
\[
\begin{split}
    \mu_{\mathcal{D}^{(p)},\mathcal{U}_{k^{p}}}(aX,aY)&=max(d_{\mathcal{D}^{(p)}}(aX,aY),d_{U_{k^{p}}}(aX,aY))\\
    &=max(d_{\mathcal{D}^{(p)}}(X,Y),d_{U_{k^{p}}}(X,Y))\\
    &=\mu_{\mathcal{D}^{(p)},\mathcal{U}_{k^{p}}}(X,Y)
\end{split}
\]
\end{proof}

We finish by proving the following set of important claims. In Claim \ref{sumEps}, we prove that there is a set of rows $I\subseteq [k]$ for which for all $i\in I$, $X[i,\cdot]$ is sufficiently far from $\PVAL_{i}$, given the size of $I$. Then, in Claim \ref{ISize} we prove that with high probability $\exists a^* \in [\log(k/\kappa)+1]$ such that a randomly chosen $z_{a^*} \in \F^k$ of Hamming weight roughly $2^{a^*}$ has some row from $I$. Finally, in Claim \ref{TensRedSound}, we combine these with Lemma \ref{linSub} to show that sampling $\log (k/\kappa) + 1$ many random vectors $\{z_a\}_{a \in [\log (k/\kappa)+1]}$ in the polynomial folding protocol, results in at least one folded instance of $\PVAL$ which is sufficiently far from the correspondingly folded instance $z_a \cdot X$. 

\begin{claim}\label{sumEps}
If the verifier does not reject in Step 1, then there exists an integer $b \in \{0,\cdots, \log (k)\}$, and a subset $I \subseteq [k]$, s.t. $\forall i \in I, \varepsilon_{i} \geq k\varepsilon/(2^{b+1}\rho)$ and $\vert I\vert  \geq 2^{b}/4\log (k)$.
\end{claim}

\begin{claim}\label{ISize}
In Step \ref{VPrandvecpolyfold} of protocol \ref{TensRedAlg}, for $a \in [log(k/\kappa) + 1]$, let $I_{a}$ be the set of non-zero coordinates in $\vec{z}_{a}$ (this set is of size $2^{a}\cdot \kappa$). Take $b$ as guaranteed by Claim \ref{sumEps} and $a^{\ast} = min(log(k/\kappa)$, $\log (k) - b$). With all but $e^{-\kappa/4 \log (k)}$ probability over the verifier’s choice of $z_{a^{\ast}}$, there exists $i^{\ast} \in I_{a^{\ast}}$ s.t. $\varepsilon_{i^{\ast}} > \varepsilon \cdot 2^{a^{\ast}}/2\rho$.
\end{claim}

\begin{claim}\label{TensRedSound}
Take $a^{\ast}$ as guaranteed by Claim \ref{ISize}. With all but $((\vert \F\vert -1)^{-1}+e^{-\kappa/4 \log (k)} )$ probability over the verifier’s choice of $\vec{z}_{a^{\ast}}$, it holds for $\vec{v}_{a^{\ast}}=z_{a^{\ast}}\cdot Y'$ that 
\begin{equation*}
    \mu_{\mathcal{D}^{(p)},\mathcal{U}_{k^{p}}}(\vec{z}_{a^{\ast}}\cdot X,\PVAL(J_{2},\vec{v}_{a^{\ast}}))> \varepsilon\cdot2^{a^{\ast}}/4\rho.
\end{equation*}
\end{claim}

The proofs of the above claims are analogous to the uniform setting and can be found in Appendix \ref{AppClaimProofs}.

\begin{proof}[Proof of Theorem \ref{TensRed}]
Completeness follows as the only stage where the verifier can reject is in stage 1. There, the honest prover would send the true value of $Y$ which is consistent with $\vec{v}$ on $J$. This also results in a valid folded version of $\PVAL$, $\PVAL(J_{2}, \Vec{z_{a}}\cdot Y)$, for which $\forall a\in [\log(k/\kappa)+1]: \vec{z}_{a}\cdot X\in \PVAL(J_{2}, \Vec{z_{a}}\cdot Y)$, for any random linear combination $\Vec{z}_{a}$ picked by the verifier.   

Bounded locality follows from the fact that sampling an element of the folded vector will require queries to the input equal to the Hamming weight of the corresponding folding vector $\vec{z}_{a}$.    

The soundness follows directly from the Claim \ref{TensRedSound} whereby we have that with all but $((\vert \F\vert -1)^{-1} + e^{-\kappa/(4 \log (k))})$ probability over the verifier's randomness,  X will have $\mu_{\mathcal{D}^{(p)},\mathcal{U}_{k^{p}}}$ distance at least $\varepsilon\cdot2^{a^{\ast}}/4\rho$ from satisfying at least one of the new instances of $\PVAL$.
\end{proof}

\subsection{$\IPP$ for $\PVAL$ over $\rho$-dispersed distributions}
\label{sec:IPP}

Now that we have that polynomial folding lemma, we can develop the overall interactive protocol for distinguishing between being in $\PVAL$ and being far from $\PVAL$ along the $\mu_{\mathcal{D},\mathcal{U}_{n}}$ distance. The $\IPP$ is presented in Protocol \ref{FinIPP}.

Note that since $\sD^{(m-1)}$ is also a $\rho$-dispersed distribution by Lemma \ref{lem:fold_dispersed}, we can iterate Theorem \ref{TensRed} on $[k]^{m-1}$.

\begin{protocol}
\caption{$\IPP$ for $\PVAL$ over $\rho$-dispersed distributions}
\label{FinIPP}
The implicit input is $X \in \mathbb{F}^{k^{m}}$, the explicit input is $(\mathbb{F}, k, m, J, \vec{v})$ and a round parameter $r \in \mathbb{N}$.

Take $n = \vert X\vert  = k^{m}$ to be the input message length, and set a soundness amplification parameter $\kappa = 8 \log (r) \cdot \log (k)$.

\begin{enumerate}
    \item Set $\mathcal{W}_{0} \longleftarrow {(\lambda, \lambda, J, \vec{v})}$, where $\lambda$ is the empty string. For $s \in {1, . . . , r}$, $\mathcal{W}_{s}\longleftarrow \Phi$.
    \item\label{FinIPPPFSteps} Proceed in phases $s \longleftarrow 0, . . . , r - 1$:

    For each $((\vec{z}_{1}, . . . , \vec{z}_{s}),(a_{1}, . . . , a_{s}), J, \vec{v})$ in $\mathcal{W}_{s}$, in parallel, the prover and the verifier run Protocol \ref{TensRedAlg} (the Polynomial Folding Protocol) with $k =k$, $l_{1}=\poly(k)$, $p=m-s-1$, $\kappa=8\cdot \log (k) \log (r)$. Taking $X_{s} = \vec{z}_{s}\cdot (. . . \cdot(\vec{z}_{1} \cdot X))$, the instance is $(X_{s}, J, \vec{v})$.

    The output of each run is a collection of tuples $\{(a, \vec{z_{a}}, J_{2}, \vec{z_{a}}\cdot  Y^\prime)\}_{a\in[log(k/\kappa)+1]}$. For each $a\in[log(k/\kappa)+1]$, add $((\vec{z}_{1}, . . . , \vec{z}_{s}, \vec{z}_{s+1},a),(a_{1}, . . . , a_{s}, a), J_{s+1}, \vec{v}_{s+1},a)$ to $\mathcal{W}_{s+1}$.

    \item For each $((\vec{z}_{1}, . . . , \vec{z}_{r}),(a_{1}, . . . , a_{r}), J_{r}, \vec{v}_{r}) \in \mathcal{W}_{r}$, do the following in parallel:
    \begin{enumerate}
        \item \label{SentX} Prover sends Verifier: $X_{r} = \vec{z}_{r} \cdot (. . . \cdot (\vec{z}_{1} \cdot X))$ where $X_r \in \F^{k^{m-r}}$.
        \item Verifier: receive $X_{r}'$ and check if $P_{X_{r}'}|_{J_{r}} = \vec{v}_{r}$, else reject immediately.
        \item\label{MainIPPQueries} Verifier: set $\varepsilon_{r} = \varepsilon \cdot \prod_{s=1}^{r}\rho^{-1}(2^{a_{s}} /4)$. Pick $(10/\varepsilon_{r})$ uniformly random coordinates in $X'_{r}$ and then the same number of coordinates along the $\mathcal{D}$ distribution. For each coordinate $j$ that was picked, verify that $X'_{r}[j] = (\vec{z}_{r} . (. . . . (\vec{z}_{1} . X)))[j]$ by querying the appropriate coordinates in the original input message $X$. If any of these checks fail, then reject immediately.
    \end{enumerate}
    \item If the verifier did not reject so far then it accepts.
\end{enumerate}
\end{protocol}

\begin{theorem}
\label{ippthm}
For $n,k\in \N$, $ r\leq \log_{k}(n)$, $k^{r}\leq \frac{1}{\varepsilon}$, $r\leq k$, $r=\omega(1)$, $m=\log_{k}(n)$, a field $\F$ for which $\vert \F \vert=\polylog(n)$ such that $10r\leq |\F|\leq 1/\varepsilon$, $\sD$ a $\rho$-Dispersed distribution over $[k]^{m}$ and $(J,\vec{v})$ defining an instance of $\PVAL$, Protocol \ref{FinIPP}, $(P_0, V_0)$, satisfies the following properties:

\begin{enumerate}
    \item \textbf{Completeness}: If $X\in \PVAL(J,\Vec{v})$ then the verifier accepts with probability $1$. In other words,
    \begin{equation*}
        X\in \PVAL(J,\Vec{v})\implies \underset{V_{0},\mathcal{O}_{\mathcal{D}}(X)}{\mathbb{P}}\left[(P_{0}(X,\sD),V_{0}^{X,\mathcal{O}_{\sD}(X)})(n,\varepsilon) \text{ accepts }\right]=1.
    \end{equation*}
    \item \textbf{Soundness}: For $\sD\in \Delta([n])$, if $\mu_{\sD,\mathcal{U}}(X,\PVAL(J,\vec{v}))>\varepsilon$ then the verifier rejects with probability at least $\frac{1}{2}$.In other words,
    \begin{equation*}
        \mu_{\sD,\mathcal{U}}(X,\PVAL(J,\vec{v}))>\varepsilon\implies \underset{V_{0},\mathcal{O}_{\mathcal{D}}(X)}{\mathbb{P}}\left[(P^{\ast}_{0}(X,\sD),V_{0}^{X,\mathcal{O}_{\sD}(X)})(n,\varepsilon) \text{ rejects } \right]\geq \frac{1}{2}.
    \end{equation*}
\end{enumerate}

This protocol has query complexity $\rho^{r}(1/\varepsilon)^{1+o(1)}$, sample complexity $\rho^{r}(1/\varepsilon)^{1+o(1)}$, communication complexity $(n/k^{r}+|J|\cdot k)(1/\varepsilon)^{o(1)}$, and the number of messages is $(2r + 1)$ (round complexity is $r+1$). The honest prover runs in $\poly(n)$ time, and the verifier runs in $((\rho^{r}/\varepsilon) + n/k^{r} + \vert J\vert k)n^{o(1)}$ time.
\end{theorem}

\begin{proof}
 The honest prover runs in time $\poly(n)$, this follows from construction, as all the prover sends the verifier is a series of matrices $Y'$ and inner products of the various values of $z$ with $X$. The communication complexity will be the total communication from the polynomial folding protocol for each iteration of step \ref{FinIPPPFSteps} ($1\leq s<r$) along with the complexity of sending each value of $X_{r}$ in step \ref{SentX}.

\[
\begin{split}
    &\sum_{s=1}^{r}\log(k/\kappa+1)^{s}\cdot O(|J|\cdot k\cdot\log (k)\cdot \log \vert \F \vert)+\log(k/\kappa+1)^{r}\frac{n}{k^{r}}\log |\F|\\
    &=(\log(k/\kappa)+1)^{r}\cdot O(|J|\cdot k\cdot\log (k)\cdot \log \vert \F \vert)+\log(k/\kappa+1)^{r}\frac{n}{k^{r}}\log |\F|\\
    &=O(\log(k))^{r} \cdot \log \vert \F \vert \cdot O(n/k^{r}+|J|\cdot k\cdot \log (k))\\
    &=\frac{1}{\varepsilon^{o(1)}}(n/k^{r}+|J|\cdot k)\\
    &=(n/k^{r}+|J|\cdot k)\frac{1}{\varepsilon^{o(1)}}.
\end{split}
\]

\noindent This follows from the fact that $k^{r}=O\left(\frac{1}{\varepsilon}\right)$ and $|\F|\leq 1/\varepsilon$.

The completeness of this protocol follows as the only parts it can reject are the polynomial folding protocols which are perfectly complete and step 3b which says that $X_{r}'$ is consistent with $\vec{v}_{r}$. The latter holds since $X_{r}'=X_{r}$ is consistent with $\vec{v}_{r}$. 

Additionally, the verifier runs in time $n^{o(1)} \cdot ((\rho^{r}/\varepsilon) + n/k^{r} + k\vert J\vert)$, this follows from construction. The first part from sampling $X$, the second from processing the $X_{r}$ sent by the prover and the last part from processing the $Y'$ sent by the prover in the polynomial folding protocol. Note also that $\log |\F|$, the number of rounds and the number of folded instances of $\PVAL$ are encompassed by the $n^{o(1)}$ term.

The query complexity only has contributions from step \ref{MainIPPQueries}. At this step for each tuple $(a_{1},\cdots,a_{r})$, the verifier takes $10/\varepsilon_{r}=10/\varepsilon \cdot \prod_{i=1}^r \frac{4\rho}{2^{a_{s}}}$ samples and uniform queries to $X_{r}$. This results in total sample complexity

\begin{equation*}
    \sum_{(a_{1},\cdots, a_{r})\in[\log(k/\kappa)+1]^{r}}\frac{10}{\varepsilon}\prod_{s=1}^{r}4\rho/2^{a_{s}} \leq \rho^{r}/\varepsilon^{1+o(1)}.
\end{equation*}

Furthermore, the number of queries taken will be the number of samples and queries to $X_{r}$ times the number of queries from $X$ to obtain a query from $X_{r}$ which, by bounded locality of the polynomial folding protocol is equal to 
\begin{equation*}
    \tau_{r}=\prod_{s=1}^{r}2^{a_{s}}\kappa.
\end{equation*}

Therefore the total query complexity is
\[
\begin{split}
    \sum_{(a_{1},\cdots, a_{r})\in [\log(k/\kappa)+1]^{r}}\frac{10\tau_{r}}{\varepsilon_{r}} &=\sum_{(a_{1},\cdots, a_{r})\in [\log(k/\kappa)+1]^{r}}\frac{10}{\varepsilon}\left(\prod_{s=1}^{r}4\rho/2^{a_{s}}\right)\left(\prod_{s=1}^{r}2^{a_{s}}\kappa\right)\\
    &=\sum_{(a_{1},\cdots, a_{r})\in [\log(k/\kappa)+1]^{r}}\frac{10}{\varepsilon}\left(4\rho\kappa\right)^{r}\\
    &=(\log (k/\kappa) +1)^{r}\frac{10\cdot (4\rho\kappa)^{r}}{\varepsilon}\\
    &=\rho^{r}/\varepsilon^{1+o(1)}
\end{split}
\]

This follows since $\log^{r}(k)=1/\varepsilon^{o(1)}$, $4^{r}=1/\varepsilon^{o(1)}$, and $\kappa^{r}=(\log (k)\log (r))^{r}<\log ^{2r}(k)=1/\varepsilon^{o(1)}$ as $r<k$.

The total set of messages are those from $r$ iterations of the polynomial folding protocol along with $1$ message from step \ref{SentX}, in total this amounts to $2r+1$ messages and round complexity $r$.

Soundness follows as given the initial input, $X$ that is $\varepsilon$-far from $\PVAL$ along the $\mu_{\mathcal{D},\mathcal{U}_{n}}$ distance. By soundness of the polynomial folding protocol (Theorem \ref{TensRed}) and with a union bound over $r$ such rounds, with all but $r\cdot((\vert \mathbb{F}\vert -1)^{-1}+e^{-\kappa/4\log (k)})$ probability, there is some resulting tuple: $U_{r}^{\ast}=((\vec{z}_{1}^{\ast}, . . . , \vec{z_{r}}^{\ast}),(a_{1}^{\ast}, . . . , a_{r}^{\ast}), J_{r}^{\ast}, \vec{v}_{r}^{\ast})$, such that the instance $(X_{r}^{\ast}, J_{r}^{\ast}, \vec{v}_{r}^{\ast})$ specified by $U_{r}^{\ast}$ is $\varepsilon_{r}^{\ast}$-far from $\PVAL(J_{r}^{\ast},\Vec{v}_{r}^{\ast})$ on $P_{X_{r}'}$ along the $\mu_{\mathcal{D}^{(m-r)},\mathcal{U}_{n/k^{r}}}$ distance, where:

\begin{equation*}
    \varepsilon_{r}^{\ast} \geq \varepsilon\cdot \prod_{s=1}^{r} \cdot 2^{a_{s}^{\ast}}/4\rho.
\end{equation*}

As we will assume that the verifier does not reject in Step 3, the corresponding $X_{r}'^{\ast}$ has to satisfy this instance of $\PVAL$, therefore it must be far from $X_{r}^{\ast}$ along $\mu_{\mathcal{D}^{(m-r)},\mathcal{U}_{n/k^{r}}}$. In Step 3c, the verifier picks $10/\varepsilon_{r}^{\ast}$ uniformly random coordinates and then the same number of coordinates along the folded $\mathcal{D}$ distribution. As the $\mu_{\mathcal{D}^{(m-r)},\mathcal{U}_{n/k^{r}}}$-distance between $X_{r}^{\ast}$ and $X_{r}'^{\ast}$ is at least $\varepsilon^*_r$ far along one of these distributions then the verifier will reject with probability at least $9/10$.

A union bound over all the $r$-many polynomial foldings, ensures that the soundness error is at most:
\begin{equation*}
    r\cdot(1/(\vert \mathbb{F}\vert -1)+e^{-\kappa/4\log (k)})+1/10 \leq r\cdot(1/(10r-1)+e^{-2\log (r)})+1/10 <1/2
\end{equation*}
\end{proof}

\begin{remark}
The only step in Protocol \ref{FinIPP} that is altered from the proof of the $\IPP$ under uniform distribution (Theorem \ref{thm:unifRVW}) is in Step 3c, which involves sampling indices from the folded indices along $\sU$, as well as $\sD$. Thus, the number of queries increase by a factor of $\rho^{r}$ by the fact that we have $\rho$-Dispersed distributions, but the prover run-time and communication complexity remain unchanged.
\end{remark}

\noindent Finally, the $\IPP$ for languages computable by low-depth circuits over $\rho$-dispersed distributions is provided in Protocol \ref{pcl:NCIPP} and the proof of Theorem \ref{thm:dispersed_ipp_nc} is provided in Appendix \ref{sec:dispersed_ipp_nc}.

\begin{protocol}[h]
\caption{$\IPP$ over $\rho$-dispersed distributions, $(P,V)$, for a language $L$ with a circuit of size $S(n)$ and depth $\Delta_{L}(n)$.}
\label{pcl:NCIPP}
Let $(P_{\NC}, V_{\NC})$ be the interactive reduction from Theorem \ref{thm:df_nc_pval_reduction} and let $(P_{0}, V_{0})$ be Protocol \ref{FinIPP} from Theorem \ref{ippthm}.

The input is $X\in \{0,1\}^{n}$ and we set $k=\log(n)$, $m=\log_{k}(n)$, $r=\log(1/\varepsilon)/\log(k)$ and $|\F|=\poly(n)$.

\begin{enumerate}
    \item \label{NCpartIPP}$(P,V)$ run $(P_{\NC},V_{\NC})$ on $X$. The output of this protocol is $J \subset \F^{m}$ for which $|J|=4\varepsilon n\cdot\log (n)$ and $\Vec{v}\in \F^{t}$.
    \item  \label{RestIPP}$(P,V)$ run $(P_{0},V_{0})$ to verify membership of $X$, identified as an element in $\F ^{k^{m}}$, in $\PVAL(J,\Vec{v})$. $V$ rejects, if $V_{0}$ rejects.
\item $V$ accepts otherwise.
\end{enumerate}
\end{protocol}
\section{$\IPP$s over Product Distributions}
\label{sec:dfipp_learnable}
In Section \ref{sec:dfipp_low-depth}, we show the existence of distribution-free $\IPP$s for $\NC$ languages with query complexity $O\left(\frac{1}{\varepsilon}\right)$ and communication complexity $\tilde{O}\left(\varepsilon\cdot n+\frac{1}{\varepsilon}\right)$. Motivated by matching the communication complexity of $\varepsilon^{1-o(1)}\cdot n$ from \cite{RVW}, in Section \ref{sec:laconic_dfipp_special_dist} we construct $\IPP$s for $\rho$-dispersed distributions that achieve this communication complexity, while having query complexity $\frac{\rho^{\log (1/\varepsilon)/\log \log (n)}}{\varepsilon^{1+o(1)}}$. This, however, does not provide optimal query complexity for certain product distributions, eg., distributions that are concentrated on a small set of rows along some dimension, as they could be $\Omega(k)$-dispersed. 

As such, in this section, we construct $\IPP$s over product distributions that match the complexities of the uniform $\IPP$ from \cite{RVW}. In what follows, we prove Theorem \ref{thm:informal_product_dfipp} to show a white-box $\IPP$ for $\NC$ over any $m$-product distribution samplable using polynomial-sized circuits. The theme of this section is to use the sample oracle to learn the distribution $\sD$ and then use the learned distribution to design an $\IPP$ over a distribution family $\sF$, improving the query complexity of the $\IPP$ from Theorems \ref{thm:simplerncdfippsizedepth} and \ref{thm:dispersed_ipp_nc} which does not acquire any information about the distribution. In Appendix \ref{sec:eff_learn}, we show a framework for translating any interactive proof that can learn a distribution family $\sF$ using only samples, into a black-box $\IPP$ for $\NC$ over any distribution in $\sF$, that generalises the intuition for Theorem \ref{thm:informal_product_dfipp}. We start by defining the relevant family of distributions.

\begin{definition}[$m$-Product Distributions]
\label{def:prod_dist}
Let $\F$ be any field. For any $n \in \N$ and any integral function $m = m(n)$, let $k$ be an integer such that $n = k^m$. Then $\mathcal{D}=\{\sD_n\}$ is called an $m$-product distribution ensemble, where $\sD_n$ is a distribution over $[k]^m$ (by fixing some canonical bijection from $[k]^m$ to $[n]$), if there exists distributions $\sD_1, \sD_2, \dots, \sD_m \in \Delta([k])$, such that for any index $(i_1, \dots, i_m) \in [k]^m$, $\sD(i_1, \dots, i_m) = \prod_{r=1}^m \sD_r(i_r)$.
\end{definition}

To understand this better, consider a product distribution $\sD$ over $[n]$ where the probability of picking each index $i \in \{0,1\}^{\log(n)}$ is given by $\log(n)$ independent random variables over $\{0,1\}$, where the $j^{\text{th}}$-bit in the index is $1$ with probability $p_j$. Alternatively, we can view the input as a ($\log(n)$)-dimensional tensor, having 2 elements in each dimension (by fixing some enumeration over the cells into input indices). Now, $\prod_j p_j$ represents the probability of sampling a cell in the tensor. $m$-product distributions generalise this by considering product distributions over $m$-dimensional tensors, having $k$ elements in each dimension, such that $k^m = n$. 

Given this, we can define the oracle $\sO_\sD$ that provides labeled samples to the verifier. For example, any $2$-product distribution $\sD_n$ over $[\sqrt{n}] \times [\sqrt{n}]$ can be defined as a pair $\sD_n = \sD_1 \times \sD_2$, where $\sD_1$ and $\sD_2$ are distributions over $[\sqrt{n}]$. In such a case, the implicit input $X \in \{0,1\}^n$ can be viewed as a matrix in $\{0,1\}^{\sqrt{n} \times \sqrt{n}}$ (by fixing a bijection between the indices of $X$ and cells in the matrix), and $\sO_\sD$ can be viewed as an oracle that provides a sample $((i,j),X_{ij})$, where $(i,j)$ are the row and column indices of this matrix sampled from $\sD$. Similarly, for a $\left( \frac{\log (n)}{\log \log (n)} \right)$-product distribution $\sD$, we can view the $\sO_\sD$ as being defined over $\left( \frac{\log (n)}{\log \log (n)} \right)$-dimensional tensors with $\log (n)$ elements in each dimension. 

\vspace{0.1in}
We next state the following parallel variant of the Set Lower Bound protocol \cite{GS86} (observed in \cite{BT06}) that will be required in our $\IPP$.

\begin{lemma}[Corollary 2.7 of \cite{BT06}]
    \label{lem:lb_bt06}
    For any circuit $C:\{0,1\}^\ell \rightarrow\{0,1\}^{\log (k)}$, any $\tau \in (0,1)$, define the promise problem $\Pi = \{\Pi_{\ell,k}\}$
    \begin{equation*}
        \begin{split}
        \Pi^{Y}_{\ell,k} &:=\{(C,\tau,y_{1},p_{1},\cdots, y_{k},p_{k}):\forall i\in[k]:|C^{-1}(y_{i})|\geq p_i k\}\\
        \Pi^{N}_{\ell,k} &:=\{(C,\tau,y_{1},p_{1},\cdots, y_{k},p_{k}):\exists i\in[k]:|C^{-1}(y_{i})| \leq (1-\tau)p_{i}k\}
        \end{split}    
    \end{equation*}

    Then, for any $\delta > 0$, $\ell, k \in \N$, there exists a constant-round interactive proof for $\Pi_{\ell,k}$ with completeness probability $1-\delta$ and soundness probability $\delta$. The verifier runs in time $O \left( \frac{\poly(\vert C \vert) \cdot k^2}{\delta \tau^2} \right)$ and communication complexity of the interactive proof is $O \left( \frac{\ell \cdot k^2}{\delta \tau^2} \right)$, where $\vert C \vert$ is the size of the circuit $C$ in terms of its input size $\ell$. Moreover, the honest prover runs in time at most $2^\ell \cdot \poly(k,\vert C \vert)$.\footnote{The honest prover may have to go over all possible inputs to $C$ to find one which is mapped to $0^{\log (p_i k)}$ by a pairwise independent hash function in the Goldwasser-Sipser set lower bound protocol, for each $i \in [k]$. This implies the honest prover running time stated in Lemma \ref{lem:lb_bt06}.} In particular, if $C$ is a polynomial sized-circuit then $\ell = \polylog(k)$ and thus, $\vert C \vert$ is $\polylog(k)$ as well.
\end{lemma}

\noindent We next define the notion of concatenated languages.
\begin{definition}[Concatenated languages]
    \label{def:concat_lang}
     For any fixed $k, m$, define $g^{\mathsf{cat}} : \F^{k^m} \rightarrow \F^{[k+1] \times [k] \dots \times [k]}$, as the map that concatenates $0^{k^{m-1}}$ to the first dimension of a tensor in $[k]^m$. In other words, for any $X \in \F^{[k]^m}$, we have
        \begin{equation*}
        g^{\mathsf{cat}} (X) = 
        \begin{cases}
        X_{i_1, \dots i_m}, & \text{if } 1 \leq i_1, \dots, i_m \leq k, \\
        0, & \text{if } i_1 = k+1 \text{ and, } \forall 1 < \ell \leq m, 1 \leq i_\ell \leq k
        \end{cases}
    \end{equation*}
    Moreover, for any language $L$, we define $L_0 \subseteq \F^{[k+1] \times [k] \dots [k]}$ as $L_0 = \{ g^{\mathsf{cat}}(X) \mid X \in L \cap \F^{[k]^m} \}$.
\end{definition}

\vspace{0.05in}
For example, when $X \in \{0,1\}^{[\sqrt{n}]^2}$, $g^{\mathsf{cat}}(X)$ denotes the matrix obtained by concatenating the vector $0^{\sqrt{n}}$ as the last row. When $X \in \{0,1\}^n$, $g^{\mathsf{cat}}$ appends a $0$ to the end of $X$. In this case, for any language $L \in \{0,1\}^*$, $L_0$ is defined as $\{(X \circ 0) \mid X \in L_n \}$. Note that, one can easily generalise Definition \ref{def:concat_lang} to concatenating any $w \in \F^{k^{m-1}}$ along the $j^{\text{th}}$-dimension of $X$, for some $j \leq m$, but for the purposes of this section this definition suffices.

In Section \ref{sec:gran}, we show a reduction from testing for a language $L$ over a \textit{known} $m$-product distribution, to testing a closely related language $L'$ over a \textit{granular distribution}.

\subsection{Granularisation}
\label{sec:gran}

\begin{definition}[$m$-grained distributions]
    We say that a distribution $\mathcal{D}$ in $\Delta([n])$ is $m$-grained, for some $m \in \N$, if for every $i\in [n]$, there exists an $1 \leq a_{i} \leq m$ such that $\mathcal{D}(i)=\frac{a_{i}}{m}$. 
\end{definition}

\begin{definition}[$\sD$-extending a matrix]
    \label{def:extend_matrix}
    Let $\sD$ be a grained distribution over $[k]$ and let $B = \{b_1, \dots, b_k\}$ be it's granularities. Let $X$ be a matrix in $\F^{k \times k}$. We call another matrix $M \in \F^{(k+r) \times k}$, where $r = \sum_j (b_j-1)$, as a $\sD$-extension of $X$ if the following hold for any row $M_i$, $1 \leq i \leq k+r$.   
    \begin{equation}
        M_i = 
        \begin{cases}
            X[i,\cdot], & \text{ if } 1 \leq i \leq k \\
            X[1,\cdot], & \text{ if } k < i \leq k+b_1-1 \\
            X[j,\cdot], & \text{ for } j \in [k], \text{ such that } k + \sum_{\ell = 1}^{j-1} (b_\ell-1) < i \leq k + \sum_{\ell = 1}^j (b_\ell-1) \\
        \end{cases}
    \end{equation}

    This definition naturally generalises to $\sD$-extend a tensor $X \in \F^{k^m}$ to another tensor $M \in \F^{(k+r) \cdot k^{m-1}}$.
\end{definition}
Intuitively, the $\sD$-extension $M \in \F^{(k+r) \times t}$ is just the matrix $X$ with $b_1-1$ repetitions of row $X_1$ appended to its last row, and so on in sequential order, until $b_k-1$ repetitions of row $X_k$ are appended at the end. In the case where $m=1$, the $\sD$-extension just appends the relevant bits to the end of the string.

\begin{lemma}[The granularisation algorithm]
\label{lem:make_granular}
Let $L$ be any language and $\varepsilon > 0$. Then the following hold true.
\begin{itemize}
    \item Let $X \in \{0,1\}^n$ be an input to $L$. Let $\sD$ be any distribution over $[n]$ such that $\sD(i) = p_i$ for every $i \in [n]$.
    
    Then, there exists an algorithm $\sA_{\gran}$ that takes as input $\{p_1, \dots, p_n\}$ and outputs the granularities $\{a_1, \dots, a_{n+1}\}$ of an $8n$-grained distribution $\sD'$ over $[n+1]$ (i.e., for every $j \in [n+1]$, $\sD'(j) = a_j/8n$), that runs in time $O(n)$, such that:
    \begin{itemize}
        \item If $X \in L$, then for $X' = g^{\mathsf{cat}}(X)$, $X' \in L_0$.
        \item If $d_\sD (X,L) > \varepsilon$, then $d_{\sD'} (X',L_0) > \varepsilon/2$.
    \end{itemize}
    
    \item Let $X \in \{0,1\}^{k^m}$ be an input to $L$ (take $n = k^m$). Let $\sD = \sD_1 \times \dots \times \sD_m$ be an $m$-product distribution, where each $\sD_i \in \Delta([k])$. Let $\sD_1$ be described by the probability distribution vector $\{p_{11}, \dots, p_{1k}\}$. 
    
    Then, $\sA_{\gran}$ takes as input $\{p_{11}, \dots, p_{1k}\}$ and outputs $\{a_{11}, \dots, a_{1(k+1)}\}$ as the granularities of an $8k$-grained distribution $\sD_1'$ over $[k+1]$, running in time $O(k)$, such that:
    \begin{itemize}
        \item If $X \in L$, then for $X' = g^{\mathsf{cat}}(X)$, $X' \in L_0$.
        \item Let $\sD' = \sD'_1 \times \sD_2 \times \dots \times \sD_m$ be defined over $[k+1] \times [k] \times \dots \times [k]$. If $d_\sD (X,L) > \varepsilon$, then $d_{\sD'} (X',L_0) > \varepsilon/2$.
    \end{itemize}
\end{itemize}
\end{lemma}
While Lemma \ref{lem:make_granular} is inspired from \cite{Gol20}, their work focuses on the reduction from testing whether an \textit{unknown} input distribution (via samples) equals some fixed distribution $\sD$, to testing whether an unknown distribution equals the uniform distribution (via granular distributions). In our case, firstly the input distribution is known, and further, our focus is on the property testing setting where an implicit input string is provided. The proof of this Lemma is provided in Appendix \ref{sec:GranProof}.

\subsection{The White-box $\IPP$ for $\PVAL$}
\label{sec:final_wb_ipp_prod}
We start with a white-box $\IPP$ for the $\PVAL$ problem over polynomially-samplable $m$-product distributions.
\begin{theorem}
    \label{thm:wb_dfipp_prod_pval}
    For any $m,n \in \N$, let $\sF$ be a set of polynomially samplable $m$-product distributions over $[k]^m$, such that $n = k^m$. Let $\F$ be a field such that $|\F|=\polylog(k)$. Let $J \subset \F^m$ of size $t$ and $\Vec{v} \in \F^t$. Let $r \in \N$ be the round parameter such that $10 < r \leq \log (1/\varepsilon)/\log (k)$ and $|\F|>10r$. 

    Then, for every $\varepsilon > 0$ and $\sD \in \sF$, Protocol \ref{pcl:dfipp_whitebox_pval} is a white-box $\IPP$ for $\PVAL(\F,k,m,J,\Vec{v})$ over $\sF$ with proximity parameter $\varepsilon$, and completeness and soundness probabilities $2/3$, where the soundness promise is over the $\mu_{\sD,\sU}$ metric. 
    
    This $\IPP$ has query complexity $1/\varepsilon^{1+o(1)}$, communication complexity $\polylog(n) \cdot \left(k^{2}+\frac{k}{\varepsilon^{o(1)}}\cdot \eps \cdot n \right)$, and the verifier runs in $n^{o(1)}(\eps \cdot n \cdot k +k^{2}+ \frac{1}{\varepsilon})$ time. Moreover, the $\IPP$ has $O(r)$ many rounds and the honest prover runs in $2^{\polylog(n)}$ time.
\end{theorem}

\begin{proof}
Let $\sD$ be a fixed (but unknown) $m$-product distribution given as $\sD_1 \times \dots \times \sD_m$, where each $\sD_i$ is supported on $[k]$. Let $C : \{0,1\}^{\polylog (n)} \rightarrow \{0,1\}^{\log (n)}$ be the polynomial-sized circuit that samples $\sD$ (i.e., for every $i \in [n], \mathbb{P}_{x \sim U_{\polylog (n)}} \left[ C(x) = i \right] = \sD(i)$).

Let $X \in \{0,1\}^{k^m}$ be the implicit input to $\PVAL(J,\Vec{v})$. We view $X$ as an $m$-dimensional tensor with length $k$ in each dimension.\footnote{For simplicity, we fix a bijection between the indices of the string $X$ and the cells of a tensor in $\F^{k^m}$, e.g., in the lexicographic order of enumerating the cells of a tensor.} Similar to the setting in Section \ref{sec:dfipp_low-depth}, the input $X$ either has the promise that it belongs to $\PVAL(J,\Vec{v})$ or that $\mu_{\sD, \sU_n}(X,\PVAL(J,\Vec{v}) > \varepsilon$. 

\vspace{0.05in}
\textbf{Notation:}  For any $0 \leq \ell \leq r-1$ and every $j = (j_0, \dots, j_{m-1})$ in $J$, let $J_\ell \subset \F^{k^{m-\ell}}$ be the set of points specified by the last $(m-\ell)$ coordinates of any point in $J$. Following a similar process as Protocol \ref{FinIPP} (in Section \ref{sec:IPP}), we build a $(\log (8k/\kappa))+1)$-arity tree of depth $r$, before making the queries in the resulting folded instances in each leaf of this tree. 

Any node in layer $\ell$ has a label $(a_1, \dots, a_\ell)$ that specifies the tuple $(z_1, \dots, z_\ell, a_1, \dots, a_\ell, J_\ell, \Vec{v}_\ell)$ (note that $J_0 = J$ and $\Vec{v}_0 = \Vec{v}$). Here, each $z_i \in \F^{8k}, a_i \in [\log (8k/\kappa)+1]$, such that $\Hamming(z_i) = 2^{a_i}\kappa$, and $\Vec{v}_\ell$ is some vector in $\F^t$. $\mathcal{S}_\ell$ maintains the set of nodes in any layer $\ell$, where $\mathcal{S}_0 = \{\lambda, \lambda, J, \Vec{v}\}$ is just the original instance $(X,J,\Vec{v})$ and is written this way for technical reasons. Any tuple in $\mathcal{S}_\ell$ determines the corresponding folded instance $X_\ell$ after $\ell$ rounds and the corresponding $\PVAL$ claim of satisfying $v_\ell$ with respect to $J_{\ell}$.

It is worth defining the folded instance $X_\ell \in \F^{k^{m-\ell}}$ in more detail. Let $\mathcal{E}_i$ supported on $[k+1]$ be the granular approximation to $\sD_i$ obtained from $\sA_\gran$. We abuse the dot product notation, to define $X_\ell$ as the \textit{$\ell$-wise dot product} of $z_1, \dots, z_{\ell-1} \in \F^{8k}$ with $X$, as $z_\ell \cdot (\dots \cdot (z_1 \cdot X))$, where at each stage the dot product is in fact, computed between $z_i$ and the $\mathcal{E}_{i}$-extension of $g^{\mathsf{cat}} (X_{i-1})$ (with $X_0$ set to $X$). For eg., if $U_1 \in \F^{8k \times k^{m-1}}$ is the $\sD'_1$-extension of $g^{\mathsf{cat}}(X_0)$, then $z_1 \cdot X \in \F^{k^{m-1}}$ is in fact, the dot product, $z_1 \cdot U_1 = \sum_{j=1}^{8k} z_{1j} \cdot U_1[j,\cdot] \in \F^{k^{m-1}}$.

The extended polynomial protocol is similar to Protocol \ref{FinIPP}, except that it outputs tuples with respect to extended matrices. Using that we construct white-box $\IPP$ over $\sF$ in Protocol \ref{pcl:dfipp_whitebox_pval}. 
\begin{protocol}
\caption{Extended Polynomial Folding Protocol}
\label{pcl:epf_protocol}
\begin{algorithmic}{}
 \vspace{0.1in}
\STATE \textbf{Explicit Inputs:} The granularity set $B = \{b_1, \dots, b_{k+1}\}$ and $(\F, k, s, \hat{J}, \Vec{v},\kappa)$, where $\hat{J} \subset \F^s$ of size $t$ and $\vec{v} \in \F^t$.
\vspace{0.1in}
\STATE \textbf{Prover Input:} The prover input is $X \in \F^{k^{s}}$ (note that this is the implicit input to the verifier, but it is unused).
\vspace{0.1in}
\STATE Let $\hat{J} = (J_1, J_2)$, where $J_2 \subset \F^{s-1}$.
\begin{enumerate}
    \item For each $i \in [k]$, let $X[i,\cdot]$ be the $i^{\text{th}}$ row of $X$. The prover computes $P_{X[i,\cdot]}$ on every point in $J_2$. It sends these values in the matrix $Y \in \F^{k \times t}$.         
    \item The verifier receives $\Tilde{Y} \in \F^{k \times t}$, and rejects if there exists $j = (j_1,j_2) \in \Hat{J}$ where $j_2 \in \F^{s-1}$, such that the univariate $\LDE$ of degree at most $(k-1)$ of the $j_2^{\text{th}}$-column in $\Tilde{Y}$, $P_{\Tilde{Y}[\cdot, j_2]}$ evaluated on $j_1$ is not equal to $\Vec{v}[j]$.
            
    \item The verifier uses $B$ to extend the matrix $g^{\mathsf{cat}}(\Tilde{Y})$ into $U \in \F^{8k \times t}$.
            
    \item For each $a \in [\log (8k/\kappa)+1]$, the verifier samples a uniformly random vector $z_a \in \F^{8k}$ of Hamming weight $2^a\cdot\kappa$ and sends it to the prover. It then outputs $\log (8k/\kappa)+1$ such tuples $(z_a, a, J_2, z_a \cdot U))$ (the usual dot product).
    \end{enumerate}
\end{algorithmic}
\end{protocol}

\begin{protocol}[htp]
\caption{White-box $\IPP$ for $\PVAL$ over $m$-product distributions}
\label{pcl:dfipp_whitebox_pval}
\begin{algorithmic}{}
\vspace{0.1in}
\STATE \textbf{Implicit Input:} The implicit input is $X \in \F^{k^m}$, where $\vert X \vert = k^m = n$. The verifier can access $X$ using a query oracle. The verifier is also given a circuit $C: \{0,1\}^{\polylog(n)} \rightarrow \{0,1\}^{\log (n)}$, that samples a fixed (but unknown) $m$-product distribution $\sD = \sD_1 \times \dots \times \sD_m$ over $[k]^m$, which the verifier may use to simulate the sample oracle $\sO_\sD$ (or some other distribution), by querying $X$.
\vspace{0.1in}
\STATE \textbf{Explicit Inputs:} $(\F, k, m, J, \Vec{v})$, $\varepsilon > 0$ and a round parameter $r \leq \frac{\log(1/\varepsilon)}{\log (k)}$.
\vspace{0.1in}
\STATE \textbf{Prover Access:} The prover can access $X$ in its entirety and also has access to $C$.
\vspace{0.1in}
\STATE \textbf{Input Promise:} Either $X \in \PVAL(J,\Vec{v})$ or $\mu_{\sD,\sU} (X,\PVAL(J,\Vec{v}) > \varepsilon$.
\vspace{0.1in}
\STATE \textbf{The Protocol:}
\STATE Set $\mathcal{S}_0 = (\lambda, \lambda, J, \Vec{v})$. Further, set $B_0 = \emptyset, X_0 = X, J_0 = J$, and $\Vec{v}_0 = \Vec{v}$. 
\STATE For each round $0 \leq \ell \leq r-1$, do the following:
\begin{enumerate}
    \item \textbf{Learn $\sD_{\ell+1}$:} The prover sends the $(\ell+1)^{\text{th}}$ marginal distribution $\sP_{\ell+1} = \{\Tilde{p}_{(\ell+1),1}, \dots, \Tilde{p}_{(\ell+1),k}\}$. The verifier and the prover run the interactive proof from Lemma \ref{lem:lb_bt06} with inputs $C, \sP_{\ell+1}, \tau = 1/1000$ and $\delta = 1/20r$, such that the verifier rejects if the interactive proof rejects.
        
    \item \textbf{Granularise $\sP_{\ell+1}$:} The verifier then runs $\sA_\gran$ (Item 2 of Lemma \ref{lem:make_granular}) on input $\sP_{\ell+1}$ to get the granularities $B_{\ell+1} = \{b_{(\ell+1),1}, \dots, b_{(\ell+1),(k+1)}\}$ of an $8k$-grained distribution over $[k+1]$.
        
    \item \textbf{Extended Polynomial Folding:} For each tuple $((z_1, \dots, z_\ell),(a_1, \dots, a_\ell), J_\ell, \Vec{v}_\ell)$ in $\mathcal{S}_\ell$, in parallel, the prover and verifier run the extended polynomial folding protocol, Protocol \ref{pcl:epf_protocol}, with explicit inputs $(B_\ell, \F, k,m-\ell,J_\ell, \Vec{v}_\ell, \kappa=32\cdot\log (8k)\log (r))$. The implicit input instance is $X_\ell \in \F^{k^{m-\ell}}$, which is the $\ell$-wise dot product of $z_1, \dots, z_\ell \in \F^{8k}$ with $X_0$ (obtained from the tensor extensions given $B_0, \dots, B_\ell$). In essence, the underlying $\PVAL$ input instance is $(X_\ell, J_\ell, \Vec{v}_\ell)$.

    In return, Protocol \ref{pcl:epf_protocol} outputs a collection of tuples $\{(a,z^a_{\ell+1},J_{\ell+1},\Vec{v}^a_{\ell+1})\}$, one for each $a \in [\log (8k/\kappa)+1]$, where $z^a_{\ell+1} \in \F^{8k}$ is a uniformly random vector of weight $2^a \cdot \kappa$, and $(J_{\ell+1},\Vec{v}^a_{\ell+1})$ is the new claim corresponding to the folding obtained using $z^a_{\ell+1}$. Thus, for each $a$, the tuple $((z_1, \dots, z^a_{\ell+1}),(a_1, \dots, a_\ell, a), J_{\ell+1}, \Vec{v}^a_{\ell+1})$ is added to $\mathcal{S}_{\ell+1}$.
\end{enumerate}
    
\STATE \textbf{Verify the folded instances:} For each $((z_{1}, . . . , z_{r}),(a_{1}, . . . , a_{r}), J_{r}, \vec{v}_{r}) \in \mathcal{S}_{r}$, in parallel:
\begin{enumerate}
    \item The prover sends $X_{r}$, the $r$-wise dot product of $X$ with $z_1, \dots, z_r$. The verifier receives $\widetilde{X}_{r}$ and checks if $P_{\widetilde{X}_{r}} \vert_{J_r} = \vec{v}_{r}$, else it rejects immediately.
    \item The verifier sets $\varepsilon_{r} = \varepsilon \left( \prod_{s=1}^{r}(2^{a_{s}}/16) \right)$. It picks $(10/\varepsilon_{r})$ uniformly random coordinates in $\widetilde{X}_{r}$, as well as along $\widehat{\sD}_{r+1}$ (i.e., $\sD^{r+1} \times \dots \times \sD^m$) using the sampler $C$. 
        
    For each coordinate $j$ that was sampled, verify that $\Tilde{X}_{r}[j] = X_r[j]$ by using the sets $B_1, \dots, B_r$ and the vectors $z_1, \dots, z_r$ to compute the appropriate queries to the original input message $X$ (this can be done because any extension only works with rows of $X$ or the all $0$'s row). If any of these checks fail, then the verifier rejects immediately.
\end{enumerate}
        
\STATE If the verifier did not reject so far, it accepts.
\end{algorithmic}
\end{protocol}

\paragraph*{Completeness:}  In any invocation of the extended polynomial folding protocol, if the prover sends the matrix $Y$, the verifier always accepts. Moreover, by the definition of matrix extensions, we see that $X' \in \F^{8k \times k^{s-1}}$, which is the extension of $g^{\mathsf{cat}}(X)$ using $B$, also obeys similar consistency claims, i.e., for each $i \in [8k], P_{X'[i,\cdot]}$
evaluate on every point in $J_2$ is equal to $U[i,\cdot]$.\footnote{In particular, when $X'_i = 0^{k^{s-1}}$, we see that $P_{X'_i}$ is identically zero over $\F^{s-1}$.} This means that for each $a \in [\log (8k/\kappa)+1]$ and $z_a \in \F^{8k}$, by the linearity of computing the $\LDE$ on $z_a \cdot X'$, we see that $z_a \cdot X' \in \PVAL(J_2, z_a \cdot U)$.

Now, for the completeness of $\IPP$ itself, we first see that, since all the extended polynomial folding instances are YES instances with probability $1$, so are the ones in $\mathcal{S}_r$. Thus, the honest prover would send the correct $X_r$ and the verifier will not reject. In fact, the only place that the verifier may not accept is during the learning step, with probability at most $1/20r$; taking a union bound over the $r$ rounds, we see that Protocol \ref{pcl:dfipp_whitebox_pval} accepts with probability at least $19/20$.

\paragraph*{Soundness:} The proof of soundness follows a similar chain of steps as that of Protocol \ref{FinIPP}. Towards this end, suppose that $\mu_{\sD,\sU}(X,\PVAL) > \varepsilon$. We establish a new distance preservation lemma for analysing the extended polynomial folding protocol. The idea here is to use the distribution extended instances to provide a tighter analysis of the query complexity. 

For any $\ell$, such that $0 \leq \ell \leq r-1$, consider the $\ell^{\text{th}}$ round of the extended polynomial folding protocol. With high probability, from Lemma \ref{lem:lb_bt06}, we see that for each $i \in [k]$, $\sP_{\ell+1}(i) \geq (1-\tau) \cdot \sD_{\ell+1}(i)$. Moreover, let $\mathcal{E}_{\ell+1}$ be the granularised distribution on $[k+1]$ output by $\sA_\gran$ on input $\sP_{\ell+1}$. Furthermore, we define the following truncated product distributions, $\widehat{\sD}_{\ell+1} = \sD_{\ell+1} \times \dots \times \sD_m$ supported over $[k]^{m-\ell}$ and $\widehat{\sU}_{\ell+1}$ as the uniform distribution over $[k]^{m-\ell}$ (note that $\widehat{\sD}_1 = \sD$ and $\widehat{\sU}_1$ is the uniform distribution over $[k]^m$). 

Finally, define $X'_\ell \in \F^{8k \times k^{m-\ell-1}}$ as the $\mathcal{E}_{\ell+1}$-extension of $g^\mathsf{cat}(X_\ell)$ and $U_\ell \in \F^{8k \times t}$ as the $\mathcal{E}_{\ell+1}$-extension of $g^\mathsf{cat}(\widetilde{Y}_{\ell})$, where $\widetilde{Y}_{\ell}$ is the proof in the $\ell^\text{th}$ round of Protocol \ref{pcl:epf_protocol}.

\begin{lemma}[Distance preservation lemma for product distributions]
    \label{lem:dpl_product}
    For any $0 \leq \ell \leq r-1$ and for any $\gamma > 0$,
    \begin{equation*}
        \mu_{\widehat{\sD}_{\ell+1},\widehat{\sU}_{\ell+1}} (X_\ell, \PVAL(J_\ell,\Vec{v}_\ell)) > \gamma \implies 
        \sum_{i=1}^{8k} \mu_{\widehat{\sD}_{\ell+2}, \widehat{\sU}_{\ell+2}} (X'_{\ell,i}, \PVAL(J_{\ell+1},U_{\ell,i})) > 2k(1-\tau)\gamma
    \end{equation*}
\end{lemma}

\begin{proof}
    For ease of exposition, we prove the statement for the case where $\ell = 0$. The lemma statement follows from the same calculations for any $\ell \leq r-1$. 
    
    Recall that $X_0 = X, J_0 = J, \Vec{v}_0 = \Vec{v}$ and when $\ell=0$, $\gamma$ equals $\varepsilon$. Let $B = \{b_1, \dots, b_{k+1}\}$ be granularities of the distribution $\mathcal{E}_1$ returned by $\sA_\gran$ on input $\sP_1$. Define $X' \in \F^{8k \times k^{m-1}}$ and $U \in \F^{8k \times t}$ to be the $\mathcal{E}_1$-extensions of $g^{\mathsf{cat}}(X)$ and $g^{\mathsf{cat}}(\widetilde{Y})$ respectively. 
    
    Let $W_i \in \F^{k^{m-1}}$ be the instance that minimises the distance along $\mu_{\widehat{\sD}_{2}, \widehat{\sU}_{2}}$ between $X[i,\cdot]$ and $\PVAL(J_{1}, \widetilde{Y}[i,\cdot])$, and let $W$ be the instance in $\F^{k \times k^{m-1}}$ composed of these $W_i$'s as rows. Observe that $W \in \PVAL(J,\Vec{v})$ and thus, $\mu_{\sD, \sU}(X,W) > \varepsilon$. Let $W' \in \F^{8k \times k^{m-1}}$ be the $\mathcal{E}_1$-extension of $g^\mathsf{cat}(W)$. Since $X'$ and $W'$ are both extended using the same distribution $\mathcal{E}_1$, for each $i \in [8k]$, the copies of the closest instance $W[i,\cdot]$ in $W'$ correspond exactly to the copies of the row $X[i,\cdot]$ in $X'$. Thus, for each $i \in [8k]$, $W'[i,\cdot]$ is the closest instance along $\mu_{\widehat{\sD}_{2}, \widehat{\sU}_{2}}$ to $X'[i,\cdot]$. 
   
    We next study the distance between $X'$ and $W'$ along $\mu_{(\sU_{8k} \times \widehat{\sD}_{2}), (\sU_{8k} \times \widehat{\sU}_{2})}$, by firstly computing the distance over the $\widehat{\sD}_{1}$ distribution. 
    \begin{equation*}
        \begin{split}
            \mu_{(\sU_{8k} \times \widehat{\sD}_{2}), (\sU_{8k} \times \widehat{\sU}_{2})} (X', W') &\geq d_{\sU_{8k} \times \widehat{\sD}_{2}} (X', W') \\
            &= \sum_{u=1}^{8k} \underset{(i,j) \sim \sU_{8k} \times \widehat{\sD}_{2}}{\mathbb{P}} \left[ X'[i,j] \neq W'[i,j] \mid i = u \right] \cdot \underset{i \sim \sU_{8k}}{\mathbb{P}} [i=u]\\
             &= \sum_{u=1}^{k+1} \underset{(i,j) \sim \mathcal{E}_{1} \times \widehat{\sD}_{2}}{\mathbb{P}} \left[ g^{\mathsf{cat}}(X)[i,j] \neq g^{\mathsf{cat}}(W)[i,j] \mid i=u \right] \cdot \frac{b_u}{8k} \\
            &= \sum_{u=1}^{k+1} \underset{(i,j) \sim \mathcal{E}_{1} \times \widehat{\sD}_{2}}{\mathbb{P}} \left[ g^{\mathsf{cat}}(X)[i,j] \neq g^{\mathsf{cat}}(W)[i,j] \mid i=u \right] \cdot \underset{i \sim \mathcal{E}_{1}}{\mathbb{P}} [i=u]\\
            &= d_{\mathcal{E}_{1} \times \widehat{\sD}_{2}} (g^{\mathsf{cat}}(X),g^{\mathsf{cat}}(W))\\
        \end{split}
    \end{equation*}
    Here, the third expression holds from the definition of $\sD_{1}$-extension, where the $u^{\text{th}}$ row appears $b_u$ many times in the extension with probability $1/8k$ each. Thus, (abusing the index $u$) we see that the probability of sampling any row $u \in [k+1]$ of $g^{\mathsf{cat}}(X) \in \F^{(k+1)\cdot k^{m-1}}$ is $b_u/8k$, which is the same as $\mathcal{E}_{1}$.

    From Item 2 of Lemma \ref{lem:make_granular}, $\sA_\gran$ guarantees us that
    \begin{equation*}
        \begin{split}
            d_{\mathcal{E}_{1} \times \widehat{\sD}_{2}} (g^{\mathsf{cat}}(X),g^{\mathsf{cat}}(W)) 
            &> \frac{1}{2} d_{\sP_{1} \times \widehat{\sD}_{2}} (X, W) \\
            &= \frac{1}{2}\sum_{L=1}^k \underset{(i,j) \sim \sP_{1} \times \widehat{\sD}_{2}}{\mathbb{P}} [X[i,j] \neq W[i,j] \mid i = L] \cdot \sP_{1}(L)  \\
            &\geq \frac{(1-\tau)}{2}\sum_{L=1}^k \underset{(i,j) \sim \sD_{1} \times \widehat{\sD}_{2}}{\mathbb{P}} [X[i,j] \neq W[i,j] \mid i = L] \cdot \sD_{1}(L)  \\
            &= \frac{(1-\tau)}{2}d_{\widehat{\sD}_{1}}(X, W) \\
        \end{split}
    \end{equation*}
    The third expression comes from the guarantees provided by the parallel set lower bound protocol (with high probability) on $\sP_{1}$.

    We next look the same calculations for distances over $\sU_{8k} \times \widehat{\sU}_{2}$. It is worth noting from the proof of Lemma \ref{lem:make_granular}, that when the underlying distribution is $\sU_k$, $\sA_\gran$ outputs the distribution $\mathcal{E}_{1}$ over $[k+1]$ with granularities $b_1 = \dots = b_k = 8$ and $b_{k+1} = 0$ (the distribution is still uniform over $[k]$). In such a case, we maintain the caveat since the concatenated row $0^{k^{m-1}}$ has probability $0$ under $\sP_1$, it can be eliminated from the extension altogether and thus, the $\mathcal{E}_1$-extended instance is still over $\F^{8k \times k^{m-1}}$. Thus, by defining $X'$ and $W'$ in a similar fashion, 
    \begin{equation*}
        \begin{split}
            \mu_{(\sU_{8k} \times \widehat{\sD}_{2}), (\sU_{8k} \times \widehat{\sU}_{2})} (X', W') &\geq d_{\sU_{8k} \times \widehat{\sU}_{2}} (X', W') \\
             &= \sum_{u=1}^{8k} \underset{(i,j) \sim \sU_{8k} \times \widehat{\sU}_{2}}{\mathbb{P}} \left[ X'[i,j] \neq W'[i,j] \mid i = u \right] \cdot \underset{i \sim \sU_{8k}}{\mathbb{P}} [i=u]\\
            &= \sum_{u=1}^{k} \underset{(i,j) \sim \mathcal{E}_{1} \times \widehat{\sU}_{2}}{\mathbb{P}} \left[ X[i,j] \neq W[i,j] \mid i = u \right] \cdot \frac{b_{u}}{8k}\\
            &= \sum_{u=1}^{k} \underset{j \sim \widehat{\sU}_{2}}{\mathbb{P}} \left[ X[u,j] \neq W[u,j] \right] \cdot \frac{b_{u}}{8k}\\
            &\geq \sum_{u=1}^{k} \underset{j \sim \widehat{\sU}_{2}}{\mathbb{P}} \left[ X[u,j] \neq W[u,j] \right] \cdot \frac{2}{8k}\\
            &= \frac{1}{4}d_{\widehat{\sU}_{1}}(X, W)
        \end{split}
    \end{equation*}

    Put together, we have that 
    \begin{equation}
    \label{eq:gran_lb}
        \mu_{(\sU_{8k} \times \widehat{\sD}_{2}), (\sU_{8k} \times \widehat{\sU}_{2})} (X', W') > (1-\tau)\mu_{\widehat{\sD}_{1},\widehat{\sU}_{1}} (X, W)> \frac{(1-\tau)\gamma}{4}
    \end{equation}

    On the other hand, we have the following upper bound
    \begin{equation}
    \label{eq:gran_ub}
        \begin{split}
            \mu_{(\sU_{8k} \times \widehat{\sD}_{2}), (\sU_{8k} \times \widehat{\sU}_{2})} &(X', W') 
            \\
            &= \max \left\{ \left( \frac{1}{8k} \sum_{u=1}^{8k} \underset{j \sim \widehat{\sD}_{2}}{\mathbb{P}} \left[ X'[u,j] \neq W'[u,j] \right] \right), \left( \frac{1}{8k} \sum_{u=1}^{8k} \underset{j \sim \widehat{\sU}_{2}}{\mathbb{P}} \left[ X'[u,j] \neq W[u,j] \right] \right) \right\}\\
            &= \frac{1}{8k}\max \left\{  \sum_{u=1}^{8k} \underset{j \sim \widehat{\sD}_{2}}{\mathbb{P}} \left[ X'[u,j] \neq W'[u,j] \right], \sum_{u=1}^{8k} \underset{j \sim \widehat{\sU}_{2}}{\mathbb{P}} \left[ X'[u,j] \neq W'[u,j] \right] \right\} \\
            &\leq \frac{1}{8k} \sum_{u=1}^{8k} \max \left\{ \underset{j \sim \widehat{\sD}_{2}}{\mathbb{P}} \left[ X'[u,j] \neq W'[u,j] \right],\underset{j \sim \widehat{\sU}_{2}}{\mathbb{P}} \left[ X'[u,j] \neq W'[u,j] \right] \right\} \\
            &= \frac{1}{8k} \sum_{u=1}^{8k} \mu_{\widehat{\sD}_2, \widehat{\sU}_2} (X'_u, W'_u)
        \end{split}
    \end{equation}
    Here, the third expression follows from the fact that for any $c_1, \dots, c_k > 0$ and $d_1, \dots, d_k > 0$, we have that $\max \{\sum_i c_i, \sum_i d_i\} \leq \sum_i \max \{c_i,d_i\}$.
    
    Recall that $W'_i$ is the closest element in $\PVAL(J_{1}, U_{i})$ to $X'_i$, for each $i \in [8k]$. Combining this with equations \ref{eq:gran_lb} and \ref{eq:gran_ub}, we get the following expression, from which the lemma follows.
    \begin{equation*}
    \begin{split}
        \frac{1}{8k} \sum_{i=1}^{8k} \mu_{\widehat{\sD}_2, \widehat{\sU}_2} (X'_i, \PVAL(J_{1}, U_{i})) \geq \mu_{(\sU_{8k} \times \widehat{\sD}_{2}), (\sU_{8k} \times \widehat{\sU}_{2})} (X', W') > \frac{(1-\tau)\gamma}{4}
    \end{split}
    \end{equation*}
\end{proof}

Applying the distance preservation lemma from Lemma \ref{lem:dpl_product} in the proofs of Claims \ref{sumEps}, \ref{ISize}, and \ref{TensRedSound} as before, we have the following claim about the extended polynomial folding protocol.
\begin{claim}
\label{clm:good_folding_extended}
For any $0 \leq \ell \leq r-1$ and $\gamma > 0$, suppose that $\mu_{\widehat{\sD}_{\ell+1},\widehat{\sU}_{\ell+1}}(X_\ell,\PVAL(J_\ell,\vec{v}_\ell))>\gamma$. Let $X'_\ell$ be the $\mathcal{E}_{\ell+1}$-extension of $g^{\mathsf{cat}}(X_\ell)$ and let $U_\ell$ be the $\mathcal{E}_{\ell+1}$-extension of $\Tilde{Y}_\ell \in \F^{k \times t}$, which is the proof sent in the $(\ell+1)^{\text{th}}$ round of Protocol \ref{pcl:epf_protocol}. 

Then, there exists an $a^* \in [\log (8k/\kappa)+1]$, such that for a uniformly random $z_{a^*} \in \F^{8k}$ of Hamming weight $2^{a^*}\kappa$, with probability all but $(1/(\vert \F\vert -1)+e^{-\kappa/16 \log (k)})$ probability over the verifier’s choice of $\vec{z}_{a^{\ast}}$, it holds for $\vec{v}_{\ell+1}=z_{a^{\ast}}\cdot U_\ell$ that
\begin{equation*}
    \mu_{\widehat{\sD}_{\ell+2},\widehat{\sU}_{\ell+2}} (z_{a^*} \cdot X'_\ell,\PVAL(J_{\ell+1},\vec{v}_{\ell+1})) \geq \gamma \cdot2^{a^{\ast}}/16.
\end{equation*}
\end{claim}

Therefore, given that $\mu_{\sD,\sU}(X,\PVAL(J,\vec{v}))>\varepsilon$, then by the distance preservation lemma and taking a union bound over all rounds, we get that with probability all but $r \cdot (1/(\vert \F\vert -1)+e^{-\kappa/16 \log (k)})$, there exists tuple $((z_1,\dots,z_{r}), (a_{1},\cdots, a_{r}), J_{r},\vec{v}_{r}) \in \mathcal{S}_r$ and corresponding folded instance $X_{r}$, such that 
\begin{equation*}
    \mu_{\widehat{\sD}_{r+1},\widehat{\sU}_{r+1}}(X_r, \PVAL(J_r,\vec{v}_r)) > \varepsilon\cdot \prod_{s=1}^{r}k \cdot 2^{a_{s}}/16
\end{equation*}

Thus, the overall probability of accepting $X$ (by a union bound over the learning step, the extended polynomial folding step, and the folded instance verification step) for $r>10$ is at most
\begin{equation*}
    r\delta + r(1/(\vert \F\vert -1) + e^{-\kappa/16 \log (k)}) + (1-\varepsilon_{r})^{\frac{10}{\varepsilon_{r}}} < \frac{1}{20}+\frac{r}{10r-1}+\frac{r}{r^{2}}+ \frac{1}{10}\leq \frac{1}{3}
\end{equation*}
\vspace{0.1in}

\paragraph*{Query Complexity:} The input is queried only in the final verification stage. For each of the $\frac{10}{\varepsilon_{r}}$ indices sampled on each $X_{r}$ (over $\sU$ and $\sD$), the verifier makes $\prod_{s=1}^{r}2^{a_{s}}\kappa$ queries to $X$ by the \textit{bounded locality} of the extended polynomial folding protocol. There are at most $O(\log^{r} (k))=1/\varepsilon^{o(1)}$ many such instances, for every tuple in $\mathcal{S}_r$. By a similar calculation as Theorem \ref{ippthm}, the query complexity is given by 
\begin{equation*}
    \frac{1}{\varepsilon^{o(1)}}\cdot \frac{10}{\varepsilon_{r}}\cdot\prod_{s=1}^{r}2^{a_{s}}\kappa =\frac{1}{\varepsilon^{o(1)}}\cdot \frac{10}{\varepsilon}\cdot\prod_{s=1}^{r}\frac{16}{2^{a_{s}}}\cdot\prod_{s=1}^{r}2^{a_{s}}\kappa =\frac{1}{\varepsilon^{1+o(1)}}
\end{equation*}
\noindent This follows from the fact that $k^{r}=\frac{1}{\varepsilon}$, $O(\log (k))^{r}=\frac{1}{\varepsilon^{o(1)}}$ and $\kappa^{r}=O(\log (k) \log (r))^{r}=\frac{1}{\varepsilon^{o(1)}}$. 

\paragraph*{Communication Complexity:}
The communication complexity from $r$ iterations of the learning step is $O(rk)$ for sending the $r$ marginals $\sP_1, \dots, \sP_r$ and $O(\polylog (n)\cdot k^{2})$, for our setting of $\delta$ and $\tau^{2}$ in Lemma \ref{lem:lb_bt06}, and observing that the input length of $C$ is at most $\polylog(n)$ and the number of rounds $r$ is at most $\log (n)$. 

Next, from a similar analysis as Theorem \ref{ippthm}, we see that $r$ iterations of the extended polynomial folding have $r\cdot O(\log (k))^{r}\cdot kt\cdot \log\vert \F \vert=\frac{1}{\varepsilon^{o(1)}}\cdot k\cdot n\cdot \varepsilon\cdot \polylog(n)$ communication complexity, from sending all the matrices $Y$. This calculation follows from the fact that $r$, $O(\log (k))^{r}$ and $k=\frac{1}{\varepsilon^{o(1)}}$ along with that $t=n\varepsilon\log (n)$ and $|\F|=\polylog(k)$. Further, communicating the folded instances $X_r$ in the final step uses $\frac{n}{k^{r}\varepsilon^{o(1)}}$ many bits.

Therefore the total communication complexity is
\begin{equation*}
    \polylog(n)\cdot k^{2}+\frac{1}{\varepsilon^{o(1)}}\cdot n\cdot \varepsilon\polylog(n) + \frac{n}{k^{r}\cdot \varepsilon^{o(1)}}=\polylog (n) \cdot \left(k^{2}+\frac{k}{\varepsilon^{o(1)}}\cdot n\cdot\varepsilon\right).
\end{equation*}

\paragraph*{Number of Messages:} The number of messages (and the round complexity) is $O(r)$, since the learning phase is done in constantly many rounds from Lemma \ref{lem:lb_bt06} and the polynomial folding protocols are performed in parallel.

\paragraph*{Honest Prover Running Time:} 
The running time of the honest prover in the $\ell^{\text{th}}$ round is $O \left( 2^{\polylog(n)} \right)$ to compute the probability distribution vector of $\sD_{\ell+1}$ by going over all possible inputs to $C$ and $O \left( 2^{\polylog(n)} \right)$ in the IP from Lemma \ref{lem:lb_bt06}. As seen earlier, a $\poly(n)$-time honest prover is sufficient for the other stages in the protocol. In total, the honest prover runs in $O \left( 2^{\polylog(n)} \right)$ time.

\paragraph*{Verifier running time:} In total, the verifier runs in $O\left(\frac{r \cdot \poly(|C|)k^{2}}{\delta\tau^{2}}\right)= k^{2}\polylog (n)$ for learning, and $O(rk)=O(k\log (n))$ for granularisation. The rest follows from a similar analysis as Theorem \ref{ippthm}, and put together, the verifier running time is $n^{o(1)}(k\varepsilon n +k^{2}+ 1/\varepsilon)$.
\end{proof}

We end this section with the white-box $\IPP$ for low-depth circuits over polynomially-samplable $m$-product distribution families. This $\IPP$ builds on the framework from Section \ref{sec:laconic_dfipp_special_dist}.
\begin{theorem}
    \label{thm:dfipp_product_whitebox}
    For any $m,n \in \N$, let $\sF$ be a set of polynomially samplable $m$-product distributions over $[k]^m$, such that $n = k^m$. Moreover, for every $n \in \N$, let $L \subseteq \{0,1\}^{n}$ be a language computable by circuits of depth $\Delta(n)$ and size $S=S(n)$. 
    
    Then, for every large enough input $n \in \N$ and every $\varepsilon > 0$, there exists a white-box $\IPP$ for $L$ over $\sF$ with proximity parameter $\varepsilon$, and completeness and soundness probabilities $2/3$.
    
    The query complexity of the white-box $\IPP$ is $1/\varepsilon^{1+o(1)}$, the communication complexity is \sloppy{$\polylog(n) \cdot \left(k^{2}+k\cdot\varepsilon^{1-o(1)}\cdot n \right)+\varepsilon\cdot n\cdot \poly(\Delta)$}, and the verifier running time is $n^{o(1)}(k\cdot\varepsilon\cdot n +k^{2}+ 1/\varepsilon)+\varepsilon \cdot n\cdot\poly(\Delta)$. Moreover, the number of messages is $O(\log(1/\varepsilon)/\log(k))+\Delta \cdot \log(S))$ and the honest prover running time is $2^{\polylog(n)}+\poly(S)$.
\end{theorem}

The theorem follows by combining the interactive reduction from Theorem \ref{thm:df_nc_pval_reduction} (which doesn't access the implicit input), with the white-box $\IPP$ over $\sF$ for $\PVAL$ in Protocol \ref{pcl:dfipp_whitebox_pval}.

\section*{Acknowledgements}
We are grateful to Oded Goldreich for his insightful comments, some of which led to rephrasing Theorem \ref{thm:informal_dfipp_nc} to indicate the trade-off between the query and communication complexities with better clarity. We are also thankful to Marcel Dall'Agnol for many helpful discussions.

Tom Gur and Ninad Rajgopal are supported by the Tom Gur's UKRI Future Leaders Fellowship MR/S031545/1. Tom Gur is also supported in part by EPSRC New Horizons Grant EP/X018180/1 and EPSRC RoaRQ Grant EP/W032635/1. Ron Rothblum is funded by the European Union (ERC, FASTPROOF, 101041208). Views and opinions expressed are however those of the author(s) only and do not necessarily reflect those of the European Union or the European Research Council. Neither the European Union nor the granting authority can be held responsible for them.

\bibliographystyle{alpha}
\bibliography{thebibliography}
\appendix

\section{Proofs of Claims \ref{sumEps}, \ref{ISize}, and \ref{TensRedSound}}\label{AppClaimProofs}

\begingroup
\def\thetheorem{\ref{sumEps}}
\begin{claim}
If the verifier does not reject in Step 1, then there exists an integer $b \in \{0,\cdots, \log (k)\}$, and a subset $I \subseteq [k]$, s.t. $\forall i \in I, \varepsilon_{i} \geq k\varepsilon/(2^{b+1}\rho)$ and $\vert I\vert  \geq 2^{b}/4\log (k)$.
\end{claim}
\addtocounter{theorem}{-1}
\endgroup

\begin{proof}
We have by Lemma \ref{epsilons} that $\sum_{i\in [k]}\varepsilon_{i}>k\varepsilon/\rho$. We suppose by contradiction that for all $b\in \{0,\cdots ,\log (k)\}$, the number of rows $i$ s.t. $\varepsilon_{i}\in[k\varepsilon/(2^{b+1}\rho),k\varepsilon /(2^{b}\rho))$ is less than $2^{b}/4\log (k)$, it follows that:
\[
\begin{split}
    \sum_{i=1}^{k}\varepsilon_{i}&<\bigg(\sum_{b=0}^{\log (k) -1}(k\varepsilon/(2^{b}\rho))\cdot (2^{b}/4\log (k))\bigg) +(\varepsilon/2\rho)k\\
    &=(\log (k) +1)(k\varepsilon/4\rho\log (k)) + (k\varepsilon /2\rho)\\
    &<k\varepsilon/\rho.
\end{split}
\]

Here the left summand on the first line are the contributions from where $\varepsilon_{i}>\varepsilon/2\rho$, the right are the rest for which there can be at most $k$.
\end{proof}

\begingroup
\def\thetheorem{\ref{ISize}}
\begin{claim}
In Step \ref{VPrandvecpolyfold} of protocol \ref{TensRedAlg}, for $a \in [log(k/\kappa) + 1]$, let $I_{a}$ be the set of non-zero coordinates in $\vec{z}_{a}$ (this set is of size $2^{a}\cdot \kappa$). Take $b$ as guaranteed by Claim \ref{sumEps} and $a^{\ast} = min(log(k/\kappa)$, $\log (k) - b$). With all but $e^{-\kappa/4 \log (k)}$ probability over the verifier’s choice of $I_{a^{\ast}}$, there exists $i^{\ast} \in I_{a^{\ast}}$ s.t. $\varepsilon_{i^{\ast}} \geq \varepsilon \cdot 2^{a^{\ast}}/2\rho$.
\end{claim}
\addtocounter{theorem}{-1}
\endgroup
\begin{proof}
We know by Claim \ref{sumEps} that there is some $b\in\{0,\cdots, \log (k)\}$, and a set $I$ of at least $2^{b}/4 \log (k)$ rows each of which has $\varepsilon_{i}>k\varepsilon/(2^{b+1}\rho)$.

When we pick $min(k, \kappa k/2^{b})$ random rows to include in $I_{a^{\ast}}$, with all but $\left( 1-\frac{\vert I\vert }{k} \right)^{\kappa k/2^{b}}\leq e^{-\kappa/4\log (k)}$ probability, there will be non zero intersection between $I$ and $I_{a^{\ast}}$ (via a ``birthday Paradox" argument). The cardinality of $I_{a^{\ast}}$ is equal to $\kappa k/2^{b}$, but if this is greater than $k$, then setting it to $k$ suffices. The latter holds true, because the total number of rows is $k$, and in particular we end up picking every row in $I$. Now, we set $a^{\ast} = min(\log(k/\kappa)$, $\log (k) - b)$ (as $|I_{a^{\ast}}|=2^{a^{\ast}}\kappa$), to ensure that the size of our set $I_{a^{\ast}}$ is large enough but has size no greater than $k$.
\end{proof}
\begingroup
\def\thetheorem{\ref{TensRedSound}}
\begin{claim}
Take $a^{\ast}$ as guaranteed by Claim \ref{ISize}. With all but $((\vert \F\vert -1)^{-1}+e^{-\kappa/4 \log (k)} )$ probability over the verifier’s choice of $\vec{z}_{a^{\ast}}$ , it holds for $\vec{v}_{a^{\ast}}=z_{a^{\ast}}\cdot Y'$ that 
\begin{equation*}
    \mu_{\mathcal{D}^{(p)},\mathcal{U}_{k^{p}}}(X_{a^{\ast}},\PVAL(J_{2},\vec{v}_{a^{\ast}}))\geq \varepsilon\cdot2^{a^{\ast}}/4\rho.
\end{equation*}
\end{claim}
\addtocounter{theorem}{-1}
\endgroup
\begin{proof}
Let $T$ be the linear subspace of messages in $\mathbb{F}^{k^{p}}$, whose encodings are $0$ on $J_{2}$(the set of column coordinates of the elements of $J$ and therefore the coordinates of the row of $Y$). Also let $A_{i}$ be the set of vectors that when you add them to the $i^{\text{th}}$ row of $X$, they evaluate to the corresponding row of $y$. Observe that for any $\vec{s}\in A_{i}$: $A_{i}=\vec{s}+T$. Take any such vector $\vec{s}_{i}$ for each row $i$.

By the Claim \ref{ISize}, with all but $e^{-\kappa/4\log (k)}$ probability over the choice of non-zero coordinates of $I_{a^{\ast}}$, there is some $i^{\ast}\in I_{a^{\ast}}$ s.t. $\varepsilon_{i^{\ast}}\geq 2^{a^{\ast}}\varepsilon/2\rho$.

In this case, we give the non-zero values of $\vec{z}_{a^{\ast}}$(the values in $I_{a^{\ast}}$) uniformly random elements of $\mathbb{F}$. We now have a value for $X_{a^{\ast}}=\vec{z}_{a^{\ast}}\cdot X$ and a corresponding $\vec{v}_{a^{\ast}}=\vec{z}_{a^{\ast}}\cdot Y'$, we want to lower bound the $\mu_{\mathcal{D}^{(p)},\mathcal{U}_{k^{p}}}$ distance from this $X_{a^{\ast}}$ and any satisying $X'$.

Let $A$ be the set of shift vectors that when added to $X_{a^{\ast}}$ are consistent with $\vec{v}_{a^{\ast}}$. $A=\vec{s}+T$ as before for any shift vector $\vec{s}\in A$, the minimal $\mu_{\mathcal{D}^{(p)},\mathcal{U}_{k^{p}}}$ weight of $A$ (what we are now trying to minimise) is the same as the $\mu_{\mathcal{D}^{(p)},\mathcal{U}_{k^{p}}}$ distance from $\vec{s}$ to $T$.

$\vec{s}_{a^{\ast}}=\sum_{i\in [k]}\vec{z}_{a^{\ast}}[i]\vec{s}_{i}$ is a uniformly random vector in the linear span of $\{\vec{s}_{i}\}_{i\in I_{a^{\ast}}}$. There is some $i^{\ast}\in I_{a^{\ast}}$ s.t. $\varepsilon_{i^{\ast}}>\varepsilon2^{a^{\ast}}/2\rho$. Therefore by lemma \ref{linSub}, with a union bound we have that with probability at least all but $((\vert F\vert -1)^{-1}+e^{-\kappa/4 \log (k)} )$, we have for some $a^{\ast}$, $\mu_{\mathcal{D}^{(p)},\mathcal{U}_{k^{p}}}(X_{a^{\ast}},\PVAL(J_{2},\vec{v}_{a^{\ast}}))\geq \varepsilon\cdot2^{a^{\ast}}/4\rho$.
\end{proof}

\section{Proof of Theorem \ref{thm:dispersed_ipp_nc}}
\label{sec:dispersed_ipp_nc}
\begin{proof}[Proof of Theorem \ref{thm:dispersed_ipp_nc}]
The $\IPP$ for any such $L$ over $\rho$-dispersed distributions is specified in Protocol \ref{pcl:NCIPP}. This protocol has perfect completeness and soundness error $1/4$, to achieve the required soundness we simply repeat $O(1)$ times. Let $X \in \{0,1\}^n$ be the input to $L$. The properties and the complexity of the protocol are proved below.

\vspace{0.1in}
\textbf{Completeness}: Both the protocols used in this $\IPP$ have perfect completeness and therefore $X\in L$ implies that the verifier accepts with probability $1$.

In the protocol from Theorem \ref{thm:df_nc_pval_reduction}:
\begin{equation*}
    X\in L\implies \underset{V_{\NC}}{\mathbb{P}}[X\in \PVAL(J,\Vec{v})]=1.
\end{equation*}

In protocol \ref{FinIPP}:
\begin{equation*}
    X\in \PVAL(J,\Vec{v})\implies \underset{V_{0},\mathcal{O}_{\mathcal{D}}(X)}{\mathbb{P}}\left[(P_{0}(X,\sD),V_{0}^{X,\mathcal{O}_{\sD}(X)})(n,\varepsilon)=1\right]=1.
\end{equation*}
Therefore, in the overall $\IPP$, if $X\in L$, $V$ will accept with probability $1$.

\textbf{Soundness}: Assume that $d_{\mathcal{D}}(X,L)>\varepsilon$. For each repetition of this protocol, the probability that the verifier outputs $(J,\Vec{v})$ such that $\mu_{\mathcal{D},\mathcal{U}}(X,\PVAL(J,\Vec{v}))>\varepsilon$ is at least $\frac{1}{2}$ by the soundness condition of Theorem \ref{thm:df_nc_pval_reduction}. Similarly, by Theorem \ref{ippthm}, there is also probability at least $\frac{1}{2}$ that given $\mu_{\mathcal{D},\mathcal{U}}(X,\PVAL(J,\Vec{v}))>\varepsilon$, $V_{0}$ rejects. 

In the protocol from Theorem \ref{thm:df_nc_pval_reduction}:
\begin{equation*}
    d_{\sD}(X,L)>\varepsilon \implies \underset{{V_{\NC}}}{\mathbb{P}}[\mu_{\sD,\mathcal{U}}(X, \PVAL(J,\Vec{v}))>\varepsilon]\geq \frac{1}{2}.
\end{equation*}

In protocol \ref{FinIPP}:
\begin{equation*}
    \mu_{\sD,\mathcal{U}}(X,\PVAL(J,\vec{v}))>\varepsilon\implies \underset{V_{0},\mathcal{O}_{\mathcal{D}}(X)}{\mathbb{P}}\left[(P^{\ast}_{0}(X,\sD),V_{0}^{X,\mathcal{O}_{\sD}(X)})(n,\varepsilon)=0\right]\geq \frac{1}{2}.
\end{equation*}

Therefore $V$ rejects with probability $\frac{1}{4}$ each repetition of this protocol. In total this means that the verifier rejects with probability at least $\left(1-\frac{3}{4}\right)=\frac{1}{4}$. We can achieve our soundness condition by standard soundness amplification. 

The complexities of this protocol are achieved by summing the complexities of the component protocols. 

\vspace{0.1in}
\textbf{Communication Complexity:} The communication complexity from step \ref{NCpartIPP} is $\varepsilon n \cdot \poly(\Delta_{L})$, and is $(n\varepsilon+4\varepsilon n\log^{2} (n))\frac{1}{\varepsilon^{o(1)}}$ from step \ref{RestIPP}. In total, \sloppy{$c(n) =  (\varepsilon\cdot n+4\varepsilon n \log^{2}(n)) \frac{1}{\varepsilon}^{o(1)}+\varepsilon\cdot n\cdot \poly(\Delta_{L}) $}.

\vspace{0.1in}
\textbf{Query Complexity}: Note that the reduction to $\PVAL$ has no access to the input or the distribution so do not contribute to either the query or sample complexity. Therefore, we only need to consider step \ref{RestIPP} for which the query complexity is $\frac{\rho^{\log (1/\varepsilon)/\log\log (n)}}{\varepsilon^{1+o(1)}}$.

\vspace{0.1in}
\textbf{Sample Complexity}: Similarly, we only need consider step \ref{RestIPP} for which the sample complexity is $\frac{\rho^{\log (1/\varepsilon)/\log\log (n)}}{\varepsilon^{1+o(1)}}$.

\vspace{0.1in}
\textbf{Prover Running Time}: The prover running time from step \ref{NCpartIPP} is $\poly(S)$, and is $\poly(n)$ from step \ref{RestIPP}. In total,  the running time is $\poly(n,S)$.

\vspace{0.1in}
\textbf{Verifier Running Time}: The verifier running time from step \ref{NCpartIPP} is $\varepsilon\cdot n\cdot\poly(\Delta_{L})$, and is $\left( \frac{\rho^{\log (1/\varepsilon)/\log\log (n)}}{\varepsilon}+n\varepsilon + 4n\varepsilon \log^{2} (n) \right) n^{o(1)}$ from step \ref{RestIPP}. In total, the verifier runs in $n^{o(1)} \left( \frac{\rho^{\log (1/\varepsilon)/\log\log (n)}}{\varepsilon}+\varepsilon\cdot n\cdot\poly(\Delta_{L}) \right)$ time.

\vspace{0.1in}
\textbf{Round Complexity}: The round complexity from step \ref{NCpartIPP} is $O(\Delta_{L}\cdot\log (S))$, and is $\frac{\log(1/\varepsilon)}{\log\log (n)}+1$ from step \ref{RestIPP}. In total,  the round complexity is $O\left( \frac{\log(1/\varepsilon)}{\log\log (n)}+\Delta_{L}\cdot\log (S)\right)$.
\end{proof}

\section{Proof of the Granularisation Lemma}
\label{sec:GranProof}
\begin{proof}[Proof of Lemma \ref{lem:make_granular}]
We prove the first item of the lemma here and this extends analogously to the second item as well.

\begin{algorithm}[!h]
\caption{$\mathcal{A}_{gran}$: Algorithm for granularising an input distribution.}
The input is the distribution $\sD$ over $[n]$ presented as the list of values $\{p_{1},\cdots, p_{n}\}$.
\label{algo:grainer}
\begin{enumerate}
    \item Set $t=0$.
    \item For each $i\in [n]$ \label{grainrounds}
    \begin{enumerate}
        \item return $a_{i}=\lfloor 6n\cdot p(i)\rfloor + 2$.
        \item assign $t=t+a_{i}$.
    \end{enumerate}
    \item return $a_{n+1}=8n-t$.
\end{enumerate}
\end{algorithm}

We show that the set $\{a_{1},\cdots, a_{n+1}\}$ returned by $\mathcal{A}_{gran}$ form the granularities of a distribution $\sD'$ over $[n+1]$. For this, we first show that for any $\forall i\in [n+1]:\sD'(i)=\frac{a_{i}}{8n}\in[0,1]$ and then that $\sum_{i\in [n+1]}\sD'(i)=1$.

\begin{equation*}
    \forall i\in [n]: a_{i}=\lfloor 6n\cdot p(i)\rfloor + 2 \in \left[2,6n+2\right]\subseteq[0,8n].
\end{equation*}
Therefore $\forall i\in[n]: \sD'(i)=\frac{a_{i}}{8n}\in[0,1]$.

Let $q=\sum_{i\in [n]}\frac{\lfloor 6n p_{i}\rfloor+2}{8n}$. Observe that $q \in (0,1]$, since 
\begin{equation*}
    0 < \sum_{i\in [n]} \frac{\lfloor 6n p_{i}\rfloor+2}{8n} \leq \sum_{i\in [n]}\frac{ 6n p_{i}+2}{8n} \leq \frac{1}{4} + \sum_{i\in [n]} \frac{3}{4} p_{i} \leq 1
\end{equation*}
where the we use the fact that $\{p_1, \dots, p_n \}$ form a probability distribution. Moreover, observe that $\sD'(n+1) = 1-q \in [0,1)$. This implies that every element of the distribution $\sD'$ is within the range $[0,1]$ and therefore $\sD'$ is a distribution as 
\begin{equation*}
    \sum_{i\in[n+1]}\sD'(i)=\sum_{i\in[n+1]}\frac{a_{i}}{8n}=\sum_{i\in[n]}\frac{a_{i}}{8n}+a_{n+1}=\left(\sum_{i\in[n]}\frac{a_{i}}{8n}\right)+\frac{8n-\sum_{i\in[n]}a_{i}}{8n}=1.
\end{equation*}
The linear running time of $\sA_{\textsf{gran}}$ follows from inspection of Algorithm \ref{algo:grainer}.

Next, we see that if $X\in L$, by the definition of $L_{0}$, $X'=g^{\mathsf{cat}}(X)\in L_{0}$. On the other hand, to show that this transformation of $\sD$ maintains the distance for any input NO instance, we prove the following claim.
\begin{claim}
    $\forall i\in[n], \frac{a_{i}}{8n}\geq \frac{p_{i}}{2}$ 
\end{claim}
\begin{proof}
    Suppose $p(i)\leq \frac{1}{3n}$, then
    \[
        \frac{a_{i}}{8n}\geq \frac{2}{8n} \geq \frac{1}{3n}\cdot \frac{1}{2} \geq \frac{p_{i}}{2}
    \]
    Suppose instead that $p_{i}\geq \frac{1}{3n}$, let $p_{i}=\frac{r}{6n}+s$ for $r\in \N, s\in[0,\frac{1}{6n})$:
     \[
    \begin{split}
        \frac{a_{i}}{8n}&=\frac{r+2}{8n} \geq \frac{r}{12n}+ \frac{s}{2}\\
        &\geq \frac{p_{i}}{2}
    \end{split}
    \]
\end{proof}

Let $Y'\in L_{0}$, this implies that there is some $Y\in L$ such that $Y'=g^{cat}(Y)$. 
\[\begin{split}
    d_{\sD'}(X',Y')&= \sum_{i\in [n+1]}|X'_{i}-Y'_{i}|\sD'(i)\\
    &= \sum_{i\in [n]}|X'_{i}-Y'_{i}|\sD'(i)\\
    &\geq \sum_{i\in [n]}|X_{i}-Y_{i}|\frac{\sD(i)}{2}\\
    &\geq \frac{1}{2}d_{\sD}(X,Y) > \varepsilon/2.
\end{split}\]

This follows for every value of $Y'\in L_{0}$ including the minimiser, therefore $d_{\sD'}(X',L_{0})>\varepsilon/2$.
\end{proof}

\section{$\IPP$s for Efficiently Learnable Distribution Families}
\label{sec:eff_learn}
In this section, we demonstrate an $\IPP$ for $\NC$ over logspace-uniform low-depth circuit classes, given that the distribution we test against is efficiently learnable. 

For any pair of distributions $\sD_{1},\sD_{2}\in \Delta(\Omega_{n})$, define the total variation distance ($d_{TV}$) as follows 
\begin{equation*}
    d_{TV}(\sD_{1},\sD_{2})=\sum_{i\in [n]}\big\vert\mathbb{P}_{j\sim \sD_{1}}[j=i]-\mathbb{P}_{j\sim \sD_{2}}[j=i]\big\vert.
\end{equation*}

\begin{definition}[Efficiently Learnable Class of Distributions] 
\label{def:eff_learn_dist_ip}
Let $\mathcal{C}\subseteq \{\Delta(\Omega_n)\}_{n \in \mathbb{N}}$ be a class of distributions. We say that $\mathcal{C}$ is learnable by an interactive proof, if there exists $(P,V)$ where the verifier $V$ is given sample access to an unknown (but fixed) $\sD = \{\sD_n\} \in \mathcal{C}$ (i.e., $V$ gets samples from $[n]$ according to $\sD_n$) and the honest prover has full knowledge of $\sD$, along with common inputs $n,\varepsilon$, such that the proof system satisfies the following properties for every large enough $n$.
\begin{itemize}
   \item Completeness: For every $\sD \in \mathcal{C}$, the verifier outputs $\tilde{\sD} = \left( P(\mathcal{D}),V^{\sD} \right)(n,\varepsilon)$ as a distribution vector of probabilities or the value `reject' ($\bot$), such that
    \begin{equation*}
        \underset{V, \sD}{\mathbb{P}}\left[ (\tilde{\sD} \neq \bot) \land (d_{TV}(\mathcal{D},\hat{\mathcal{D}})<\varepsilon) \right] = 1.
    \end{equation*}
    
    \item Soundness: For every $\sD \in \mathcal{C}$ and for any computationally unbounded prover $P^{\ast}$, the output $\tilde{\sD} = \left( P(\mathcal{D}),V^{\sD} \right)(n,\varepsilon)$ is either a distribution vector of probabilities or the value `reject' ($\bot$), such that
    \begin{equation*}
        \underset{V, \sD}{\mathbb{P}} \left[ (\tilde{\sD} \neq \bot) \land (d_{TV}(\mathcal{D},\hat{\mathcal{D}}) \geq \varepsilon) \right]\ \leq 0.1.
    \end{equation*}
\end{itemize}
The sample complexity $s(n,\varepsilon)$, proof complexity $p(n,\varepsilon)$, verifier running time $t(n,\varepsilon)$ are as defined earlier. Furthermore, we specify the honest prover running time here to be $\poly(n)$.
\end{definition}

We next state the following lemma on the closeness of an input with respect to a distribution $\sD'$ that is close to the underlying distribution $\sD$, in total variation distance.
\begin{lemma}\label{tvineq}
For any strings $X,Y' \in\{0,1\}^{n}$ and distributions $\mathcal{D}, \mathcal{D}'$ over $[n]$,
\begin{equation*}
    d_{\mathcal{D}}(X,Y')\leq d_{TV}(\mathcal{D},\mathcal{D}')+d_{\mathcal{D}'}(X,Y').
\end{equation*}
\end{lemma}

\begin{proof} The proof consists of the following sequence of calculations.
\[
\begin{split}
    d_{\mathcal{D}}(X,Y') &= \underset{i\sim\mathcal{D}}{\mathbb{P}}\left [X_{i}\neq Y'_{i}\right]\\
    &= \sum_{i\in [n]\wedge X_{i}\neq Y'_{i}}\underset{j\sim\mathcal{D}}{\mathbb{P}}\left[j=i\right]\\
    &= \sum_{i\in [n]\wedge X_{i}\neq Y'_{i}}\underset{j\sim\mathcal{D}'}{\mathbb{P}}\left[j=i\right] + \left( \underset{j\sim\mathcal{D}}{\mathbb{P}}\left[j=i\right]-\underset{j\sim\mathcal{D}'}{\mathbb{P}}\left[j=i\right] \right)\\
    &\leq \sum_{i\in [n]\wedge X_{i}\neq Y'_{i}}\underset{j\sim\mathcal{D}'}{\mathbb{P}}\left[j=i\right] +\bigg\vert\underset{j\sim\mathcal{D}}{\mathbb{P}}\left[j=i\right]-\underset{j\sim\mathcal{D}'}{\mathbb{P}}\left[j=i\right]\bigg\vert\\
    &\leq \left( \sum_{i\in [n]\wedge X_{i}\neq Y'_{i}}\underset{j\sim\mathcal{D}'}{\mathbb{P}}\left[j=i\right]\right) +d_{TV}(\mathcal{D},\mathcal{D}')\\
    &= d_{\mathcal{D}'}(X,Y')+d_{TV}(\mathcal{D},\mathcal{D}')
\end{split}
\]
\end{proof}

\begin{corollary}\label{tvineqLang}
    For any string $X \in\{0,1\}^{n}$, language $L\subseteq\{0,1\}^{n}$ and distributions $\mathcal{D}, \mathcal{D}'$ over $[n]$,
\begin{equation*}
    d_{\mathcal{D}}(X,L)\leq d_{TV}(\mathcal{D},\mathcal{D}')+d_{\mathcal{D}'}(X,L).
\end{equation*}
\end{corollary}

\begin{proof}
    There exists $Y'\in L$ such that $d_{\mathcal{D}'}(X,L)=d_{\mathcal{D}'}(X,Y')$. Therefore, by Lemma \ref{tvineq}

    \begin{equation*}
    d_{\mathcal{D}}(X,L)\leq d_{\mathcal{D}}(X,Y')\leq d_{TV}(\mathcal{D},\mathcal{D}')+d_{\mathcal{D}'}(X,Y')=d_{TV}(\mathcal{D},\mathcal{D}')+d_{\mathcal{D}'}(X,L).
\end{equation*}
\end{proof}

The following lemma demonstrates a reduction from proving $\IPP$s over learnable distributions to uniform $\IPP$s given that we know the distribution. 

\begin{lemma}
\label{lem:make_uniform_oracle}
Let $L$ be any language computable by logspace-uniform circuits of size $S(n)$ and depth $D(n)$, and let $\varepsilon > 0$. Let $\sD$ be any distribution over $[n]$ such that $\sD(i) = p_i$ for every $i \in [n]$, where each $p_i \in [0,1]$.

Then, there exists an algorithm $\sB_\gran$ that given explicit inputs $\{p_1, \dots, p_n \}$, as well as oracle access to a string $X\in \{0,1\}^{n}$, outputs a vector $\vec{Q} \in \{0,1\}^{8n \log (n)}$, such that for a (parameterised) language $L'_{Q}$ computable by logspace-uniform circuits of size $S(n)+\tilde{O}(n)$ and depth $D(n)+O(\log (n))$, there exists $X'\in \{0,1\}^{8n}$ for which the following holds.
\begin{itemize}
        \item If $X \in L$, then $X' \in L'_Q$.
        \item If $d_\sD (X,L) > \varepsilon$, then $d_{U_{8n}} (X',L'_Q) > \varepsilon/2$.
    \end{itemize}

This algorithm runs in time $\tilde{O}(n)$. Additionally, any query to $X'$ can be implemented using a single query to $X$ and $O(1)$ running time, given explicit access to $\vec{Q}$.
\end{lemma}

In other words, Lemma \ref{lem:make_uniform_oracle} says that given explicit access to a distribution $\sD$, $\sB_\gran$ provides an ``implicit" reduction between  $L$ and a parameterised language $L'_Q$ computable by log-space uniform circuits of similar size and depth. By this we mean that $\sB_\gran$ reduces a testing problem for $L$ over $\sD$ to another testing problem for $L'$ over the uniform distribution, by simulating oracle access to the input $X' \in \{0,1\}^{8n}$ to $L'$ using the oracle to the original input $X \in \{0,1\}^n$.

\begin{proof}[Proof Sketch]
$\sB_\gran$ first runs $\sA_\gran$ from Lemma \ref{lem:make_granular} on input $\sD$ to obtain a set $A$ of granularities of a distribution $\sD'$ over $[n+1]$ for which the following hold.
\begin{equation*}
    X\in L\implies g^{\mathsf{cat}}(X)\in L_0
\end{equation*}
and 
\begin{equation*}
    d_{\sD}(X,L)>\varepsilon\implies d_{\sD'}(g^{\mathsf{cat}}(X),L_{0})>\varepsilon/2.
\end{equation*}

Given granularities $A = \{a_{1},\cdots, a_{n+1}\}$, in $\tilde{O}(n)$ running time we can obtain a vector $\vec{Q}$ defined as follows.
\begin{equation*}
    \Vec{Q}_{i}=\begin{cases}
            i, & i\in [n]\\
			1, & i\in[n+1, n+a_{1}-1]\\
            2, & i\in[n+a_{1},n+ a_{1}+a_{2}-2]\\
            \vdots\\
            n+1 & i\in [n+1+\sum_{j=1}^{n-1}a_{j}-1, 8n]
		 \end{cases}
\end{equation*}

Let $X'\in\{0,1\}^{8n}$ be the extension of $g^{cat}(X)$ using $A$. From this, each query $i$ to $X'$ becomes the query $\vec{Q}_{i}$ to $X$ as $X_{\Vec{Q}_{i}}=X'_{i}$ by the definition of extensions. Therefore in total this algorithm runs in time $\tilde{O}(n)$ and a query to the oracle for $X'$ makes a single query to $X$ and has $O(1)$ running time.

Define a parameterised language $L'_Q$ as the set of strings that are the extensions of $L_0$ using $A$ (i.e., $\sD'$-extensions of $L_0$). Formally,
\begin{equation*}
    L'_Q = \left \{W \in \{0,1\}^{8n} \mid \exists Y \in L_0 \cap \{0,1\}^{n+1} \text{ such that } W \text{ is the extension of $Y$ using $A$} \right \}
\end{equation*}

Let $X'$ be the $\sD'$-extension of $g^{\mathsf{cat}}(X)$. From this, Item 1 follows as 
\begin{equation*}
    X\in L\implies g^{\mathsf{cat}}(X)\in L_{0}\implies X'\in L'_Q
\end{equation*}

On the other hand, for any $Y,\widetilde{Y} \in \{0,1\}^{n}$, with $Y'$ and $\widetilde{Y}'$ being the $\sD'$-extensions of $g^{\mathsf{cat}}(Y)$  and $g^{\mathsf{cat}}(\widetilde{Y})$ respectively, we have 
\[
\begin{split}
    d_{\sD'}(g^{\mathsf{cat}}(Y),g^{\mathsf{cat}}(\widetilde{Y}))&=\sum_{i=0}^{n+1}\frac{a_{i}}{8n}|g^{\mathsf{cat}}(Y)_{i}-g^{\mathsf{cat}}(\widetilde{Y})_{i}|\\
    &=\sum_{i=0}^{8n}\frac{1}{8n}|Y'_{i}-\widetilde{Y}'_{i}|\\
    &=d_{\sU_{8n}}(Y',\widetilde{Y}')
\end{split}
\]

\noindent Therefore,
\begin{equation*}
    d_{\sD}(X,L)>\varepsilon\implies d_{\sD'}(g^{\mathsf{cat}}(X),L_{0})=d_{\sU_{8n}}(X',L'_Q)>\varepsilon/2.
\end{equation*}

The last thing we need to prove is that $L'_Q$ is computable by log-space uniform circuits of size $S(n)+8n$ and depth $D(n)+\log(8n)$. 
Thus, a log-space Turing machine can first generate a circuit for $L$ on the first $n$ indices and then generate an $O(\log (n))$-depth circuit of $\mathsf{AND}$ gates at the top to check $\bigwedge_{j\in[n+1,8n]}X_{\vec{Q}_{j}}=X_{j}$. This is performed using additional circuitry of size $O(n\log (n))$ ($n$ AND statements each requiring $\log (n)$ bitwise comparisons) and depth $O(\log (n))$. 
\end{proof}

With this lemma, we can extend our framework from Theorem \ref{thm:dispersed_ipp_nc} to construct $\IPP$s for languages computable by low-depth circuits over the class of efficiently learnable distributions (using interactive proofs) in the sense of Definition \ref{def:eff_learn_dist_ip}.
\begin{theorem}
\label{thm:dfIPPLearnable}
    Let $\mathcal{C}\subseteq \{\Delta(\Omega_n)\}_{n \in \mathbb{N}}$ be a class of distributions that is learnable using an interactive proof that has sample complexity $s(n,\varepsilon)$, proof complexity $p(n,\varepsilon)$, verifier running time $T_V(n,\varepsilon)$ and honest prover running time $\poly(n)$. Moreover, for every $n \in \N$, let $L \subseteq \{0,1\}^{n}$ be a language computable by circuits of depth $\Delta(n)$ and size $S=S(n)$.

    Then, for every large enough input length $n \in \N$ and every $\varepsilon > 0$, there exists an $\IPP$ for $L$ over $\mathcal{C}$ with proximity parameter $\varepsilon$, and perfect completeness and soundness error $2/3$.
    
    This $\IPP$ has query complexity $O(1/\varepsilon)$, sample complexity $s(n,\varepsilon/4)$ and communication complexity $p(n,\varepsilon/4)+\Tilde{O}(\varepsilon n)+\varepsilon\cdot n\cdot \poly(\Delta,\log(n)))$. In addition, the honest prover runs in $\poly(S,n)$ time and the verifier runs in $t(n,\varepsilon/2)+O(1/\varepsilon)+\varepsilon\cdot n\cdot \poly(\Delta,\log(n))+\tilde{O}(n)$.
\end{theorem}

\begin{proof}[Proof Sketch]

This protocol proceeds by first applying the learning algorithm for $\mathcal{C}$ with proximity parameter $\varepsilon/2$  to learn a distribution $\tilde{\sD}$ which is $\varepsilon/2$-close to $\sD$ (it rejects if the learner outputs $\bot$). Next, it runs the granularisation algorithm $\sB_\gran$ from Theorem \ref{lem:make_uniform_oracle} to reduce testing $X$ along the distribution $\tilde{\sD}$ for membership of $L$ to testing $X'$ along the uniform distribution for membership of $L_Q'$, for which we can simulate oracle access to $X'$ using $X$ and the output of $\sB_\gran$, $\vec{Q}$. Finally, it runs the $\IPP$ from Theorem \ref{thm:uniform_ipp_pval_rvw} for testing $X'$'s membership in $L_Q'$ with proximity parameter $\varepsilon/4$.

Perfect completeness follows as in the learning stage the verifier learns a valid distribution $\tilde{\sD}$ from the honest prover, the completeness of $\sB_\gran$ and of Theorem \ref{thm:unifRVW} ensures that if $X\in L$ then $X'\in L_Q'$ and therefore that the $\IPP$ from Theorem \ref{thm:unifRVW} accepts with probability $1$.

For soundness, we show that the transformation to verifying $L'_Q$ preserves distance.
\[
\begin{split}
    d_{\sD}(X,L)>\varepsilon &\implies d_{\Tilde{\sD}}(X,L)+d_{TV}(\sD, \Tilde{\sD})>\varepsilon\\
    &\implies d_{\Tilde{\sD}}(X,L)+\frac{\varepsilon}{2}>\varepsilon\\
    &\implies d_{\Tilde{\sD}}(X,L)>\frac{\varepsilon}{2}\\
    &\implies d_{\sU_{8n}}(X',L_Q')>\frac{\varepsilon}{4}\\
\end{split}
\]

The first expression comes from Corollary \ref{tvineqLang}, the second comes from the guarantees on $\Tilde{\sD}$ of the distribution learner, and the fourth from the distance preservation from Lemma \ref{lem:make_uniform_oracle}. This means that with probability at least $0.9$, the verifier either rejects (because the learner outputs $\bot$) or $\Tilde{\sD}$ is a good approximation of $\sD$, and with probability $2/3$, the $\IPP$ will reject. Put together, the verifier rejects the input $X$ with probability at least $1/2$.

The query complexity follows from the fact that each query to $X'$ is simulated by $1$ query to $X$ and from the query complexity from Theorem \ref{thm:uniform_ipp_pval_rvw}. The verifier running time follows as these queries also cost an additional $\tilde{O}(n)$ to generate (via $\sB_\gran$) and the additional running time from the new oracle can only introduce at most a multiplicative factor of $O(1)$. The other complexities follow from construction.
\end{proof}

\noindent From Theorem \ref{thm:dfIPPLearnable}, we get the following $\IPP$ over learnable distributions for $\NC$ languages that matches the parameters of the uniform $\IPP$ on every $\varepsilon$.
\begin{corollary}
\label{cor:learnable_nc_improvement}
Let $\sF$ be a class of distributions that can be learnt by an interactive protocol $\sA$ that takes sample access to some $\sD$ in $\sF$, using $O(1/\varepsilon)$ samples and $\Tilde{O}(\varepsilon n)$ communication complexity, where the verifier runs in time $T_V(n)$ and the honest prover runs in $\poly(n)$ time. 
     
Then, there exists an $\IPP$ for $\NC$ over $\sF$, with sample complexity $O(1/\varepsilon)$, query complexity $O(1/\varepsilon)$ and communication complexity $\Tilde{O}(\varepsilon n)$. Moreover, the honest prover runs in time $\poly(n)$.
\end{corollary}
\end{document}